\documentclass[11pt]{article}
\usepackage[utf8]{inputenc}

\usepackage{amsmath,amssymb,amsfonts,amsthm}
\usepackage{bm}
\usepackage{colonequals}
\usepackage{comment}
\usepackage{epsfig} \usepackage{latexsym,nicefrac,bbm}
\usepackage{xspace}
\usepackage{color,fancybox,graphicx,subfigure,fullpage}
\usepackage[top=1.25in, bottom=1.25in, left=1in, right=1in]{geometry}
\usepackage{tabularx}
\usepackage{hyperref}
\usepackage{cleveref}
\usepackage{pdfsync}
\usepackage[boxruled,linesnumbered]{algorithm2e}
\usepackage{multicol}
\usepackage{enumitem}
\usepackage{float}
\usepackage{tikz}

\newtheorem{theorem}{Theorem}[section]

\newtheorem{conjecture}[theorem]{Conjecture}
\newtheorem{definition}[theorem]{Definition}

\newtheorem{remark}[theorem]{Remark}
\newtheorem{proposition}[theorem]{Proposition}
\newtheorem{corollary}{Corollary}[section]

\newcommand{\CC}{\mathbb{C}}

\newcommand{\EE}{\mathbb{E}}
\newcommand{\FF}{\mathbb{F}}

\newcommand{\RR}{\mathbb{R}}

\newcommand{\ZZ}{\mathbb{Z}}
\DeclareSymbolFont{bbold}{U}{bbold}{m}{n}
\DeclareSymbolFontAlphabet{\mathbbold}{bbold}
\newcommand{\One}{\mathbbold{1}}

\renewcommand{\emptyset}{\varnothing}

\newcommand{\sym}{\mathrm{sym}}
\newcommand{\spec}{\mathrm{spec}}

\DeclareMathOperator*{\argmin}{arg\,min}
\DeclareMathOperator{\Aut}{Aut}

\DeclareMathOperator{\tr}{tr}

\newcommand{\Mod}[1]{\ (\mathrm{mod}\ #1)}
\newcommand{\SOS}{\mathrm{SOS}}
\newcommand{\tEE}{\widetilde{\mathbb{E}}}
\renewcommand{\epsilon}{\varepsilon}
\renewcommand{\Tilde}{\widetilde}

\newcommand\numberthis{\addtocounter{equation}{1}\tag{\theequation}}

\title{A degree 4 sum-of-squares lower bound for the clique number of the Paley graph}

\usepackage{authblk}

\author{Dmitriy Kunisky\thanks{Email: \texttt{dmitriy.kunisky@yale.edu}. Partially supported by ONR Award N00014-20-1-2335 and a Simons Investigator Award to Daniel Spielman.}\,}

\author{Xifan Yu\thanks{Email: \texttt{xifan.yu@yale.edu}. Partially supported by a Simons Investigator Award to Daniel Spielman.}\,}

\affil{Department of Computer Science, Yale University}

\date{April 25, 2024}

\begin{document}

\maketitle

\thispagestyle{empty}

\begin{abstract}
    We prove that the degree 4 sum-of-squares (SOS) relaxation of the clique number of the Paley graph on a prime number $p$ of vertices has value at least $\Omega(p^{1/3})$.
    This is in contrast to the widely believed conjecture that the actual clique number of the Paley graph is $O(\mathrm{polylog}(p))$.
    Our result may be viewed as a derandomization of that of Deshpande and Montanari (2015), who showed the same lower bound (up to $\mathrm{polylog}(p)$ terms) with high probability for the Erd\H{o}s-R\'{e}nyi random graph on $p$ vertices, whose clique number is with high probability $O(\log(p))$.
    We also show that our lower bound is optimal for the Feige-Krauthgamer construction of pseudomoments, derandomizing an argument of Kelner.
    Finally, we present numerical experiments indicating that the value of the degree 4 SOS relaxation of the Paley graph may scale as $O(p^{1/2 - \epsilon})$ for some $\epsilon > 0$, and give a matrix norm calculation indicating that the pseudocalibration construction for SOS lower bounds for random graphs will not immediately transfer to the Paley graph.
    Taken together, our results suggest that degree~4 SOS may break the ``$\sqrt{p}$ barrier'' for upper bounds on the clique number of Paley graphs, but prove that it can at best improve the exponent from $1/2$ to $1/3$.
\end{abstract}

\clearpage

\section{Introduction}

\subsection{Maximum and Planted Clique Problems in Random Graphs}

For a graph $G$, we denote by $\omega(G)$ the number of vertices in the largest \emph{clique} or complete subgraph in $G$.
Computing $\omega(G)$ is a classical NP-hard problem in combinatorial optimization, which is moreover hard to approximate within any polynomial factor $n^{1-\varepsilon}$ for $\varepsilon > 0$ \cite{Karp-1972-Reducibility,Hastad-1996-CliqueHardnessApproximation}.
Aside from this worst-case hardness, an average-case setting of computing $\omega(G)$ was proposed by Karp \cite{Karp-1976-ProbabilisticAnalysis}.
In this setting, the input graph is an Erd\H os-R\'enyi (ER) random graph $G$ on $n$ vertices, where each edge is present independently with probability $\frac{1}{2}$.
We denote this by distribution by $G \sim \mathcal{G}\left(n, \frac{1}{2}\right)$.
It is known that (see, e.g., \cite[Section 11.1]{Bollobas-2001-RandomGraphs}), with high probability,
\begin{equation}
\omega(G) \in \left[(2 - o(1))\log_2 n, (2 + o(1))\log_2n\right].
\end{equation}
In \cite{Karp-1976-ProbabilisticAnalysis}, Karp showed that a simple greedy algorithm with high probability finds a clique of size roughly $\log_2 n$, and asked whether a polynomial-time algorithm can with high probability find a clique of size $(1 + \varepsilon)\log n$ for any constant $\varepsilon > 0$.
The problem remains open, but, perhaps surprisingly, evidence has accumulated that such an algorithm does not exist \cite{Jerrum-1992-LargeCliques,GZ-2019-LandscapePlantedClique}.

A natural related problem is that of algorithmically \emph{bounding} the size of the largest clique in $G$, outputting a number that is always an upper bound on $\omega(G)$.
For example, under $G \sim \mathcal{G}(n, \frac{1}{2})$, a simple algorithm based on the maximum degree can produce a $O(\sqrt{n\log n})$ bound \cite{Kucera-1995-ComplexityGraphPartitioning}.
Spectral algorithms operating on the eigenvalues of the adjacency matrix of $G$ can improve this to $O(\sqrt{n})$ (for instance, using Haemers' generalization to irregular graphs of Hoffman's classical spectral bound on the clique number \cite{Haemers-1995-InterlacingEigenvaluesGraphs}).

The question of algorithmically bounding the clique number is also related to the problem of \emph{hypothesis testing} between $G \sim \mathcal{G}(n, \frac{1}{2})$ and $G$ drawn from another distribution where a typical $G$ contains a \emph{planted} clique of size much larger than $2\log_2 n$, since if we have an algorithm that always produces a valid bound on $\omega(G)$ and this bound is typically small for $G \sim \mathcal{G}(n, \frac{1}{2})$, then we can use its output to detect the planting of a sufficiently large clique.
The above then shows that we may detect the presence of a clique of size $C\sqrt{n}$ for sufficiently large $C$; \cite{AKS-1998-PlantedClique} moreover showed that an efficient spectral algorithm can even \emph{recover} the vertex set of a planted clique of this size.\footnote{Observe that, while a brute force search can both detect and recover a planted clique of any size $(2 + \varepsilon)\log_2 n$, this brute force search does not run in polynomial time.}

A long line of work considered whether using convex relaxations of $\omega(G)$ that produce bounds that are in general stronger than spectral bounds can break this ``$\sqrt{n}$ barrier'' for $G \sim \mathcal{G}(n, \frac{1}{2})$, with a particular focus on semidefinite programming (SDP) relaxations. \cite{Juhasz-1982-LovaszThetaRandomGraph} showed that Lov\'{a}sz's $\vartheta$ function \cite{Lovasz-1979-ShannonCapacityGraph} also has value $\Omega(\sqrt{n})$; \cite{FK-2000-PlantedClique} later considered further aspects of using the $\vartheta$ function for detecting and recovering planted cliques. \cite{FK-2003-LovaszSchrijverIndependentSet} showed the same $\Omega(\sqrt{n})$ lower bound for any constant level of the Lov\'{a}sz-Schrijver hierarchy of SDPs, of which the $\vartheta$ function is merely the first and weakest.
The stronger sum-of-squares (SOS) hierarchy of relaxations proved harder to analyze.
The pioneering but flawed analysis of \cite{MW-2013-PlantedClique} was fixed by \cite{MPW-2015-PlantedClique}, albeit at the cost of falling short of an $\Omega(\sqrt{n})$ lower bound.
Many subsequent works, first on the degree 4 SOS relaxation \cite{DM-2015-SOSPlantedClique,RS-2015-SOS4PlantedClique,HKP-2015-PlantedCliqueSOSMPW} and culminating in the development of the \emph{pseudocalibration} technique for larger degrees \cite{BHKKMP-2019-PlantedClique}, ultimately established an $\Omega(n^{1/2 - o(1)})$ lower bound for any constant degree of the SOS hierarchy.\footnote{The SOS hierarchy consists of a sequence of SDPs producing smaller and smaller upper bounds on $\omega(G)$, indexed by an even number called the \emph{degree}. See Section~\ref{sec:sos-clique-def} for a precise definition.}

All of these results apply, as we have mentioned, to the average case of computing the clique number over $G \sim \mathcal{G}(n, \frac{1}{2})$.
Some recent literature has revisited other average-case SOS lower bounds and identified \emph{deterministic} instances over which the same quality of lower bound holds (see in particular the work of \cite{DFHT-2020-ExplicitSOSLowerHDX,HL-2022-ExplicitCSPSOSBounds}, derandomizing the result of \cite{Grigoriev-2001-SOSParity} on refuting 3-XORSAT instances).\footnote{Here we are interested in quantitative lower bounds showing large integrality gaps, rather than arbitrarily small integrality gaps---deterministic explicit examples giving the latter for high degrees of SOS have been shown before for several problems in works such as \cite{Grigoriev-2001-SOSParity,Laurent-2003-CutPolytopeSOS}.}
In this paper, we initiate the study of the same question for the clique problem, by derandomizing the SOS lower bound of \cite{DM-2015-SOSPlantedClique} for the degree 4 SOS relaxation of $\omega(G)$ with $G \sim \mathcal{G}(n, \frac{1}{2})$.
The deterministic graphs that achieve this derandomization are the \emph{Paley graphs}, whose clique number is a question of independent interest in number theory.
We first review some background on the Paley graphs, and then describe our results.

\subsection{Paley Graphs, Pseudorandomness, and Derandomization}

The Paley graphs are an infinite family of graphs that exhibit certain \emph{pseudorandom} properties, behaving in some regards similarly to a typical $G \sim \mathcal{G}(n, \frac{1}{2})$.
They are defined on vertex sets identified with finite fields $\mathbb{F}_q$ of order $q \equiv 1 \Mod{4}$, where edges connect pairs of elements of $\mathbb{F}_q$ whose differences are quadratic residues. We denote the Paley graph on $\FF_q$ by $G_q$; the reader may see Section~\ref{sec:paley-def} for a more precise definition.

Many quantities that may be computed from Paley graphs are the same as those of typical graphs drawn from $\mathcal{G}(q, \frac{1}{2})$. In the simplest instance, Paley graphs are regular of degree $\frac{q - 1}{2}$, roughly the average degree of the corresponding random graph. \cite{CGW-1989-QuasirandomGraphs} showed that the same holds for the number of occurrences of any subgraph of constant size, for the first eigenvalue being asymptotically $\frac{q}{2}$, and the second eigenvalue being $o(q^{\frac{1}{2} + \varepsilon})$ for any $\varepsilon > 0$.

How far can we take this analogy?
It is natural to ask for subgraph counts of graphs of size growing slowly with $q$, and the clique number is just such a question: under $G \sim \mathcal{G}(q, \frac{1}{2})$ we have $\mathbb{E}[\omega(G)] \sim 2\log_2 q$, and we might expect the same for $\omega(G_q)$.

However, the clique number of Paley graphs is not well understood.
Let us review what is currently known.
Hoffman's spectral bound \cite{Hoffman-1970-Eigenvalues,Haemers-1995-InterlacingEigenvaluesGraphs} implies the upper bound
\begin{equation}
    \label{eq:Gq-hoffman}
    \omega(G_q) \leq \sqrt{q}.
\end{equation}
In fact, this is easy to derive by elementary combinatorial means (see, e.g., \cite{Yip-2022-CliqueNumberPaleyPrimePower}) and for this reason is sometimes called the \emph{trivial} upper bound on $\omega(G_q)$.
This is tight for $q = p^{2k}$ an even power of a prime, as $\FF_{\sqrt{q}}$ may be realized as a subfield of $\FF_q$ all of whose elements are quadratic residues in this case \cite{BDR-1988-CliqueNumbersPaleyGraphs}.

However, for odd prime powers, and even the simplest case $q = p$ a prime, the clique number is believed to be much lower.
The upper bound on the diagonal Ramsey number established by \cite{ES-1935-CombinatorialProblemGeometry} implies that
\begin{equation}
    \omega(G_p) \geq \left(\frac{1}{2} + o(1)\right) \log_2 p.
\end{equation}
By a number-theoretic analysis of the least quadratic non-residue modulo $p$, \cite{GR-1990-LowerBoundsLeastQNR} improved this, showing that for infinitely many primes $p$,
\begin{equation}
    \omega(G_p) \geq \log p \log\log\log p.
\end{equation}
Moreover, conditional on the Generalized Riemann Hypothesis, the $\log\log\log p$ term may be improved to $\log \log p$ \cite[Theorem 13.5]{Montgomery-1971-MultiplicativeNumberTheory}.\footnote{It is still possible to reconcile these results with the proposal that $G_p$ behaves like a random graph, so long as we adopt a more sophisticated random model than $\mathcal{G}(p, \frac{1}{2})$ \cite{Mrazovic-2017-RandomModelPaleyGraph}.}

On the other hand, the best known upper bound \cite{HP-2021-SumsetsPaleyClique,BSW-2021-ProductDirectionsGaloisPaley} improves only by a constant factor on the spectral bound \eqref{eq:Gq-hoffman},
\begin{equation}
    \label{eq:hp-bound}
    \omega(G_p) \leq \frac{\sqrt{2p - 1} + 1}{2} \sim \frac{\sqrt{p}}{\sqrt{2}}.
\end{equation}

In contrast to this state-of-the-art bound, $\omega(G_p)$ is widely believed to actually scale at most polylogarithmically with $p$ based on computations of $\omega(G_p)$ for small $p$.
We express this in the following conjecture; see \cite{Shearer-1986-LowerBoundsDiagonalRamsey,BMR-2013-SquaresDifferenceSets,Yip-2022-CliqueNumberPaleyPrimePower,KM-2023-SDPCliquePaley} as well as our Figure~\ref{fig:clique-numerics}.
\begin{conjecture}
    \label{conj:paley-polylog}
    For some $C, K > 0$ and all $p \equiv 1 \Mod{4}$ prime, $\omega(G_p) \leq C (\log p)^K$.
\end{conjecture}
\noindent
Numerical evidence suggests that we might in fact expect to be able to take $K = 2$, as discussed by \cite{BMR-2013-SquaresDifferenceSets,KM-2023-SDPCliquePaley} and illustrated in our Figure~\ref{fig:clique-numerics}.

Moreover, these graphs are believed to be good constructions for lower bounds on the diagonal Ramsey numbers $R(k,k)$. For example, the Paley graph of order~$17$ is the unique largest graph that contains neither a clique of size $4$ nor an independent set of size $4$, which shows that $R(4,4) = 18$ \cite{EPS-1981-CompleteSubgraphsPaley}. The current best known bound $R(6,6) \ge 102$ is established by the Paley graph of order $101$, which contains neither a clique of size $6$ nor an independent set of size $6$ \cite{Radziszowski-2011-SmallRamseyNumbers}.

Because of this application among others, it is a long-standing open problem in additive combinatorics and number theory to improve the upper bound for clique numbers of Paley graphs of prime orders, and in particular to break the ``$\sqrt{p}$ barrier'' and prove an upper bound scaling as $p^{1/2 - \epsilon}$ for some $\epsilon > 0$.\footnote{For instance, this is mentioned as ``probably a very hard problem'' in the problem list \cite{CL-2007-OpenProblemsAdditiveCombinatorics}.}
Some recent work has begun to explore whether convex relaxations of the clique number can lead to such improvements.
For instance, \cite{GLV-2009-BlockDiagonalSDPPaleyClique,KM-2023-SDPCliquePaley} explored using a hierarchy of SDPs producing bounds between that of the Lov\'{a}sz-Schrijver hierarchy and the SOS hierarchy for this purpose, and \cite{MMP-2019-LPCliquesPaley} empirically found that a modification of the Lov\'{a}sz $\vartheta$ function SDP can recover and sometimes slightly improve on the best-known upper bound \eqref{eq:hp-bound}.

\subsection{Our Contributions}

Our main result contributes to both of the lines of work outlined above.
On the one hand, it shows (conditional on Conjecture~\ref{conj:paley-polylog}) that the Paley graph gives a derandomization of the SOS lower bound of \cite{DM-2015-SOSPlantedClique} for ER random graphs.
On the other hand, it shows that a powerful convex optimization approach to upper-bounding the clique number cannot be too effective when applied to $G_p$.

\begin{theorem}
    \label{thm:main_theorem}
    There is a constant $c > 0$ such that the value of the degree 4 SOS relaxation of the clique number $\SOS_4(G)$, as defined in Section~\ref{sec:sos-clique-def}, evaluated with $G_p$ the Paley graph on $p$ vertices for $p$ any prime number with $p \equiv 1 \Mod{4}$, as defined in Section~\ref{sec:paley-def}, satisfies
    \begin{equation}
        \SOS_4(G_p) \geq cp^{1/3}.
    \end{equation}
\end{theorem}
\noindent
The main ingredients in our proof are new norm bounds for certain \emph{graph matrices} (as appear in the analysis of SOS relaxations for random graphs; see, e.g., \cite{AMP-2016-GraphMatrices}) formed from Paley graphs and certain character sum estimates for the Legendre symbol.

To elaborate on this result, we provide three further pieces of more detailed analysis.
Note that Theorem~\ref{thm:main_theorem} does not exclude the possibility that $\SOS_4(G_p) = o(\sqrt{p})$.
In Section~\ref{sec:optimality-fk}, however, we show that at least the lower bound construction we use to prove Theorem~\ref{thm:main_theorem}, involving the simple class of \emph{Feige-Krauthgamer pseudomoments} (see Definition~\ref{def:fk-pseudomoments}), cannot improve on the $p^{1/3}$ scaling of our lower bound.

On the other hand, in Section~\ref{sec:numerical}, we present some numerical evidence that $\SOS_4(G_p) \sim p^{\eta}$ for a constant $\eta \in (0, \frac{1}{2})$, with value $\eta \approx 0.4$.
As we discuss in Section~\ref{sec:numerical}, these results are similar to earlier numerical studies of \cite{GLV-2009-BlockDiagonalSDPPaleyClique}, who consider a weaker class of SDPs than the SOS hierarchy, and results of \cite{KM-2023-SDPCliquePaley}, who consider the same weaker SDPs and extract a prediction of the power scaling of their values with $p$ from numerical results.
We thus have reason to believe that our lower bound cannot be improved all the way to a scaling of $p^{1/2}$.
Unfortunately, we have not found a way to convert these numerical results into a proof of an improved bound on the clique number, but we leave this as a tantalizing open problem for future work.

Finally, to accompany these empirical results, we provide some modest theoretical evidence that the SOS hierarchy may break the $\sqrt{p}$ barrier for upper bounds on $\omega(G_p)$.
The tight analysis showing that $\EE[\SOS_{2d}(G)] = \Omega(n^{1/2 - o(1)})$ for $G \sim \mathcal{G}(n, \frac{1}{2})$ and any constant $d$ uses a construction satisfying a property called \emph{pseudocalibration} \cite{BHKKMP-2019-PlantedClique}, whose analysis hinges on norm bounds for the aforementioned graph matrices built from the adjacency matrix of $G$ \cite{AMP-2016-GraphMatrices}.
In Section~\ref{sec:failure-graphical}, we show that some of these norm bounds \emph{fail} for the Paley graph.
Thus, the analysis of the pseudocalibration construction for random graphs cannot be directly adapted to the case of Paley graphs.\footnote{We note that the initial premise of pseudocalibration, which involves comparing a pair of ``null'' and ``alternative'' random graph distributions, is not sensible to apply to the deterministic Paley graph. But, ultimately, the pseudocalibration argument yields a function mapping a graph to a matrix that one hopes will be feasible for a high-degree SOS program, and one may simply substitute the Paley graph into this function and consider the result.}

\section{Preliminaries and Proof Overview}

\subsection{Notations}

Throughout the paper, $p$ will denote a prime number, and $q$ a prime power $q = p^k$.
The finite field of order $q$ (unique up to isomorphism) is denoted by $\mathbb{F}_q$, and its group of units by $\mathbb{F}_q^\times$. A nonzero element $y$ of $\mathbb{F}_q$ is called a \emph{quadratic residue} of $\mathbb{F}_q$ if $y = x^2$ for some $x \in \mathbb{F}_q$, and a \emph{quadratic nonresidue} otherwise.
We write $(\mathbb{F}_q^\times)^2$ for the set of quadratic residues. We will also freely identify $\mathbb{F}_p$ with $\ZZ / p\ZZ$, with representatives $\{0, 1, \dots, p-1\}$.

We write $[n]\colonequals \{1, 2, \dots, n\}$. For a finite set $X$, we write $2^X$ for the power set, and $\binom{X}{k}$ and $\binom{X}{\le k}$ to denote the sets of subsets of $X$ with exactly $k$ elements and at most $k$ elements respectively. We also use $X_{(k)}$ to denote the set of tuples of elements of $X$ of length $k$ with all entries distinct.

When the discussion involves variables $\{x_i\}_{i \in \mathcal{I}}$ indexed by $\mathcal{I}$, for a subset $S \subset \mathcal{I}$, we will use $x^S$ to denote the monomial $\prod_{i \in S}x_i$.

We use $\boldsymbol{1} \in \mathbb{R}^n$ to denote the all-ones vector. We use $I \in \mathbb{R}^{n \times n}$ to denote the identity matrix, $J \in \mathbb{R}^{n \times n}$ to denote the all-ones matrix, and $\boldsymbol{0} \in \mathbb{R}^{n \times n}$ to denote the all-zeros matrix. The dimensions of these objects will be clear from context. For a real symmetric or Hermitian matrix $A$, we use $\spec(A)$ to denote its spectrum, which we write in double braces $\{\{ \cdots \}\}$ to indicate that the spectrum is a multiset. For matrices $A, B \in \mathbb{C}^{n \times n}$ and $C \in \mathbb{C}^{m \times m}$, we use $A \circ B \in \mathbb{C}^{n \times n}$ to denote the Hadamard product (entrywise product) of $A$ and $B$, and $A \otimes C \in \mathbb{C}^{nm \times nm}$ to denote the Kronecker product (tensor product) of $A$ and $C$.

For a graph $G = (V, E)$, we use $V(G)$ to denote its vertex set and $E(G)$ to denote its edge set. We use $\overline{G}$ to denote the complement of $G$. For vertices $u,v \in V(G)$, we use $u \sim_G v$ to indicate that $u$ and $v$ are adjacent in $G$ and $u \not\sim_G v$ to indicate that they are not adjacent. We will use $A_G$ to denote the $\{0, 1\}$ adjacency matrix of $G$, and $S_G$ to denote the Seidel or $\{\pm 1\}$ adjacency matrix. We drop the subscript $G$ when the graph is clear from context. Conventionally, the Seidel adjacency matrix is $-1$ on pairs of adjacent vertices, $+1$ on pairs of nonadjacent vertices, and $0$ on the diagonal. In this paper, we abuse this term to mean the matrix that is $1$ on pairs of adjacent vertices, $-1$ on pairs of nonadjacent vertices, and $0$ on the diagonal, as this is more conveniently written in terms of the Legendre symbol in the context of Paley graphs (see Section~\ref{sec:graph-mx-decomp}). It is easy to see that the $A_G$ and $S_G$ are related by $S_G = 2A_G - J + I$. Lastly, we write $\mathcal{K}(G)$ for the set of subsets of $V(G)$ that form cliques in $G$.

We will use the standard asymptotic notations $O(\cdot), \Omega(\cdot), \Theta(\cdot),$ and $o(\cdot)$. We will use $\widetilde{O}(\cdot)$ and $\widetilde{\Omega}(\cdot)$ to additionally suppress polylogarithmic factors.

For a complex number $z \in \mathbb{C}$, we denote the complex conjugate of $z$ as $\overline{z}$. For a complex vector or matrix $Z$, we will use $Z^*$ to denote the conjugate transpose of $Z$.

\subsection{Problem Setup}

Let us now specify in full detail the SOS relaxations $\SOS_{2d}$ of the clique number, and the Paley graphs $G_p$.

\subsubsection{Sum-Of-Squares Relaxations of the Clique Number}
\label{sec:sos-clique-def}

Let $G$ be a graph of order $n$. The clique number $\omega(G)$ of $G$ is equal to the value of the following polynomial optimization program:

\begin{equation}
\label{clique_program}
\omega(G) = \left\{\begin{array}{ll}
    \text{maximize} & \sum_{i \in V(G)} x_i \\
    \text{subject to} & x_i^2 = x_i \text{ for all } i \in V(G), \\
    & x_ix_j = 0 \text{ for all } i, j \in V(G) \text{ with } i \neq j \text{ and } i \not \sim_G j \end{array}\right\}.
\end{equation}

It is easy to see that the feasible solutions of the program above are in one-to-one correspondence with the indicator vectors of the cliques in $G$. Before we introduce the SOS relaxations of the clique number, let us first define the \emph{pseudoexpectation} operators over which the SOS relaxations optimize.

\begin{definition}[Pseudoexpectation]
    We say $\widetilde{\mathbb{E}}: \mathbb{R}[x_1, \dots, x_n]_{\le 2d} \to \mathbb{R}$ is a \emph{degree $2d$ pseudoexpectation} with respect to polynomial constraints $\{f_i(x) = 0\}_{i=1}^a$, $\{g_j(x) \ge 0\}_{j=1}^b$ if the following properties hold:
    \begin{itemize}
        \item $\widetilde{\mathbb{E}}$ is linear.
        \item $\widetilde{\mathbb{E}}[1] = 1$.
        \item $\widetilde{\mathbb{E}}[f_i(x)p(x)] = 0$, for all $p(x) \in \mathbb{R}[x_1, \dots, x_n]$ such that $\deg(f_i p) \le 2d$, for all $1\le i\le a$.
        \item $\widetilde{\mathbb{E}}[p(x)^2] \ge 0$, for all $p(x) \in \mathbb{R}[x_1, \dots, x_n]_{\le d}$.
        \item $\widetilde{\mathbb{E}}[g_j(x)p(x)^2] \ge 0$, for all $p(x) \in \mathbb{R}[x_1, \dots, x_n]$ such that $\deg(g_j p^2) \le 2d$, for all $1\le j\le b$.
    \end{itemize}
\end{definition}
In the case of the maximum clique program \eqref{clique_program}, the polynomial constraints are generated by the Boolean constraints $x_i^2 - x_i = 0$ for $i \in V(G)$ and the clique constraints $x_ix_j = 0$ for $i,j\in V(G)$ with $i \neq j$ and $i \not\sim_G j, $.  For convenience, let us identify the vertex set $V(G)$ with $[n]$ where $n = |V(G)|$. Then, the degree $2d$ SOS relaxation of the polynomial optimization program \eqref{clique_program} written in terms of pseudoexpectations is
\begin{equation}
\label{pseudoexpectation_SOS}
    \SOS_{2d}(G) = \left\{
    \begin{array}{ll}
    \text{maximize} & \sum_{i = 1}^n \tEE[x_i] \\
    \text{subject to} & \tEE: \RR[x_1, \dots, x_n]_{\leq 2d} \to \RR \text{ linear}, \\
    & \tEE[1] = 1, \\
    & \tEE[(x_i^2 - x_i)p(x)] = 0 \text{ for all } i \in [n], \deg(p) \leq 2d - 2, \\
    & \tEE[x_ix_jp(x)] = 0 \text{ for all } i \not\sim_G j, \deg(p) \leq 2d - 2, \\
    & \tEE[p(x)^2] \geq 0 \text{ for all } \deg(p) \leq d.
    \end{array}
    \right\}.
\end{equation}
To see that this is indeed a relaxation of the clique program \eqref{clique_program}, observe that for any probability measure $\mu: 2^{[n]} \to \mathbb{R}^{\ge 0}$ supported on the cliques of the graph $G$, the corresponding expectation operator $\mathbb{E}_{\mu}$ is a pseudoexpectation of any degree.

For every monomial $x^S$ for $S \in \binom{[n]}{\le 2d}$, $\widetilde{\mathbb{E}}[x^S]$ is called the \emph{pseudomoment} of $S$ of the corresponding pseudoexpectation $\widetilde{\mathbb{E}}$. By linearity, every pseudoexpectation of degree $2d$ is uniquely determined by its pseudomoments of degree at most $2d$, i.e., by the set $\{\widetilde{\mathbb{E}}[x^S]: S \subseteq [n], |S| \le 2d\}$. We may therefore encode the pseudoexpectation in the \emph{pseudomoment matrix} $M \in \mathbb{R}_{\sym}^{\binom{[n]}{\le d} \times \binom{[n]}{\le d} }$ with entries
\begin{equation}
    M_{S, T} = \widetilde{\mathbb{E}}[x^S x^T].
\end{equation}
This is especially convenient since the positivity of $\tEE$ on squared polynomials is equivalent to positive semidefiniteness of $M$. We can then rewrite the above program \eqref{pseudoexpectation_SOS} in the form of an SDP:
\begin{equation}
\label{SDP_SOS}
    \SOS_{2d}(G) = \left\{
    \begin{array}{ll}
    \text{maximize} & \sum_{i = 1}^n M_{\emptyset, \{i\}} \\
    \text{subject to} & M \in \RR^{\binom{[n]}{\leq d} \times \binom{[n]}{\leq d}} \\
    & M_{\emptyset, \emptyset} = 1, \\
    & M_{S, T} \text{ depends only on } S \cup T, \\
    & M_{S, T} = 0 \text{ whenever } S \cup T \notin \mathcal{K}(G), \\
    & M \succeq 0.
    \end{array}
    \right\}.
\end{equation}
We will not verify in detail the equivalence of \eqref{SDP_SOS} and \eqref{pseudoexpectation_SOS}; the reader may consult \cite{Laurent-2009-SOS} for an overview of this pseudomoment matrix framework, or the papers \cite{DM-2015-SOSPlantedClique,RS-2015-SOS4PlantedClique,HKP-2015-PlantedCliqueSOSMPW,BHKKMP-2019-PlantedClique} on SOS relaxations of $\omega(G)$ for further details.

\begin{remark}[Pseudomoment matrix compression]
    \label{rem:clique-indexing}
    We note that the row and column of $M$ indexed by any $S \notin \mathcal{K}(G)$ is forced by the constraints to be identically zero.
    These entries do not affect the positivity of $M$ and do not play a role in the objective function, so we may just as well take $M$ to be indexed by \emph{cliques} of size at most $d$ rather than arbitrary subsets of vertices.
\end{remark}

In the special case $2d = 2$, the SDP in \eqref{SDP_SOS} takes the form
\begin{equation}
    \SOS_{2}(G) = \left\{
    \begin{array}{ll}
    \text{maximize} & \sum_{i = 1}^n y_i \\
    \text{subject to} & y \in \RR^n, Y \in \RR^{n \times n}_{\sym}, \\
    & Y_{i,i} = y_i \text{ for all } i \in [n], \\
    & Y_{i, j} = 0 \text{ for all } i, j \in [n] \text{ with } i \neq j \text{ and } i \not\sim_G j, \\
    & M = \begin{bmatrix}
    1 & y^\top\\
    y & Y
    \end{bmatrix} \succeq 0.
    \end{array}
    \right\}.
\end{equation}
One can show (see \cite{GLS-2012-GeometricAlgorithmsOptimization,GL-2017-LovaszThetaVariants}) that the program above is equivalent to the \emph{Lov\'asz $\vartheta$ function} of the complement graph $\overline{G}$, a well-known upper bound on $\omega(G)$ due to \cite{Lovasz-1979-ShannonCapacityGraph}:
\begin{equation}
    \SOS_2(G) = \vartheta(\overline{G}).
\end{equation}
This SDP enjoys many special properties, some of which we will mention below; the reader may consult the above references for further information.

On the other hand, once the degree increases to $2d = 4$, the resulting SDP is not as well understood.
This SDP, which we study in the remainder of the paper, takes the form
\begin{equation}
\label{degree_4_SOS}
    \SOS_{4}(G) = \left\{
    \begin{array}{ll}
    \text{maximize} & \sum_{i = 1}^n M^{0, 1}_{\emptyset, i} \\
    \text{subject to} & M^{r, c} \in \RR^{\binom{[n]}{r} \times \binom{[n]}{c}} \text{ for } r, c \in \{0, 1, 2\}, \\
    & M^{r, c}_{S, T} \text{ depends only on } S \cup T, \\
    & M^{r, c}_{S, T} = 0 \text{ whenever } S \cup T \notin \mathcal{K}(G), \\
    & M = \begin{bmatrix}
    1 & M^{0, 1} & M^{0, 2} \\
    M^{1, 0} & M^{1, 1} & M^{1,2}\\
    M^{2, 0} & M^{2,1} & M^{2,2}
\end{bmatrix} \succeq 0
\end{array}
    \right\}.
\end{equation}

\subsubsection{Paley Graphs}
\label{sec:paley-def}

We now give the definition and some useful basic properties of the Paley graphs.

\begin{definition}[Paley graph]
Let $q = p^k$ be a prime power such that $q \equiv 1 \Mod{4}$. The \emph{Paley graph} $G_q$ of order $q$ then has vertex set $V(G_q) \colonequals \mathbb{F}_q$ and edge set
\begin{equation}
    E(G_q) \colonequals \left\{\{a,b\} \in \binom{\mathbb{F}_q}{2}: a - b \in (\mathbb{F}_q^\times)^2 \right\}.
\end{equation}
\end{definition}
\noindent
The condition $q \equiv 1 \Mod{4}$ ensures that $-1$ is a square in $\mathbb{F}_q$. As a result, $a - b \in (\mathbb{F}_q^\times)^2$ if and only if $ b- a\in (\mathbb{F}_q^\times)^2$, so the edge set is well-defined.

\begin{proposition}[Regularity]
    \label{fact:regularity}
    For any prime power $q \equiv 1 \Mod{4}$, the Paley graph $G_q$ is strongly regular with parameters $(q, \frac{q-1}{2}, \frac{q-5}{4}, \frac{q-1}{4})$, i.e., it is regular of degree $\frac{q-1}{2}$, every pair of adjacent vertices share $\frac{q-5}{4}$ common neighbors, and every pair of non-adjacent vertices share $\frac{q-1}{4}$ common neighbors.
\end{proposition}

\begin{proposition}[Spectrum]
    \label{fact:spectrum}
    For any prime power $q \equiv 1 \Mod{4}$, the spectrum of the $\{0,1\}$ adjacency matrix $A_{G_q}$ of $G_q$ is
    \begin{equation}
        \spec(A_{G_q}) = \bigg\{\bigg\{\frac{q-1}{2}, \underbrace{\frac{-1 + \sqrt{q}}{2}, \dots, \frac{-1 + \sqrt{q}}{2}}_{\frac{q-1}{2} \text{ times}},  \underbrace{\frac{-1-\sqrt{q}}{2}, \dots, \frac{-1-\sqrt{q}}{2}}_{\frac{q-1}{2} \text{ times}}\bigg\}\bigg\},
    \end{equation}
    and the spectrum of the Seidel or $\{\pm 1\}$ adjacency matrix $S_{G_q}$ of $G_q$ is
    \begin{equation}
        \spec(S_{G_q}) = \bigg\{\bigg\{0, \underbrace{\sqrt{q}, \dots, \sqrt{q}}_{\frac{q-1}{2} \text{ times}},  \underbrace{-\sqrt{q}, \dots, -\sqrt{q}}_{\frac{q-1}{2} \text{ times}}\bigg\}\bigg\}.
    \end{equation}
\end{proposition}

\begin{proposition}[Automorphisms]
    \label{fact:automorphisms}
    For any prime $p \equiv 1 \Mod{4}$, the automorphism group of $G_p$ is the set of maps $f: \FF_p \to \FF_p$ given by $f(x) = ax + b$ for any $a \in (\FF_p^{\times})^2$ and $b \in \FF_p$.
    In particular, $G_p$ is both vertex- and edge-transitive.
\end{proposition}

We will study the SOS relaxations of the clique number of Paley graphs, $\SOS_{2d}(G_q)$.
Recall that the degree 2 SOS relaxation of the clique number of the Paley graph $G_q$ is equal to the Lov\'asz theta function of its complement, $\SOS_2(G_q) = \vartheta(\overline{G_q})$.
Since $G_q$ is self-complementary (under the automorphism $x \mapsto gx$ for $g$ a multiplicative generator of $\FF_q^{\times}$), $\vartheta(\overline{G_q}) = \vartheta(G_q)$. Since $G_q$ is vertex-transitive, by Lov\'asz's result in \cite{Lovasz-1979-ShannonCapacityGraph},
\begin{equation}
    \vartheta(\overline{G_q})\vartheta(G_q) = |V(G_q)| = q,
\end{equation}
whereby combining our observations shows that
\begin{equation}
    \SOS_2(G_q) = \sqrt{q}.
\end{equation}
This is the same as the upper bound of the clique number given by Hoffman's spectral bound.
Thus, degree 2 SOS does not improve on the spectral bound, and degree 4 SOS, which we begin to analyze with Theorem~\ref{thm:main_theorem}, is the first more interesting degree.

\subsection{Proof Overview}

To prove Theorem~\ref{thm:main_theorem}, we will construct a feasible pseudomoment matrix $M$ for the program \eqref{degree_4_SOS} that has objective value $\Omega(p^{1/3})$. We will consider the following type of pseudomoments, which we call \emph{Feige-Krauthgamer (FK) pseudomoments}, first studied by Feige and Krauthgamer \cite{FK-2003-LovaszSchrijverIndependentSet} to prove lower bounds on Lov\'asz-Schrijver relaxations for the maximum independent set of random graphs (sometimes these are called \emph{MPW pseudomoments} after their use by the later paper \cite{MPW-2015-PlantedClique}).

\begin{definition}[Feige-Krauthgamer pseudomoments]
    \label{def:fk-pseudomoments}
    Consider the degree $2d$ SOS relaxation of the clique number of a graph $G$. We say the pseudomoments of a degree $2d$ pseudoexpectation $\widetilde{\mathbb{E}}$ are \emph{Feige-Krauthgamer (FK) pseudomoments} if there exists a sequence $1 = \alpha_0, \alpha_1, \alpha_2, \dots, \alpha_{2d} \in \RR$ such that
    \begin{equation}
        \widetilde{\mathbb{E}}[x^S] = \begin{cases}
            \alpha_{|S|} & \text{ if } S \in \mathcal{K}(G) \text{ (i.e., if $S$ is a clique in } G\text{)}\\
            0 & \text{ otherwise.}
        \end{cases}
    \end{equation}
\end{definition}
\noindent
We note that FK pseudomoments automatically satisfy all conditions on a pseudoexpectation other than positivity.

The line of work beginning with \cite{MW-2013-PlantedClique} sought to use FK pseudomoments to prove lower bounds on SOS relaxations of $\omega(G)$ for random graphs $G$.\footnote{Some works, wanting to study an SOS relaxation that included the ``exact'' constraint $\sum_{i = 1}^n x_i = k$ for some $k$, adjusted the FK pseudomoments to satisfy the consequences of this constraint (see, e.g., \cite{HKP-2015-PlantedCliqueSOSMPW,Pang-2021-PlantedCliqueSOSExact}). We do not take this route here.}
While eventually in \cite{HKP-2015-PlantedCliqueSOSMPW,BHKKMP-2019-PlantedClique} it was found that FK pseudomoments could \emph{not} prove optimal $\Omega(\sqrt{n})$ lower bounds, earlier works still proved polynomial $\Omega(n^{\eta})$ lower bounds with $\eta < \frac{1}{2}$ using FK pseudomoments, which are simpler to define and to work with than the alternatives developed later.
In particular, our analysis will closely follow that of \cite{DM-2015-SOSPlantedClique}, who used FK pseudomoments to prove that $\SOS_4(G) = \widetilde{\Omega}(n^{\frac{1}{3}})$ with high probability for $G \sim \mathcal{G}(n, \frac{1}{2})$. \cite{HKP-2015-PlantedCliqueSOSMPW} later showed that, up to polylogarithmic factors, this is optimal over any choice of FK pseudomoments for the degree 4 relaxation.

\begin{remark}[Partial symmetry]
By vertex transitivity and edge transitivity of Paley graphs (Proposition~\ref{fact:automorphisms}), there always exists an optimal degree 4 pseudoexpectation giving all $\tEE[x_i]$ the same value and all $\tEE[x_ix_j]$ with $i \sim j$ in $G_p$ the same value, regardless of whether $\tEE$ is given by FK pseudomoments or not.
This strong symmetry of course fails to hold for ER random graphs.
\end{remark}

Recall that in the degree $4$ SOS program \eqref{degree_4_SOS}, we write the pseudomoment matrix $M$ in the block form
\begin{align}
    M = \begin{bmatrix}
        1 & M^{0,1} & M^{0,2}\\
        M_{1,0} & M^{1,1} & M^{1,2}\\
        M_{2,0} & M^{2,1} & M^{2,2}
    \end{bmatrix}.
\end{align}
We will follow the strategy of \cite{DM-2015-SOSPlantedClique} to successively check the Schur complement conditions for positive semidefiniteness of $M$.
Namely, we will rely on the following fact.

\begin{proposition}
    Let
    \begin{equation}
        M = \begin{bmatrix}
    A & B^\top\\
    B & C
\end{bmatrix} \in \mathbb{R}^{(a+b)\times (a+b)}
    \end{equation}
    be a real symmetric matrix written in block form, with $A \in \mathbb{R}^{a \times a}$ and $C \in \mathbb{R}^{b \times b}$. If $A \succ 0$ and $C - BA^{-1} B^T \succeq 0$, then $M \succeq 0$. We call the matrix $C - BA^{-1} B^\top$ the \emph{Schur complement} of the block $A$ in $M$.
\end{proposition}

The outline of our proof of Theorem~\ref{thm:main_theorem} is then as follows, in which we set appropriate values for $\alpha_1, \alpha_2, \alpha_3,$ and $\alpha_4$ to produce a feasible pseudoexpectation achieving the lower bound claimed in the result.
\begin{enumerate}
    \item \textbf{First Schur complement:} To show $M \succ 0$, it suffices to show that the Schur complement of the top left $1 \times 1$ block (containing just the scalar 1) in $M$,
    \begin{align}
        N \colonequals \begin{bmatrix}
            M^{1,1} & M^{1,2}\\
            M^{2,1} & M^{2,2}
        \end{bmatrix} - \begin{bmatrix}
            M^{1,0}\\
            M^{2,0}
        \end{bmatrix}
        \begin{bmatrix}
            M^{0,1} & M^{0,2}
        \end{bmatrix},
    \end{align}
    is positive semidefinite. We again write $N$ in the block form
    \begin{align}
        N = \begin{bmatrix}
            N^{1,1} & N^{1,2}\\
            N^{2,1} & N^{2,2}
        \end{bmatrix},
    \end{align}
    where $N^{i,j} \colonequals M^{i,j} - M^{i,0}M^{0,j}$ for all $i,j \in \{1,2\}$.

    \item \textbf{Filling in zero rows and columns:}
    The $N^{i,j}$ will have those rows and columns indexed by pairs $\{i, j\}$ that are not edges of $G_p$ identically equal to zero (see Remark~\ref{rem:clique-indexing}). To make use of the formalism of \emph{graph matrices} commonly used in the analysis of pseudomoments for SOS lower bounds, we fill in these rows and columns with a natural extension of the non-zero entries of the $N^{i,j}$ to get modified matrices
    \begin{align}
        H = \begin{bmatrix}
            H^{1,1} & H^{1,2}\\
            H^{2,1} & H^{2,2}
        \end{bmatrix},
    \end{align}
    where $N$ equals $H$ with some rows and columns set to zero.
    In particular, if $H \succeq 0$, then $N \succeq 0$ as well.

    \item \textbf{Second Schur complement:} To show $H \succeq 0$, it suffices to show that $H^{1,1} \succ 0$, and
    that the Schur complement of $H^{1,1}$ in $H$,
    \begin{align}
        H^{2,2} - H^{2,1}(H^{1,1})^{-1} H^{1,2},
    \end{align}
    is positive semidefinite.

    \item \textbf{Analysis of $H^{1,1}$:} Under the FK pseudomoments, $M^{1,0}M^{0,1}$ is a constant multiple of the all-ones matrix $J$.
    The matrix $H^{1,1} = N^{1,1} = M^{1,1} - M^{1,0}M^{0,1}$ is then a linear combination of the simultaneously diagonalizable matrices $I$, $J$, and $A_{G_{p}}$, and its spectrum can be computed from that of $A_{G_p}$. Thus, we obtain conditions on $\alpha_1$ and $\alpha_2$ that ensure that $H^{1,1} \succ 0$.

    \item \textbf{Graph matrices and subspace decomposition:} To prove $H^{2,2} - H^{2,1}(H^{1,1})^{-1}H^{1,2} \succeq 0$, we will first write $H^{2,2}$ and $H^{2,1}(H^{1,1})^{-1} H^{1,2}$ each as a sum of graph matrices. Then, we will use an idea from \cite{MPW-2015-PlantedClique,DM-2015-SOSPlantedClique} of separating the contributions of graph matrices along three subspaces of $\RR^{\binom{\FF_p}{2}}$, which correspond to the decomposition of this space, viewed as a representation of $S_p$ under the action $\sigma(\{a, b\}) = \{\sigma(a), \sigma(b)\}$, into irreducible subrepresentations.
    The remaining norm bounds will constitute the bulk of the proof, and it is in proving these that we will need to substitute for the probabilistic analysis of \cite{DM-2015-SOSPlantedClique,AMP-2016-GraphMatrices} more number-theoretic arguments about character sums, which are given in Section~\ref{sec:graph-matrix-proofs}.
\end{enumerate}

\section{Proof of Theorem~\ref{thm:main_theorem}}

We restate Theorem~\ref{thm:main_theorem} in more detailed terms of the FK pseudomoments that we will construct.

\begin{theorem}\label{thm:restated_main_theorem}
    There exists a constant $c > 0$ so that, setting $\alpha_1 \colonequals c p^{-2/3}, \alpha_2 \colonequals 4\alpha_1^2, \alpha_3 \colonequals 8\alpha_1^3$, and $\alpha_4 \colonequals 512 \alpha_1^4$, the FK pseudomoments defined by these parameters give a feasible solution to the degree $4$ SOS relaxation \eqref{degree_4_SOS} of the clique number of the Paley graphs $G_p$ for all sufficiently large $p$.
\end{theorem}
\noindent
Theorem~\ref{thm:main_theorem} follows, since the above gives, for all sufficiently large $p$,
\begin{equation}
    \SOS_4(G_p) \geq p \cdot cp^{-2/3} = cp^{1/3}.
\end{equation}

To remind the reader of the notations we set in the previous section, the pseudomoment matrix in the degree $4$ SOS relaxation \eqref{degree_4_SOS} is denoted
\begin{equation}
    M = \begin{bmatrix}
        1 & M^{0,1} & M^{0,2}\\
        M^{1,0} & M^{1,1} & M^{1,2}\\
        M^{2,0} & M^{2,1} & M^{2,2}
    \end{bmatrix},
\end{equation}
and we take this to be given by the FK pseudomoments proposed in Theorem \ref{thm:restated_main_theorem}. Recall that $M^{r,c} \in \mathbb{R}^{\binom{\mathbb{F}_p}{r} \times \binom{\mathbb{F}_p}{c}}$ for all $r,c\in \{0,1,2\}$. We will use
\begin{equation}
    N = \begin{bmatrix}
        N^{1,1} & N^{1,2}\\
        N^{2,1} & N^{2,2}
    \end{bmatrix}
\end{equation}
to denote the Schur complement of the top left $1 \times 1$ block in $M$.

\subsection{Filling Zero Rows and Columns}

As mentioned before, we will fill in the zero rows and columns of $N$ in order to make use of graph matrices.
In this section, we define the matrix
\begin{equation}
    H = \begin{bmatrix}
    H^{1,1} & H^{1,2}\\
    H^{2,1} & H^{2,2}
\end{bmatrix}
\end{equation}
that will achieve this filling.

\begin{definition}
    We write $\One_{k}: \binom{\mathbb{F}_p}{k} \to \{0,1\}$ for the function with $\One_{k}(S) = 1$ if $S$ is a clique in $G_p$ and $\One_{k}(S) = 0$ otherwise.
\end{definition}

\noindent
We now expand the $N^{\bullet, \bullet}$ matrices in terms of this indicator function.
\begin{proposition}
    Under the FK pseudomoments proposed in Theorem \ref{thm:restated_main_theorem}, the matrix $N$ can be written as
    \begin{equation}
        N = \begin{bmatrix}
            N^{1,1} & N^{1,2}\\
            N^{2,1} & N^{2,2}
        \end{bmatrix},
    \end{equation}
    where $N^{1,1} \in \mathbb{R}^{\mathbb{F}_p \times \mathbb{F}_p}, N^{1,2} \in \mathbb{R}^{\mathbb{F}_p \times \binom{\mathbb{F}_p}{2}}, N^{2,1} = {N^{1,2}}^\top \in \mathbb{R}^{\binom{\mathbb{F}_p}{2} \times \mathbb{F}_p}, N^{2,2} \in \mathbb{R}^{\binom{\mathbb{F}_p}{2} \times \binom{\mathbb{F}_p}{2}}$ have entries
    \begin{align}
        N^{1,1}_{a,b} &= \begin{cases}
            \alpha_1 - \alpha_1^2 & \text{ if } a = b, \\
            \alpha_2 \One_2(\{a,b\}) - \alpha_1^2 & \text{ if } a\ne b,
        \end{cases} \\
        N^{1,2}_{a, \{b,c\}} &= \begin{cases}
            \left(\alpha_2-\alpha_1\alpha_2\right) \One_2(\{b,c\}) & \text{ if } a \in \{b,c\}, \\
            \alpha_3 \One_3(\{a,b,c\}) - \alpha_1\alpha_2\One_2(\{b,c\}) & \text{ if } a \not\in \{b,c\},
        \end{cases} \\
        N^{2,2}_{\{a,b\}, \{c,d\}} &= \begin{cases}
            (\alpha_2 - \alpha_2^2)\One_2(\{a,b\}) & \text{ if } \{a,b\} = \{c,d\}, \\
            \alpha_3 \One_3(\{a,b\} \cup \{c,d\}) - \alpha_2^2 \One_2(\{a,b\}) \One_2(\{c,d\}) & \text{ if } \left|\{a,b\} \cap \{c,d\} \right| = 1, \\
            \alpha_4\One_4(\{a,b,c,d\}) - \alpha_2^2 \One_2(\{a,b\}) \One_2(\{c,d\}) & \text{ if } \{a,b\} \cap \{c,d\} = \emptyset.
        \end{cases}
    \end{align}
\end{proposition}
\noindent
Per Remark~\ref{rem:clique-indexing}, rows and columns indexed by pairs are identically zero in any of these matrices for all pairs that are not edges in $G_p$.

Next, we define matrices $H^{\bullet, \bullet}$ based on the $N^{\bullet, \bullet}$ by replacing the clique indicator functions with ``bipartite'' versions of those indicator functions, that only depend on the presence of edges between two subsets of vertices.

\begin{definition}
    We write $\One_{\ell,r} : \binom{\mathbb{F}_p}{\ell} \times \binom{\mathbb{F}_p}{r} \to \{0,1\}$ for the function with
    \begin{equation}
        \One_{\ell,r}(L, R) = \begin{cases}
            1 & \text{ if } v \sim_{G_p} w \text{ for all } v \in L \setminus R, w \in R \setminus L, \\
            0 & \text{ otherwise.}
        \end{cases}
    \end{equation}
    In other words, $\One_{\ell,r}(L, R) = 1$ if and only if all pairs of vertices in $\binom{L \cup R}{2}$ that don't belong simultaneously to $L$ or $R$ are connected in $G_p$.
\end{definition}

Now we are ready to state what matrix $H$ is: it is given by blocks $H^{1,1} \in \mathbb{R}^{\mathbb{F}_p \times \mathbb{F}_p}, H^{1,2} \in \mathbb{R}^{\mathbb{F}_p \times \binom{\mathbb{F}_p }{2} }, H^{2,1} = {H^{1,2}}^\top,$ and $H^{2,2} \in \mathbb{R}^{\binom{\mathbb{F}_p }{2} \times \binom{\mathbb{F}_p }{2}}$ having entries
    \begin{align}
        H^{1,1}_{a,b} &= \begin{cases}
            \alpha_1 - \alpha_1^2 & \text{ if } a = b\\
            \alpha_2 \One_{1,1}(\{a\},\{b\}) - \alpha_1^2 & \text{ if } a\ne b
        \end{cases},\\
        H^{1,2}_{a, \{b,c\}} &= \begin{cases}
            \alpha_2-\alpha_1\alpha_2 & \text{ if } a \in \{b,c\}\\
            \alpha_3 \One_{1,2}(\{a\},\{b,c\}) - \alpha_1\alpha_2 & \text{ if } a \not\in \{b,c\}
        \end{cases},\\
        H^{2,2}_{\{a,b\}, \{c,d\}} &= \begin{cases}
            \alpha_2 - \alpha_2^2 & \text{ if } \{a,b\} = \{c,d\}\\
            \alpha_3 \One_{2,2}(\{a,b\}, \{c,d\}) - \alpha_2^2 & \text{ if } \left|\{a,b\} \cap \{c,d\} \right| = 1\\
            \alpha_4\One_{2,2}(\{a,b\}, \{c,d\}) - \alpha_2^2 & \text{ if } \{a,b\} \cap \{c,d\} = \emptyset
        \end{cases}.
    \end{align}

    It is easy to see that proving positive semidefiniteness for $H$ also proves $N$ is positive semidefinite, due to the following observation.
    \begin{proposition}\label{prop:principal_submatrix}
        Up to permutation of rows and columns, $N$ is the direct sum of the principal submatrix of $H$ indexed by singletons and the edges of $G_p$ with a zero matrix.
    \end{proposition}
    \noindent
    The proof is simply that, for $|L|, |R| \leq 2$, we have $\One_{|L \cup R|}(L \cup R) = \One_{|L|, |R|}(L, R)$ so long as $L$ is an edge if $|L| = 2$ and $R$ is an edge if $|R| = 2$.

\subsection{Second Schur Complement Bounds}

Next, the goal is to prove under the same setting of Theorem \ref{thm:restated_main_theorem} that $H \succeq 0$.
To do this, we take another Schur complement, which we analyze below.

We will use $Q_0 = \frac{1}{p} J \in \mathbb{R}^{\mathbb{F}_p \times \mathbb{F}_p}$ to denote the orthogonal projection matrix to the constant vector, and $Q_1 = I - Q_0$ to denote the projection matrix to the orthogonal complement.

\begin{proposition} \label{prop:H11}
    Under the FK pseudomoments specified by $\alpha_1, \alpha_2, \alpha_3, \alpha_4$ in Theorem \ref{thm:restated_main_theorem}, for any constant $\varepsilon > 0$, the matrix $H^{1,1}$ satisfies
    \begin{align}
        H^{1,1} \succeq \left(\alpha_1 + \frac{p-1}{2}\alpha_2 - p\alpha_1^2\right)Q_0 + (1 - \varepsilon)\alpha_1 Q_1 \succ 0
    \end{align}
    for all sufficiently large primes $p$.
\end{proposition}

\begin{proof}
    Under the FK pseudomoments specified by $\alpha_1, \alpha_2, \alpha_3, \alpha_4$, we may express the following blocks of the pseudomoment matrix $M$ as
    \begin{align}
        M^{0,1} &= \alpha_1 \boldsymbol{1}^\top,\\
        M^{1,1} &= \alpha_1 I + \alpha_2 A_{G_p}.
    \end{align}
    We note that $H^{1,1} = N^{1,1}$, as this matrix is not affected by the filling of zero rows and columns, so we have
    \begin{align}
        H^{1,1} = N^{1,1} &= M^{1,1} - M^{1,0}M^{0,1} \nonumber \\
        &= \alpha_1 I + \alpha_2 A_{G_p} - \alpha_1^2 J.
    \end{align}
    Note that the matrices $I, J$, and $A_{G_p}$ are simultaneously diagonalizable, and the $3$ eigenspaces of $A_{G_p}$ are the common eigenspaces. Let $U_0, U_1$, and $U_2$ be the projection matrices to the $3$ eigenspaces of $A_{G_p}$ corresponding to eigenvalues $\frac{p-1}{2}, \frac{-1 + \sqrt{p}}{2}$, and $\frac{-1 - \sqrt{p}}{2}$. In particular, $U_0 = Q_0$ is the projection matrix to the span of constant vectors. Then,
    \begin{align*}
        H^{1,1} &= \alpha_1 I + \alpha_2 A_{G_p} - \alpha_1^2 J\\
        &= \left(\alpha_1 + \frac{p-1}{2}\alpha_2 - p\alpha_1^2\right)U_0 + \left(\alpha_1 + \frac{-1+\sqrt{p} }{2}\alpha_2 \right)U_1 + \left(\alpha_1 + \frac{-1-\sqrt{p} }{2}\alpha_2\right)U_2\\
        &\succeq \left(\alpha_1 + \frac{p-1}{2}\alpha_2 - p\alpha_1^2\right)U_0 + \left(\alpha_1 - o(\alpha_1)\right)(I - U_0), \numberthis
    \end{align*}
    where we used $\alpha_1 = \Theta(p^{-\frac{2}{3} })$ and $\sqrt{p}\cdot \alpha_2 = \Theta(p^{-\frac{5}{6} })$ in the last inequality. Replacing $U_0$ by $Q_0$ and $I - U_0$ by $Q_1$ shows the desired result.
\end{proof}

So, if moreover we can show $H^{2,2} - H^{2,1}(H^{1,1})^{-1} H^{1,2} \succeq 0$, we can conclude the positive semidefiniteness of $H$.
Our last simplification before proceeding to the main technical analysis is to remove the $(H^{1,1})^{-1}$ term above.
Fix some constant $\varepsilon > 0$ for all future discussions, say $\varepsilon \colonequals \frac{1}{2}$. Then, \begin{equation}
    H^{1,1} \succeq \left(\alpha_1 + \frac{p-1}{2}\alpha_2 - p\alpha_1^2\right)Q_0 + (1 - \varepsilon)\alpha_1 Q_1 \succ 0
\end{equation}
for all sufficiently large primes $p$, so
\begin{equation}
    (H^{1,1})^{-1} \preceq \left(\alpha_1 + \frac{p-1}{2}\alpha_2 - p\alpha_1^2\right)^{-1}Q_0 + \left((1 - \varepsilon)\alpha_1\right)^{-1} Q_1,
\end{equation}
and substituting this into the term appearing in the inequality we need to show,
\begin{equation}
    H^{2,1}(H^{1,1})^{-1} H^{1,2} \preceq H^{2,1}\left[\left(\alpha_1 + \frac{p-1}{2}\alpha_2 - p\alpha_1^2\right)^{-1}Q_0 + \left((1 - \varepsilon)\alpha_1\right)^{-1} Q_1\right] H^{1,2}.
\end{equation}
Note that the column sum (row sum) of $H^{2,1}$ is the same across each column (nonzero row indexed by edges of Paley graphs) due to the partial symmetry of Paley graphs. As a result, $\boldsymbol{1}$ is an eigenvector of $H^{2,1}H^{1,2}$, and $H^{2,1} Q_0 H^{1,2} = P_0 H^{2,1} H^{1,2} P_0$, where we use $P_0 = \frac{2}{p(p-1)} J \in \mathbb{R}^{\binom{\mathbb{F}_p }{2} \times \binom{\mathbb{F}_p }{2}}$ to denote the orthogonal projection matrix to the constant vector. Moreover, since $\boldsymbol{1}$ is an eigenvector of $H^{2,1}H^{1,2}$, $(I - P_0)H^{2,1}H^{1,2} P_0 = 0$. We therefore have
\begin{align*}
    &\hspace{-0.5cm}H^{2,1}\left[\left(\alpha_1 + \frac{p-1}{2}\alpha_2 - p\alpha_1^2\right)^{-1}Q_0 + \left((1 - \varepsilon)\alpha_1\right)^{-1} Q_1\right] H^{1,2}\\
    &= H^{2,1}\left[\left(\left(\alpha_1 + \frac{p-1}{2}\alpha_2 - p\alpha_1^2\right)^{-1} - \left((1 - \varepsilon)\alpha_1\right)^{-1} \right)Q_0 + \left((1 - \varepsilon)\alpha_1\right)^{-1} I\right] H^{1,2}\\
    &= \left(\left(\alpha_1 + \frac{p-1}{2}\alpha_2 - p\alpha_1^2\right)^{-1} - \left((1 - \varepsilon)\alpha_1\right)^{-1} \right) P_0 H^{2,1}H^{1,2} P_0 + \left((1 - \varepsilon)\alpha_1\right)^{-1} H^{2,1}H^{1,2}\\
    &= \left(\alpha_1 + \frac{p-1}{2}\alpha_2 - p\alpha_1^2\right)^{-1} P_0 H^{2,1}H^{1,2} P_0 + \left((1 - \varepsilon)\alpha_1\right)^{-1} (I - P_0)H^{2,1}H^{1,2}(I - P_0). \numberthis
\end{align*}
Thus, to show $H^{2,2} \succeq H^{2,1}(H^{1,1})^{-1} H^{1,2}$ holds for all sufficiently large primes $p$, it is sufficient to prove the following proposition:
\begin{proposition}\label{prop:main}
Under the FK pseudomoments specified by $\alpha_1, \alpha_2, \alpha_3, \alpha_4$ in Theorem \ref{thm:restated_main_theorem}, for any constant $\varepsilon > 0$,
\begin{align}
    H^{2,2} &\succeq \left(\alpha_1 + \frac{p-1}{2}\alpha_2 - p\alpha_1^2\right)^{-1} P_0 H^{2,1}H^{1,2} P_0 \nonumber \\
    &\hspace{1cm} + \left((1 - \varepsilon)\alpha_1\right)^{-1} (I - P_0)H^{2,1}H^{1,2}(I - P_0)
\end{align}
holds for all sufficiently large primes $p$.
\end{proposition}

\subsection{Ribbons and Graph Matrices}

To organize the remaining calculation, now let us review the construction of \emph{graph matrices} that has played a role in many SOS lower bound analyses in previous literature. We will use the following definitions as appeared in the work of \cite{jones2022sum}.
\begin{definition}[Ribbon]
    A ribbon on a ground set $V$ is a tuple $R = \left(V(R), E(R), A_R, B_R\right)$, where $(V(R), E(R))$ is a graph, and $A_R, B_R \subseteq  V(R) \subseteq V$.
\end{definition}

\begin{definition}[Matrix for a Ribbon]
    Let $G \in \mathbb{R}^{V \times V}$ be a real symmetric matrix whose off-diagonal entries are $\pm 1$ and whose diagonal entries are zero.
    For $R = \left(V(R), E(R), A_R, B_R\right)$ a on $V$, the corresponding matrix $M_G(R) \in \mathbb{R}^{\binom{V}{|A_R|} \times \binom{V}{|B_R|} }$ has rows and columns indexed by the subsets of $V$ of sizes $|A_R|$ and $|B_R|$, respectively. The entries of $M_G(R)$ is given by
    \begin{align}
        M_G(R)_{I,J} = \begin{cases}
            \prod_{\{i,j\} \in E(R)} G_{i,j} & \text{ if } I = A_R \text{ and } J = B_R \\
            0 & \text{ otherwise }
        \end{cases}.
    \end{align}
    In other words, there is only one nonzero entry of $M_G(R)$, and it is located at the row and the column corresponding to $A_R$ and $B_R$.
\end{definition}

\begin{definition}[Isomorphisms Between Ribbons]
    Two ribbons $R, S$ are isomorphic, or have the same shape, if there is a bijection $f: V(R) \to V(S)$ which is a graph isomorphism between $(V(R), E(R))$ and $(V(S), E(S))$ and also a bijection from $A_R$ to $A_S$ and from $B_R$ to $B_S$.
\end{definition}

\noindent
If we ignore the labels on the vertices of a ribbon, what remains is the shape of the ribbon.

\begin{definition}[Shape]
    \label{def:graph-matrix-shape}
    A \emph{shape} is an equivalence class of ribbons of the same shape. Each shape has associated with it a representative $\beta = \left(V(\beta), E(\beta), A_\beta, B_\beta\right)$.
\end{definition}

\begin{definition}[Embedding of a Shape]
    Given a shape $\beta$ on $V$ and an injective function $f: V(\beta) \to V$, we let $f(\beta)$ be the ribbon by labeling the vertices $V(\beta)$ in the natural way.
\end{definition}

\begin{definition}[Graph Matrix]
    Let $G \in \mathbb{R}^{V \times V}$ be a real symmetric matrix whose off-diagonal entries are $\pm 1$ and whose diagonal entries are zero. For a shape $\beta$ on $V$, the graph matrix $M_G(\beta) \in \mathbb{R}^{\binom{V}{|A_\beta|} \times \binom{V}{|B_\beta|} }$ is defined as the sum of all ribbon matrices over ribbons with shape $\beta$:
    \begin{align}
        M_G(\beta) = \sum_{R \text{ ribbon of shape } \beta} M_G\left(R\right).
    \end{align}
\end{definition}

\begin{definition}[Automorphism of a Shape]
    For a shape $\beta$, $\Aut(\beta)$ is the group of bijection from $V(\beta)$ to itself such that $A_\beta$ and $B_\beta$ are fixed as sets and the map is a graph automorphism of $(V(\beta), E(\beta))$.
\end{definition}

It is easy to see that if we sum over ribbon matrices  of all ribbons obtained from injective labelings of $\beta$, we obtain the graph matrix $M_G(\beta)$ multiplied by $|\Aut(\beta)|$. Thus,
\begin{align}
    M_G(\beta)  = \sum_{R \text{ ribbon of shape } \beta} M_G(R) = \frac{1}{|\Aut(\beta)|} \sum_{f: V(\beta) \to V \text{ injective}} M_G(f(\beta)).
\end{align}

\begin{definition}[Transpose]
    Given a ribbon $R$ or shape $\beta$, we define its transpose by swapping the two parts $A_R$ and $B_R$ (resp.~$A_\beta$ and $B_\beta$). Observe that this transposes the matrix for the ribbon/shape.
\end{definition}

\begin{definition}[Trivial Shape]
    A shape $\beta$ is trivial if $V(\beta) = A_\beta = B_\beta$ and $E(\beta) = \emptyset$. Observe that the graph matrix for a trivial shape $\beta$ is the identity matrix of dimension specified by $|V(\beta)|$.
\end{definition}

\subsection{Graph Matrix Decomposition}
\label{sec:graph-mx-decomp}

\begin{definition}[Legendre Symbol]
    Let $\mathbb{F}_p$ be the finite field of order $p$. The Legendre symbol is defined as
    \begin{equation}
        \chi(a) = \chi_p(a) \colonequals \begin{cases}
            0 & \text{ if } a \equiv 0 \Mod{p}, \\
            1 & \text{ if } a \text{ is a quadratic residue in } \mathbb{F}_p, \\
            -1 & \text{ if } a \text{ is a quadratic nonresidue in } \mathbb{F}_p.
        \end{cases}
    \end{equation}
    When the underlying finite field $\mathbb{F}_p$ is fixed and clear from context, we will omit the subscript $p$.
\end{definition}

\begin{remark}
    Recall that all the primes $p$ in our discussion are congruent to 1 modulo 4. This ensures that $\chi(-1) = 1$, and thus $\chi(a) = \chi(-a)$ for any $a \in \mathbb{F}_p$.
\end{remark}

\begin{proposition}\label{prop:indicator}
    We have $\One_{\ell,r}(L, R) = \frac{1}{2^{|L \setminus R| \times |R \setminus L|}} \prod_{(a,b) \in (L \setminus R) \times (R \setminus L)}(1 + \chi(a-b))$ for all $\ell, r \geq 0$, $L \in \binom{\FF_p}{\ell}$, and $R \in\binom{\FF_p}{r}$.
\end{proposition}
\begin{proof}
    The result follows from observing that, for $a, b \in \FF_p$ distinct, $\frac{1}{2}(1 + \chi(a-b))$ is the indicator of the edge $\{a,b\}$ existing in the Paley graph.
\end{proof}

In the following few equations, let us write $S$ for the Seidel adjacency matrix of $G_p$, so that $S_{a, b} \colonequals \chi(a - b)$.
By substituting the indicator functions $\One_{\ell,r}$ in the definition of $H$ using Proposition~\ref{prop:indicator} and expanding the products, we have
\begin{align}
    H^{2,2}_{\{a,b\}, \{c,d\}} = \begin{cases}
        \begin{array}{l}
             \alpha_2 - \alpha_2^2
        \end{array} & \begin{array}{l}
              \text{ if } \{a,b\} = \{c,d\},
        \end{array}\\[1em]
        \begin{array}{l}
            \left(\frac{\alpha_3}{2} -\alpha_2^2\right) + \frac{\alpha_3}{2} S_{b,d}
        \end{array}
         &  \begin{array}{l}
             \text{ if } a = c \text{ and } b \ne d,
         \end{array}\\[1em]
         \begin{array}{l}
               (\frac{\alpha_4}{16} - \alpha_2^2)+\frac{\alpha_4}{16}\big(S_{a,c} + S_{a,d} + S_{b, c} + S_{b,d} \\
               \hspace{6.25em}+ S_{a,c}S_{a,d} + S_{b,c}S_{b,d} + S_{a,c}S_{b,c} \\
              \hspace{6.25em}+ S_{a,d}S_{b,d} + S_{a,c}S_{b,d} + S_{a,d}S_{b,c} \\
              \hspace{6.25em}+ S_{a,c}S_{a,d}S_{b,c} + S_{b,d}S_{a,d}S_{b,c} \\
              \hspace{6.25em}+ S_{a,c}S_{a,d}S_{b,d} + S_{a,c}S_{b,c}S_{b,d} \\
              \hspace{6.25em}+ S_{a,c}S_{a,d}S_{b,c}S_{b,d}\big)
         \end{array} &
         \begin{array}{l}
              \text{ if } \{a,b\} \cap \{c,d\} = \emptyset.
         \end{array}
    \end{cases}
\end{align}
and
\begin{align*}
    &\quad (H^{2,1}H^{1,2})_{\{a,b\}, \{c,d\}}\\
    &= \sum_{i \in \mathbb{F}_p} H^{2,1}_{\{a,b\}, i} H^{1,2}_{i, \{c,d\}}\\
    &= \begin{cases}
        \begin{array}{l}
              2(\alpha_2 - \alpha_1\alpha_2)^2 + (p-2)((\alpha_1\alpha_2)^2 +  \frac{\alpha_3^2}{4} -\frac{\alpha_1\alpha_2\alpha_3}{2}) \\
              \hspace{1em} + (\frac{\alpha_3^2}{4} -\frac{\alpha_1\alpha_2\alpha_3}{2})\sum_{i \in \mathbb{F}_p \setminus \{a,b\}}(S_{a,i} + S_{b,i} + S_{a,i}S_{b,i})
         \end{array} & \begin{array}{l}
              \text{ if } \{a,b\} = \{c,d\},
         \end{array} \\[2em]
         \begin{array}{l}
              (\alpha_2 - \alpha_1\alpha_2)^2 - 2(\alpha_2 - \alpha_1\alpha_2)\alpha_1\alpha_2+ (p-3)(\alpha_1\alpha_2)^2 \\
              \hspace{1em} + \frac{(\alpha_2 - \alpha_1\alpha_2)\alpha_3}{2} - (p-3)\frac{\alpha_1\alpha_2\alpha_3}{2} + (p-3)\frac{\alpha_3^2}{8} \\
              \hspace{1em}+ \frac{(\alpha_2 - \alpha_1\alpha_2)\alpha_3}{4}(S_{a,b}S_{b,d} + S_{a,d}S_{b,d} + S_{a,b} + S_{a,d} + 2S_{b,d}) \\
              \hspace{1em} + (\frac{\alpha_3^2}{8} - \frac{\alpha_1\alpha_2\alpha_3 }{2})\sum_{i \in \mathbb{F}_p \setminus\{a,b,d\} } S_{a,i} \\
              \hspace{1em}+ (\frac{\alpha_3^2 }{8} - \frac{\alpha_1\alpha_2\alpha_3 }{4})\sum_{i \in \mathbb{F}_p \setminus\{a,b,d\} } (S_{b,i} + S_{d,i}
              + S_{a,i}S_{b,i} + S_{a,i}S_{d,i})\\
              \hspace{1em} + \frac{\alpha_3^2 }{8}\sum_{i \in \mathbb{F}_p \setminus\{a,b,d\} }(S_{b,i}S_{d,i} + S_{a,i}S_{b,i}S_{d,i})
         \end{array} & \begin{array}{l}
              \text{ if } a = c \text{ and } b \ne d,
         \end{array} \vspace{1em}\\
         \begin{array}{l}
              (\alpha_2 - \alpha_1\alpha_2)\alpha_3 - 4(\alpha_2 - \alpha_1\alpha_2)\alpha_1\alpha_2\\
              \hspace{1em}+ (p-4)(\alpha_1\alpha_2)^2 - (p-4)\frac{\alpha_1\alpha_2\alpha_3 }{2}+ (p-4)\frac{\alpha_3^2 }{16}\\
              \hspace{1em} + \frac{(\alpha_2- \alpha_1\alpha_2)\alpha_3}{2}(S_{a,c} + S_{a,d} + S_{b,c} + S_{b,d}) \\
              \hspace{1em} + \frac{(\alpha_2 - \alpha_1\alpha_2)\alpha_3}{4}(S_{a,c}S_{a,d} + S_{b,c}S_{b,d} + S_{a,c}S_{b,c} + S_{a,d}S_{b,d})\\
              \hspace{1em} + (\frac{\alpha_3^2 }{16} - \frac{\alpha_1\alpha_2\alpha_3 }{4}) \sum_{i \in \mathbb{F}_p \setminus\{a,b, c,d\} }(S_{a,i} + S_{b,i} + S_{c,i} + S_{d,i} \\
              \hspace{13.5em} + S_{a,i}S_{b,i} + S_{c,i}S_{d,i})\\
              \hspace{1em} + \frac{\alpha_3^2 }{16} \sum_{i \in \mathbb{F}_p \setminus\{a,b, c,d\} }( S_{a,i}S_{c,i} + S_{a,i}S_{d,i}
              + S_{b,i}S_{c,i} + S_{b,i}S_{d,i} \\
              \hspace{8.6em} +S_{a,i}S_{b,i}S_{c,i} + S_{a,i}S_{b,i}S_{d,i} \\
              \hspace{8.6em}+ S_{a,i}S_{c,i}S_{d,i} + S_{b,i}S_{c,i}S_{d,i}\\
              \hspace{8.6em}+ S_{a,i}S_{b,i}S_{c,i}S_{d,i})
         \end{array} & \begin{array}{l}
              \text{ if } \{a,b\} \cap \{c,d\} = \emptyset.
         \end{array} \numberthis
    \end{cases}
\end{align*}

\noindent
We now express this as a sum of graph matrices.
We present all the matrices required for this decomposition in Table~\ref{table:graph_matrices}.
Using the notations for graph matrices defined above and in the table, we can write the matrix $H^{2,2}$ and the matrix $H^{2,1}H^{1,2}$ as a weighted sum of these matrices, as follows:
\begin{align*}
    H^{2,2} &= (\alpha_2 - \alpha_2^2)I + \left(\frac{\alpha_3}{2} - \alpha_2^2\right) T^{3,0,1} + \frac{\alpha_3}{2} T^{3,1,1} + \left(\frac{\alpha_4}{16} - \alpha_2^2\right) T^{4,0,1}\\
    &\hspace{1.25cm} + \frac{\alpha_4 }{16} \left(T^{4,1,1} + T^{4,2,1} + T^{4,2,2} + T^{4,2,3} + T^{4,3,1} + T^{4,4,1} \right), \numberthis \label{eq:H22-graph-mx}\\
    H^{2,1}H^{1,2} &= \left[2(\alpha_2 - \alpha_1\alpha_2)^2 + (p-2)\left((\alpha_1\alpha_2)^2 + \frac{\alpha_3^2 }{4} - \frac{\alpha_1\alpha_2\alpha_3 }{2}\right)\right] I \\
    &\hspace{1.25cm}+ \left(\frac{\alpha_3^2 }{4} - \frac{\alpha_1\alpha_2\alpha_3 }{2}\right)\left(U^{3,1,1} + U^{3,2,1}\right)\\
    &\hspace{1.25cm}+ \left[ (\alpha_2 - \alpha_1\alpha_2)\left(\alpha_2 - 3\alpha_1\alpha_2 + \frac{\alpha_3}{2}\right) + (p-3)\left((\alpha_1\alpha_2)^2 - \frac{\alpha_1\alpha_2\alpha_3 }{2} + \frac{\alpha_3^2 }{8} \right)\right] T^{3,0,1}\\
    &\hspace{1.25cm}+ \frac{(\alpha_2 - \alpha_1\alpha_2)\alpha_3}{4} (T^{3,2,1} + T^{3,2,2} + 2T^{3,1,1} + T^{3,1,2} + T^{3,1,3}) \\
    &\hspace{1.25cm} + \left(\frac{\alpha_3^2 }{8} - \frac{\alpha_1\alpha_2\alpha_3 }{2} \right) U^{4,1,1}\\
    &\hspace{1.25cm}+ \left(\frac{\alpha_3^2 }{8} - \frac{\alpha_1\alpha_2\alpha_3 }{4} \right)\left(U^{4,1,2} + U^{4,1,3} + U^{4,2,1} + U^{4,2,2}\right) + \frac{\alpha_3^2}{8}\left(U^{4,2,3} + U^{4,3,1}\right)\\
    &\hspace{1.25cm}+ \left[(\alpha_2 - \alpha_1\alpha_2)(\alpha_3 - 4\alpha_1\alpha_2) + (p-4)\left(\alpha_1\alpha_2 - \frac{\alpha_3 }{4}\right)^2\right] T^{4,0,1}\\
    &\hspace{1.25cm}+ \frac{(\alpha_2 - \alpha_1\alpha_2)\alpha_3 }{2} T^{4,1,1} + \frac{(\alpha_2 - \alpha_1\alpha_2)\alpha_3}{4} (T^{4,2,1} + T^{4,2,2})\\
    &\hspace{1.25cm} + \left(\frac{\alpha_3^2 }{16} - \frac{\alpha_1\alpha_2\alpha_3 }{4}\right) \left(U^{5,1,1}+ U^{5,1,2} + U^{5,2,1} + U^{5,2,2} \right)\\
    &\hspace{1.25cm}+ \frac{\alpha_3^2}{16} \left(U^{5,2,3} + U^{5,3,1} + U^{5,3,2} + U^{5,4,1}\right). \numberthis \label{eq:H21H12-graph-mx}
\end{align*}

\input{graph-matrix-table}

\subsection{Graph Matrix Norm Bounds}\label{sec:graph_matrix_norm_bounds}
Now we analyze the norms of the graph matrices defined above in order to prove Proposition~\ref{prop:main}.

\begin{remark}
    Previous work of \cite{AMP-2016-GraphMatrices} established the typical norm of graph matrices when the underlying matrix $G$ is the Seidel adjacency matrix of an ER random graph from $\mathcal{G}(n, \frac{1}{2})$, where the quantities that characterize the norm bounds are the sizes of the minimum vertex separators of the shapes. In this work, using different techniques, we prove graph matrix norm bounds when the underlying matrix is the Seidel adjacency matrix of the Paley graph $G_p$.
\end{remark}

Recall that we defined $P_0 = \frac{1}{p(p-1)} J \in \mathbb{R}^{\binom{\mathbb{F}_p}{2} \times \binom{\mathbb{F}_p}{2} }$ to denote the orthogonal projection matrix to the constant vector. Following the strategies in \cite{DM-2015-SOSPlantedClique}, we define the following subspaces of $\mathbb{R}^{\binom{\mathbb{F}_p}{2} }$:
\begin{align}
    \mathbb{V}_0 &= \left\{v \in \mathbb{R}^{\binom{\mathbb{F}_p}{2} }: v_{i,j} = v_{i', j'}, \quad \forall \{i,j\}, \{i',j'\} \in \binom{\mathbb{F}_p}{2}\right\} \\
    \mathbb{V}_1 &= \left\{v \in \mathbb{R}^{\binom{\mathbb{F}_p}{2}}: \exists u \in \mathbb{R}^{\mathbb{F}_p}, \text{s.t. } \langle \boldsymbol{1}, u \rangle = 0 \text{ and } v_{\{i,j\}} = u_i + u_j, \quad \forall \{i,j\} \in \binom{\mathbb{F}_p}{2} \right\}\\
    \mathbb{V}_2 &= (\mathbb{V}_0 \oplus \mathbb{V}_1)^\perp.
\end{align}
In words, $\mathbb{V}_0$ is the span of constant vectors, $\mathbb{V}_0 \oplus \mathbb{V}_1$ is the span of vectors $v$ whose entries $v_{\{i,j\}}$ can be decomposed to a sum of $u_i + u_j$ for some $u \in \mathbb{R}^{\mathbb{F}_p }$, and $\mathbb{V}_2$ is the orthogonal complement of $\mathbb{V}_0 \oplus \mathbb{V}_1$.
Furthermore, let $P_1$ and $P_2$ be the orthogonal projection matrices to the subspaces $\mathbb{V}_1$, and $\mathbb{V}_2$ respectively. Note that this is consistent with the previously defined $P_0$, which is the orthogonal projection matrix to the span of constant vectors $\mathbb{V}_0$.

In the analysis of ER graphs, these subspaces appear because they are the decomposition of $\RR^{\binom{\FF_p}{2}}$ into irreducible subrepresentations under the action of $S_p$, with respect to which the expectation of an FK pseudomoment matrix is invariant.
This invariance does not hold for our deterministic FK pseudomoment matrix, but we will see that the same decomposition is still useful.

We will use the following norm bounds for the graph matrices defined earlier.
We defer the proofs of these statements to Section~\ref{sec:graph-matrix-proofs}.

\begin{proposition}\label{prop:T_32i}
    $\|T^{3,2,i}\| = O(\sqrt{p})$ for $i \in \{1,2\}$.
\end{proposition}

\begin{proposition}\label{prop:T_311}
    $\|T^{3,1,1}\| = O(\sqrt{p})$.
\end{proposition}

\begin{proposition}\label{prop:T_31i}
    $\|T^{3,1,i}\| = O(p)$ for $i \in \{2,3\}$.
\end{proposition}

\begin{proposition}
    \label{prop:T301}
    $T^{3,0,1} = 2(p-2)P_0 + (p-4)P_1 - 2P_2.$
\end{proposition}

\begin{proposition}\label{prop:T_431}
    $\|T^{4,3,1}\| = O(p).$
\end{proposition}

\begin{proposition}\label{prop:T_4ij}
    $\|T^{4,i,j}\| = O(p^{3/2})$ for $(i, j) \in \{(2, 1), (2, 2), (1, 1)\}$.
    Moreover, all of $\| T^{4,2,1} P_2\|$, $\|P_2 T^{4,2,2} \|$, $\|P_2 T^{4,1,1}\|$, and $\|T^{4,1,1}P_2\|$ are $O(\sqrt{p})$.
\end{proposition}

\begin{proposition}
    \label{prop:T423}
    $\|T^{4,2,3}\| = O(p).$
\end{proposition}

\begin{proposition}
    \label{prop:T401}
    $T^{4,0,1} = \frac{(p-2)(p-3)}{2}P_0 - (p-3)P_1 + P_2.$
\end{proposition}

\begin{proposition}
    \label{prop:U3i1}
    $\|U^{3, i, 1}\| = O(1)$ for $i \in \{1, 2\}$.
\end{proposition}

\begin{proposition}
    \label{prop:U_431}
    $\|U^{4,3,1}\| = O(p^{3/2})$.
\end{proposition}

\begin{proposition}
    \label{prop:U4ij}
    $\|U^{4, i, j}\| = O(p)$ for $i \in \{1, 2\}$ and $j \in \{1, 2, 3\}$.
\end{proposition}

\begin{proposition}
    \label{prop:U541}
    $\|U^{5,4,1}\| = O(p^2)$.
\end{proposition}

\begin{proposition}
    \label{prop:U_53i}
    $\|U^{5, 3, i}\| = O(p^2)$ for $i \in \{1, 2\}$.
\end{proposition}

\begin{proposition}
    \label{prop:U5ij}
    $\|U^{5, i, j}\| = O(p^2)$ for $i \in \{1, 2\}$ and $j \in \{1, 2, 3\}$, where $j \neq 3$ if $i = 1$.
\end{proposition}

\begin{theorem}\label{thm:norm_T_441}
    $\|T^{4,4,1}\| = O(p^{5/4})$.
\end{theorem}

\noindent
Of these statements, Theorem~\ref{thm:norm_T_441} is by far the subtlest---unlike the other terms, where fairly straightforward arguments work, for $T^{4,4,1}$ it turns out that a naive bound is insufficient, and we must more carefully account for character sum cancellations.
Our bound is adequate for our purposes, but we believe it is not tight; see Remark~\ref{rem:T441-tightness} for further discussion.
We note also that the bounds we prove are generally incomparable to those for random graphs following from \cite{AMP-2016-GraphMatrices}: for some graph matrices we expect a comparable norm bound but cannot prove one due to technical obstacles, while for other graph matrices the Paley graph exhibits stronger cancellations than a random graph and we can show a stronger norm bound. We compare the respective bounds in Table~\ref{table:graph_matrices}.
Moreover, as we show in Section~\ref{sec:failure-graphical}, there is an example of a graph matrix for which the norm when evaluated on the Paley graph is actually asymptotically larger than the norm when evaluated on a random graph; however, this example does not figure in our analysis.

\subsection{Final Steps}

Finally, putting all the graph matrix norm bounds together, we prove Proposition \ref{prop:main}, which will conclude the proof of Theorem~\ref{thm:restated_main_theorem}, as we have discussed earlier.

\begin{proof}[Proof of Proposition \ref{prop:main}]
    The statements in this proof will hold for all sufficiently large primes $p$.

    To show $H^{2,2} \succeq (\alpha_1 + \frac{p-1}{2}\alpha_2 - p\alpha_1^2)^{-1} P_0 H^{2,1}H^{1,2} P_0 + ((1 - \varepsilon)\alpha_1)^{-1} (I - P_0)H^{2,1}H^{1,2}(I - P_0)$, we have to show
    \begin{align*}
        &(\alpha_2 - \alpha_2^2)I + \left(\frac{\alpha_3}{2} - \alpha_2^2 \right)T^{3,0,1} + \left(\frac{\alpha_4}{16} - \alpha_2^2 \right) T^{4,0,1}\\
        &\quad - \left(\alpha_1 + \frac{p-1}{2}\alpha_2 - p\alpha_1^2\right)^{-1}P_0\Bigg( \left[2(\alpha_2 - \alpha_1\alpha_2)^2 + (p-2)\left((\alpha_1\alpha_2)^2 + \frac{\alpha_3^2 }{4} - \frac{\alpha_1\alpha_2\alpha_3 }{2}\right)\right] I\\
        &\quad + \left[ (\alpha_2 - \alpha_1\alpha_2)\left(\alpha_2 - 3\alpha_1\alpha_2 + \frac{\alpha_3}{2}\right) + (p-3)\left((\alpha_1\alpha_2)^2 - \frac{\alpha_1\alpha_2\alpha_3 }{2} + \frac{\alpha_3^2 }{8} \right)\right] T^{3,0,1}\\
        &\quad + \left[(\alpha_2 - \alpha_1\alpha_2)(\alpha_3 - 4\alpha_1\alpha_2) + (p-4)\left(\alpha_1\alpha_2 - \frac{\alpha_3 }{4}\right)^2\right] T^{4,0,1} \Bigg)P_0\\
        &\quad - \left((1 - \varepsilon)\alpha_1\right)^{-1}(I - P_0)\Bigg( \left[2(\alpha_2 - \alpha_1\alpha_2)^2 + (p-2)\left((\alpha_1\alpha_2)^2 + \frac{\alpha_3^2 }{4} - \frac{\alpha_1\alpha_2\alpha_3 }{2}\right)\right] I\\
        &\quad + \left[ (\alpha_2 - \alpha_1\alpha_2)\left(\alpha_2 - 3\alpha_1\alpha_2 + \frac{\alpha_3}{2}\right) + (p-3)\left((\alpha_1\alpha_2)^2 - \frac{\alpha_1\alpha_2\alpha_3 }{2} + \frac{\alpha_3^2 }{8} \right)\right] T^{3,0,1}\\
        &\quad + \left[(\alpha_2 - \alpha_1\alpha_2)(\alpha_3 - 4\alpha_1\alpha_2) + (p-4)\left(\alpha_1\alpha_2 - \frac{\alpha_3 }{4}\right)^2\right] T^{4,0,1} \Bigg)(I - P_0) \stackrel{?}{\succeq} M, \numberthis
    \end{align*}
    where $M$ is the sum of the remaining graph matrices of shapes having at least one edge among those appearing in the expressions \eqref{eq:H22-graph-mx} and \eqref{eq:H21H12-graph-mx}.
    Note that $\mathbb{V}_0, \mathbb{V}_1, \mathbb{V}_2$ are eigenspaces of the left hand side, with eigenvalues
    \begin{align}
        4\alpha_1^2  + 2(p-2)(4\alpha_1^3) + \frac{(p-2)(p-3)}{2} (32 \alpha_1^4 - 16\alpha_1^4) - 2p^2\alpha_1^4 - O(p^2 \alpha_1^5) &= (1 - o(1)) 6p^2 \alpha_1^4, \\
        4\alpha_1^2 + (p-4)(4\alpha_1^3) - (p-3)(32 \alpha_1^4 - 16\alpha_1^4) - O(p \alpha_1^4 ) &= (1 - o(1)) 4p\alpha_1^3, \\
        4\alpha_1^2 - 8\alpha_1^3 + (32\alpha_1^4 - 16\alpha_1^4) - O(\alpha_1^3) &= (1 - o(1))4\alpha_1^2
    \end{align}
    respectively.

    It is then sufficient to show
    \begin{align}
        \begin{bmatrix}
            3p^2 \alpha_1^4 & 0 & 0\\
            0 & 2p\alpha_1^3 & 0\\
            0 & 0 & 2\alpha_1^2
        \end{bmatrix} \stackrel{?}{\succeq} \begin{bmatrix}
            \|P_0 M P_0\| & \|P_0 M P_1\| & \|P_0 M P_2\|\\
            \|P_1 M P_0\| & \|P_1 M P_1\| & \|P_1 M P_2\|\\
            \|P_2 M P_0\| & \|P_2 M P_1\| & \|P_2 M P_2\|
        \end{bmatrix}.
    \end{align}
    Using the graph matrix norm bounds above, we have for any $i \in \{0,1,2\}$ and $j \in \{0,1,2\}$ with $(i,j) \ne (2,2)$ that
    \begin{align}
        \|P_i M P_j\| &= O(p^{3/2} \alpha_1^4),
        \intertext{and for the remaining case}
        \|P_2 M P_2\| &= O(p^2 \alpha_1^5),
    \end{align}
    so we only need to prove that the following matrix is positive semidefinite:
    \begin{align}
        \begin{bmatrix}
            3p^2 \alpha_1^4 - O(p^{3/2} \alpha_1^4) & - O(p^{3/2} \alpha_1^4) & - O(p^{3/2} \alpha_1^4)\\
            - O(p^{3/2} \alpha_1^4) & 2p\alpha_1^3 - O(p^{3/2} \alpha_1^4) & - O(p^{3/2} \alpha_1^4)\\
            - O(p^{3/2} \alpha_1^4) & - O(p^{3/2} \alpha_1^4) & 2\alpha_1^2 - O(p^2 \alpha_1^5)
        \end{bmatrix}.
    \end{align}
    When $\alpha_1 = c \cdot p^{- 2/3}$ for a sufficiently small constant $c$, we have
    \begin{align}
        2\alpha_1^2 - O(p^2 \alpha_1^5) = \Omega(\alpha_1^2).
    \end{align}
    Taking the Schur complement with respect to the bottom right block $2\alpha_1^2 - O(p^2 \alpha_1^5)$, it is sufficient to prove that the following matrix is positive-semidefinite
    \begin{align*}
        &\quad \begin{bmatrix}
            3p^2 \alpha_1^4 - O(p^{3/2} \alpha_1^4) - O\left(\frac{(p^{3/2} \alpha_1^4)^2}{\alpha_1^2}\right) & - O(p^{3/2} \alpha_1^4) - O\left(\frac{(p^{3/2} \alpha_1^4)^2}{\alpha_1^2}\right)\\
            - O(p^{3/2} \alpha_1^4) - O\left(\frac{(p^{3/2} \alpha_1^4)^2}{\alpha_1^2}\right) & 2p\alpha_1^3 - O(p^{3/2} \alpha_1^4) - O\left(\frac{(p^{3/2} \alpha_1^4)^2}{\alpha_1^2}\right)
        \end{bmatrix}\\
        &= \begin{bmatrix}
            3p^2 \alpha_1^4 - O(p^3 \alpha_1^6)& - O(p^3 \alpha_1^6)\\
            - O(p^3 \alpha_1^6) & 2p\alpha_1^3 - O(p^3 \alpha_1^6) \numberthis
        \end{bmatrix},
    \end{align*}
    which is diagonally dominant under our choice of $\alpha_i$.
\end{proof}
\noindent
With Proposition \ref{prop:main} proved, we have finished proving Theorem \ref{thm:restated_main_theorem}.

\section{Proofs of Graph Matrix Norm Bounds}
\label{sec:graph-matrix-proofs}

\subsection{Tools from Number Theory}

Before we start proving the graph matrix norm bounds we need, we first set up some basic number theoretic notations and results that we will use in the following proofs.

\begin{definition}[Character]
    A character $\theta$ of a group $G$ is a homomorphism $\theta: G \to \mathbb{C}^\times$, i.e., a function satisfying $\theta(ab) = \theta(a)\theta(b)$ for all $a, b, \in G$.
\end{definition}
\noindent
It is easy to check that if $\theta$ is a character then so is its conjugate $\overline{\theta}(g) \colonequals \overline{\theta(g)}$.

\begin{proposition}
    Let $\theta_1, \theta_2$ be distinct characters of a group $G$. Then, $\theta_1$ and $\theta_2$ are orthogonal with respect to the inner product
    \begin{align}
        \langle \theta_1, \theta_2\rangle \colonequals \frac{1}{|G|} \sum_{g\in G} \theta_1(g) \overline{\theta_2}(g).
    \end{align}
\end{proposition}

\begin{remark}
    For a field $\mathbb{F}$, we usually say a multiplicative character $\phi$ of $\mathbb{F}$ to mean a character $\phi: \mathbb{F}^\times \to \mathbb{C}^\times$ of the multiplicative group $\mathbb{F}^\times$ of $\mathbb{F}$, and an additive character $\psi$ of $\mathbb{F}$ to mean a character $\psi: \mathbb{F} \to \mathbb{C}^\times$ of the additive group of $\mathbb{F}$. As a convention, we extend the definition of a multiplicative character $\phi$ to all of $\mathbb{F}$ by setting $\phi(0) = 0$. We remark that this convention is consistent with the definition of Legendre symbol, which is a multiplicative character of $\mathbb{F}_p$.
\end{remark}

As before, $\chi = \chi_p$ denotes the Legendre symbol for $\mathbb{F}_p$. We will also use $e_p: \mathbb{F}_p \to \mathbb{C}$ to denote the following generator of the additive characters of $\mathbb{F}_p$:
\begin{align}
    e_p(x) \colonequals \exp\left(\frac{2\pi i}{p} x\right).
\end{align}

Much of analytic number theory is concerned with establishing bounds on \emph{character sums}, sums over finite fields or subsets thereof of characters subject to various transformations of their inputs.
Usually, the best such estimates establish cancellations in a character sum that are comparable to the cancellations that would occur if the values of a character were random.
The following classical result is one of the main bounds of this kind.

\begin{theorem}[Weil's bound, Chapter 11 of \cite{IK-2021-AnalyticNumberTheory}]
    \label{thm:weil}
    Let $f \in \mathbb{F}_p[t]$ be a polynomial of degree $d$. Suppose that $f$ cannot be represented as $f(t) = r\cdot g(t)^2$ for $r \in \mathbb{F}_p$ and $g \in \mathbb{F}_p[t]$. Then,
    \begin{align}
        \left| \sum_{x \in \mathbb{F}_p} \chi(f(x)) \right| \le d\sqrt{p}.
    \end{align}
\end{theorem}
\noindent
We note that if we believe this character sum to behave like a sum of random $\pm 1$ signs, we indeed should expect the sum to be of order $O(\sqrt{p})$.

We next give some simple algebraic properties of the Legendre symbol.
These may be viewed as special cases of Theorem~\ref{thm:weil}, where we may obtain more precise evaluations of character sums.

\begin{proposition}
    \label{prop:chi-sum-1}
    For any $a \in \FF_p$,
    \begin{equation}
        \sum_{x \in \FF_p} \chi(a - x) = 0.
    \end{equation}
\end{proposition}
\begin{proof}
    This follows from the fact that exactly half of the elements of $\FF_p^{\times}$ are quadratic residues.
\end{proof}

\begin{proposition}
    \label{prop:chi-sum-2}
    For any $a, b \in \FF_p$,
    \begin{equation}
        \sum_{x \in \FF_p} \chi(a - x) \chi(b - x) = -1 + \One\{a = b\}\, p.
    \end{equation}
    In terms of the Seidel adjacency matrix,
    \begin{equation}
        S_{G_p}^2 = p I - J.
    \end{equation}
\end{proposition}
\begin{proof}
    The case $a = b$ is immediate.
    Otherwise, applying the change of variables $x \leftarrow (b - a)x + a$ gives
    \begin{align*}
        \sum_{x \in \FF_p} \chi(a - x) \chi(b - x)
        &= \sum_{x \in \FF_p} \chi(x(x + 1)) \\
        &= \sum_{x \in \FF_p^{\times}} \chi(x^{-1})\chi(x + 1) \\
        &= \sum_{x \in \FF_p^{\times}} \chi(1 + x^{-1}) \\
        &= \sum_{x \in \FF_p^{\times}} \chi(1 + x) \\
        &= -1 + \sum_{x \in \FF_p} \chi(x) \\
        &= -1, \numberthis
    \end{align*}
    as claimed.
\end{proof}

We will also need to make use of some more sophisticated families of character sums.
\begin{definition}[Gauss sum]
    For $\phi$ a multiplicative character, the associated \emph{Gauss sum} is
    \begin{equation}
        G(\phi) \colonequals \sum_{x \in \FF_p} \phi(x) e_p(x).
    \end{equation}
\end{definition}

The following expresses that the Gauss sum computes the additive Fourier transform of a multiplicative character.
\begin{proposition}
    \label{prop:mult-add-fourier}
    For $\phi$ a multiplicative character and $a \in \FF_p$,
    \begin{equation}
        \phi(a) = \frac{G(\phi)}{p} \sum_{b \in \FF_p} \overline{\phi}(b)e_p(-ab) = \frac{G(\phi)}{p}  \phi(-1)\sum_{b \in \FF_p} \overline{\phi}(b)e_p(ab)
    \end{equation}
\end{proposition}
\begin{proof}
    The additive Fourier transform of the character $\phi$ extended to $\FF_p$ by $\phi(0) = 0$ is a function on additive characters $\psi(x) = e_p(tx)$, given when $t \neq 0$ by
    \begin{equation}
        \sum_{b \in \FF_p} \phi(b) e_p(tb) = \sum_{b \in \FF_p} \phi(t^{-1}b) e_p(b) = G(\phi) \overline{\phi}(t).
    \end{equation}
    The result then follows by Fourier inversion.
\end{proof}

\begin{proposition}[Section 6.4 and Proposition 8.2.2 of \cite{IR-1990-ClassicalModernNumberTheory}]
    \label{prop:gauss-sum-norm}
    For any non-trivial multiplicative character $\phi$, $|G(\phi)| = \sqrt{p}$.
    Moreover, when $p \equiv 1 \Mod{4}$, $G(\chi) = \sqrt{p}$.
\end{proposition}

\begin{definition}[Kloosterman sum]
    For $k \geq 2$ and $a \in \FF_p^{\times}$, we define the \emph{Kloosterman sum}
    \begin{equation}
        K_k(a) \colonequals \sum_{\substack{x_1, \dots, x_k \in \FF_p^{\times} \\ x_1 \cdots x_k = a}} e_p(x_1 + \cdots + x_k).
    \end{equation}
    We also write $K(a) \colonequals K_2(a)$.
\end{definition}

\begin{proposition}[Equation 14.7 of \cite{CI-2000-CubicMomentLFunctions}]
    \label{prop:kloosterman-rewrite}
    For $a \in \FF_p^{\times}$,
    \begin{equation}
        \sum_{x \in \FF_p} \chi(x^2 - 1) e_p(2ax) = K(a^2).
    \end{equation}
\end{proposition}
\begin{proposition}[\cite{Liu-2002-KloostermanTwistedMoments}]
    \label{prop:twisted-kloosterman}
    For any non-trivial multiplicative character $\phi$,
    \begin{equation}
    \left|\sum_a \phi(a) K(a^2)^2\right| \leq 2 p^{3/2}.
    \end{equation}
\end{proposition}

\begin{proposition}[Lemma 2.1 of \cite{LZZ-2018-DistributionJacobiSums}]
    \label{prop:twisted-kloosterman-2}
    For any $s, t \geq 1$ and $k \geq 2$, there is an absolute constant $C_{k, s, t} > 0$ such that, for any non-trivial multiplicative character $\phi$,
    \begin{equation}
    \left|\sum_a \phi(a) K_k(a)^s \overline{K_k(a)}^{t}\right| \leq C_{k, s, t} p^{((k - 1)(s + t) + 1) / 2}.
    \end{equation}
\end{proposition}

\begin{proposition}[Corollary 3.2 of \cite{FKM-2015-SumsOfProducts}]
    \label{prop:kloosterman-corr}
    Let $k \geq 2$ be even, $t \geq 1$, and let $a_1, \dots, a_t \in \FF_p^{\times}$ have some $a_i$ occur an odd number of times.
    Then, for an absolute constant $C_{k, t} > 0$,
    \begin{equation}
        \left|\sum_{x \in \FF_p} K_k(a_1x) \cdots K_k(a_tx)\right| \leq C_{k, t} p^{(t(k - 1) + 1) / 2}.
    \end{equation}
\end{proposition}

\subsection{Tools from Linear Algebra}

We will also repeatedly use the following simple but powerful bound on matrix norms, which applies effectively to sparse matrices with entries of small magnitude.

\begin{proposition}[Asymmetric Gershgorin bound]\label{prop:Gershgorin}
    Let $X \in \RR^{m \times n}$.
    Then,
    \begin{equation}
        \|X\| \leq \sqrt{\left(\max_{i = 1}^m \sum_{j = 1}^n |X_{ij}|\right)\left(\max_{j = 1}^n \sum_{i = 1}^m |X_{ij}|\right)}.
    \end{equation}
\end{proposition}
\begin{proof}
    Recall that the $\infty$-norm of a matrix is defined as
    \begin{align*}
        \|X\|_{\infty}
        &\colonequals \max_{v \neq 0} \frac{\|Xv\|_{\infty}}{\|v\|_{\infty}}
        \intertext{and it is simple to verify the alternative definition}
        &= \max_i \sum_j |X_{ij}|. \numberthis
    \end{align*}
    Thus our claim says that $\|X\| \leq \sqrt{\|X\|_{\infty}\|X^{\top}\|_{\infty}}$.
    We have $\|X\| = \sqrt{\lambda_{\max}(XX^{\top})}$.
    By the Gershgorin circle theorem, $\lambda_{\max}(XX^{\top}) \leq \|XX^{\top}\|_{\infty}$.
    Since the $\infty$-norm on matrices is induced as an operator norm by the $\ell^{\infty}$ vector norm, it is submultiplicative, whereby $\|XX^{\top}\|_{\infty} \leq \|X\|_{\infty} \|X^{\top}\|_{\infty}$, and the result follows.
\end{proof}
\noindent
We will use this result in the following simpler form.
\begin{corollary}
    \label{cor:sparse-matrix-norm}
    Let $X \in \RR^{m \times n}$, and suppose that $|X_{ij}| \leq C$ for all $i, j$ and that $X$ has at most K non-zero entries in each row and each column.
    Then, $\|X\| \leq KC$.
\end{corollary}

\subsection{Direct Bounds}

Two of our bounds, for graph matrices with no edges, amount to counting arguments.

\begin{proof}[Proof of Proposition~\ref{prop:T301}]
    Consider $T = T^{3,0,1} + 2I$. By easy computations, we see that $\boldsymbol{1}$ is an eigenvector of $T$ with eigenvalue $2(p-1)$, and $\mathbb{V}_1$ is an eigenspace of $T$ with eigenvalue $p-2$. Moreover, for any vector $u \in \mathbb{R}^{\binom{\mathbb{F}_p}{2} }$, we have
    \begin{align*}
        (Tu)_{\{i,j\}} &= \sum_{k_1 \in \mathbb{F}_p \setminus \{i,j\} } u_{\{k_1, j\}} + \sum_{k_2 \in \mathbb{F}_p \setminus \{i,j\} } u_{\{i, k_2\}} + 2u_{\{i,j\}}\\
        &= \sum_{k_1 \in \mathbb{F}_p \setminus \{j\} } u_{\{k_1, j\}} + \sum_{k_2 \in \mathbb{F}_p \setminus \{i\} } u_{\{i, k_2\}}\\
        &\in \mathbb{V}_0 \oplus \mathbb{V}_1. \numberthis
    \end{align*}
    Thus, we conclude $T^{3,0,1} + 2I = 2(p-1)P_0 + (p-2)P_1$, and the claim follows.
\end{proof}

\begin{proof}[Proof of Proposition~\ref{prop:T401}]
    Note that we have the combinatorial identity
    \begin{align*}
        1 = \One{\{\{a,b\} = \{c,d\}\} } + \One{\{\left|\{a,b\} \cap \{c,d\}\right| = 1\} } + \One{\{\{a,b\} \cap \{c,d\} = \emptyset\} },
    \end{align*}
    so $J = I + T^{3,0,1} + T^{4,0,1}$. By Proposition~\ref{prop:T301}, we then have
    \begin{align*}
        T^{4,0,1} &= J - I - T^{3,0,1}\\
        &= \binom{p}{2} P_0 - (P_0 + P_1 + P_2) - \left(2(p-2)P_0 + (p-4)P_1 - 2P_2 \right)\\
        &= \frac{(p-2)(p-3)}{2} P_0 - (p-3) P_1 + P_2 \numberthis
    \end{align*}
    as claimed.
\end{proof}

We next address the proofs that need no character sum calculations and can be treated with only linear algebra.

\begin{proof}[Proof of Proposition~\ref{prop:T_31i}]
    $T^{3,1,2}$ is the product of a diagonal matrix with absolute value of each entry bounded by $1$ with $T^{3,0,1}$. Thus, $\|T^{3,1,2}\| \le \|T^{3,0,1}\| = O(p)$, and similarly, $\|T^{3,1,3}\| = O(p)$.
\end{proof}

\begin{proof}[Proof of Proposition~\ref{prop:T_311}]
    $T^{3,1,1}$ is a submatrix of the sum of the four matrices $I_p \otimes S_{G_p}$, $S_{G_p} \otimes I_p$, $P(I_p \otimes S_{G_p})$, and $(I_p \otimes S_{G_p})P^{\top}$, for $P$ a suitable permutation matrix (which swaps the indices in a pair $(a, b) \in [p]^2$). Therefore, $\|T^{3,1,1}\| \leq 4\|S_{G_p}\| = O(\sqrt{p})$.
\end{proof}

\begin{proof}[Proof of Proposition~\ref{prop:T_32i}]
    $T^{3,2,1}$ is the product of a diagonal matrix with absolute value of each entry bounded by $1$ with $T^{3,1,1}$. Thus, $\|T^{3,2,1}\| \le \|T^{3,1,1}\| = O(\sqrt{p})$, and similarly, $\|T^{3,2,2}\| = O(\sqrt{p})$.
\end{proof}

\begin{proof}[Proof of Proposition~\ref{prop:T423}]
    We may write $T^{4, 2, 3} = \widetilde{T}^{4,2,3} + \Delta$, where $\widetilde{T}^{4, 2, 3}$ occurs as a submatrix of $S_{G_p}^{\otimes 2} + P(S_{G_p}^{\otimes 2})$ for a suitable permutation matrix $P$ (as in the proof of Proposition~\ref{prop:T_311}) and $\Delta$ has $|\Delta_{S,T}| \leq 1$ and has $\Delta_{S, T}\neq 0$ only when $S \cap T \neq \emptyset$.
    Then, $\|\widetilde{T}^{4, 2, 3}\| \leq 2\|S_{G_p}^{\otimes 2}\| = O(p)$, and $\|\Delta\| = O(p)$ by Corollary~\ref{cor:sparse-matrix-norm}.
\end{proof}

Next, we give the proofs that require only the basic character sum facts given in Propositions~\ref{prop:chi-sum-1} and \ref{prop:chi-sum-2}.
Recall that in the next two proofs we have $i \in \{1, 2\}$.

\begin{proof}[Proof of Proposition~\ref{prop:U5ij}]
    By the above observations, all entries of any of the $U^{5, i, j}$ have absolute value $O(1)$ (specifically, absolute value at most 4).
    Thus, $\|U^{t,i, j}\| \leq \|U^{t, i, j}\|_F = O(p^2)$.
\end{proof}

\begin{proof}[Proof of Proposition~\ref{prop:U4ij}]
    As above, all entries of any of the $U^{4, i, j}$ have absolute value $O(1)$.
    Moreover, $U^{4, i, j}$ is a direct sum of several matrices each of whose dimension is at most $p$.
    Thus, $\|U^{4, i, j}\|$ is at most the norm of any of these summands, which, again by bounding by the Frobenius norm, is $O(p)$.
\end{proof}

\begin{proof}[Proof of Proposition~\ref{prop:U3i1}]
    As above, all entries of any of the $U^{3, i, 1}$ have absolute value $O(1)$.
    Moreover, $U^{3, i, 1}$ is a diagonal matrix, so $\|U^{3, i, 1}\| = O(1)$.
\end{proof}

\begin{proof}[Proof of Proposition~\ref{prop:U_431}]
    Note that there are $O(p)$ nonzero entries in each row or column of $U^{4,3,1}$, as each entry is nonzero only when the index pairs intersect in $1$ element. Moreover, the entry $U^{4,3,1}_{\{a,b\}, \{a,c\}}$ has absolute value
    \begin{align*}
        \left|\sum_{i\in \mathbb{F}_p \setminus \{a,b,c\} } \chi(a-i)\chi(b-i)\chi(c-i)\right| = O(\sqrt{p}) \numberthis
    \end{align*}
    by Theorem~\ref{thm:weil}.
    By Corollary \ref{cor:sparse-matrix-norm}, we then find $\|U^{4,3,1}\| = O(p^{3/2})$.
\end{proof}

\subsection{Quadratic Rewriting Bound}

We now move on to an estimate for which we use Proposition~\ref{prop:chi-sum-2} in a slightly more sophisticated way, to control the powers of modified versions of graph matrices.

\begin{proof}[Proof of Proposition~\ref{prop:U541}]
    First, note that
    \begin{equation}
        \sum_{i \in \FF_p \setminus \{a,b,c,d\}} \chi(a - i)\chi(b - i)\chi(c - i)\chi(d - i) = \sum_{i \in \FF_p} \chi(a - i)\chi(b - i)\chi(c - i)\chi(d - i),
    \end{equation}
    since $\chi(0) = 0$.
    Next, we may write $U^{5,4,1} = \widetilde{U}^{5,4,1} + \Delta$, where
    \begin{equation}
        \label{eq:U541-entries}
        \widetilde{U}^{5,4,1}_{\{a, b\},\{c,d\}} = \sum_{i \in \FF_p} \chi(a - i)\chi(b - i)\chi(c - i)\chi(d - i)
    \end{equation}
    without the constraint that $\{a, b\} \cap \{c, d\} = \emptyset$, and $\Delta$ is a suitable correction with $\Delta_{\{a, b\}, \{c, d\}} \neq 0$ only when $\{a, b \}\cap \{c, d\} \neq \emptyset$.
    We have $\Delta_{\{a, b\}, \{a, b\}} = -(p - 2)$, and the non-zero off-diagonal entries of $\Delta$ are of magnitude $O(1)$ by Proposition~\ref{prop:chi-sum-2}.
    Thus, $\|\Delta\| = O(p)$ by Corollary~\ref{cor:sparse-matrix-norm}, so $\|U^{5, 4, 1}\| = \|\widetilde{U}^{5, 4, 1}\| + O(p)$.

    Now, note that $\widetilde{U}^{5,4,1}$ is a submatrix of the rectangular matrix $X \in \RR^{(\FF_p)_{(2)} \times \FF_p^2}$, having rows indexed by pairs $(a, b) \in \FF_p^2$ with $a \neq b$ (recall that we denote this set $(\FF_p)_{(2)}$) and columns indexed by arbitrary pairs $(c, d) \in \FF_p^2$, where the entries of $X$ for any pair are given again by the formula \eqref{eq:U541-entries}.
    In particular, $\|\widetilde{U}^{5,4,1}\| \leq \|X\|$.
    Now, we have, for any $(a, b), (c, d) \in (\FF_p)_{(2)}$,
    \begin{align*}
        (XX^{\top})_{(a, b), (c, d)}
        &= \sum_{i, j, k, \ell \in \FF_p} \chi(a - i)\chi(b - i) \chi(j - i)\chi(k - i) \chi(j - \ell) \chi(k - \ell) \chi(c - \ell) \chi(d - \ell) \\
        &= \sum_{i, \ell \in \FF_p} \chi(a - i)\chi(b - i) \chi(c - \ell) \chi(d - \ell) (-1 + \One\{i = \ell\}\,p)^2 \tag{by Proposition~\ref{prop:chi-sum-2}} \\
        &= (p^2 - 2p) \sum_{i \in \FF_p}\chi(a - i) \chi(b - i) \chi(c- i) \chi(d - i) \\
        &\hspace{1cm} + \left(\sum_{i \in \FF_p} \chi(a - i) \chi(b - i) \right)\left(\sum_{\ell \in \FF_p} \chi(c - \ell) \chi(d - \ell) \right)
        \intertext{and, since we have $a \neq b$ and $c \neq d$,}
        &= (p^2 - 2p) X_{(a, b),(c, d)} + 1. \numberthis
    \end{align*}
    Taking norms and using the triangle inequality, we find $\|X\|^2 \leq (p^2 - 2p) \|X\| + p^2$, or $\|X\|^2 - (p^2 - 2p) \|X\| - p^2 \leq 0$.
    In other words, $\|X\|$ must lie between the two real roots of the quadratic equation $t^2 - (p^2 - 2p) t - p^2 = 0$.
    Solving the equation shows that both roots are $O(p^2)$, so $\|X\| = O(p^2)$.
    Thus $\|\widetilde{U}^{5,4,1}\| = O(p^2)$, completing the proof.
\end{proof}

\subsection{Subspace-Specific Bounds}

We now give one proof where we must analyze the restrictions of matrices to the subspaces $\mathbb{V}_i$.

\begin{proof}[Proof of Proposition~\ref{prop:T_4ij}]
    We only prove the claim for $T^{4,2,1}$. The two other claims are proved analogously.

    First, we show $\|T^{4,2,1}\| = O(p^{3/2})$. We may write $T^{4,2,1} = \Tilde{A} + \Tilde{B}$, where $\Tilde{A}, \Tilde{B}$ are submatrices of $A, B \in \mathbb{R}^{\mathbb{F}_p^2 \times \mathbb{F}_p^2 }$ defined by
    \begin{align}
        A_{(a,b), (c,d)} &= \chi(a-c)\chi(a-d)\\
        B_{(a,b), (c,d)} &= \chi(b-c)\chi(b-d).
    \end{align}
    Note that $AA^\top = (S_{G_p}^2)^{\circ 2} \otimes J_p$ and $BB^\top = J_p \otimes (S_{G_p}^2)^{\circ 2}$. Recall that $X^{\circ 2}$ denotes the entrywise square of the matrix $X$, which is a submatrix of $X^{\otimes 2}$. Then, $\|AA^\top \| = \|(S_{G_p}^2)^{\circ 2}\|\|J_p\| \le ((\sqrt{p})^2)^2 p = p^3$, so $\|A\| \le p^{3/2}$. Similarly, $\|B\| \le p^{3/2}$. We therefore conclude $\|T^{4,2,1}\| \le \|\Tilde{A}\| + \|\Tilde{B}\| = O(p^{3/2})$.

    Next, we show $\| T^{4,2,1} P_2\| = O(\sqrt{p})$.
    Consider $T = T^{4,2,1} + \Tilde{T}$, where $\widetilde{T}$ has entries
    \begin{equation}
        \widetilde{T}_{\{a,b\}, \{c,d\}} = \left\{\begin{array}{ll} \chi(a-c)\chi(a-d) + \chi(b-c)\chi(b-d) & \text{if } \{a, b\} \cap \{c, d\} \neq \emptyset, \\
        0 & \text{if } \{a, b\} \cap \{c, d\} = \emptyset. \end{array}\right.
    \end{equation}
    The resulting matrix then has $T_{\{a,b\}, \{c,d\}} = \chi(a-c)\chi(a-d) + \chi(b-d)\chi(b-d)$.

    For any $v \in \mathbb{R}^{\binom{\mathbb{F}_p}{2} }$, we have
    \begin{align*}
        \left(T^\top v\right)_{\{a,b\}}
        &= \sum_{\{c,d\} }  (\chi(a-c)\chi(a-d)+ \chi(b-c)\chi(b-d)) v_{\{c,d\} }\\
        &= \sum_{\{c,d\} }  \chi(a-c)\chi(a-d)v_{\{c,d\}} + \sum_{\{c,d\} }  \chi(b-c)\chi(b-d) v_{\{c,d\} }\\
        &\in \mathbb{V}_0 \oplus \mathbb{V}_1. \numberthis
    \end{align*}
    So, $P_2 T^\top = 0$, and $ T^{4,2,1} P_2 = -\Tilde{T}P_2$. To bound $\|T^{4,2,1}P_2\|$, we bound $\|\Tilde{T}\|$. Similarly to the previous argument, we may embed $\Tilde{T}$ into a larger matrix $T'$ with indices being unordered pairs $(\mathbb{F}_p)_{(2)}$, and then write $T' = A'+ B'+ C' + D'$ as a sum of $4$ matrices, such that $A'^\top A', B'^\top B', C'^\top C', D'^\top D'$ are submatrices of $(I_p \otimes S_{G_p})^2$ or $(S_{G_p} \otimes I_p)^2$, from which we conclude $\|\Tilde{T}\| \le \|T'\| \le \|A'\| + \|B'\| + \|C'\| + \|D'\| \le 4\|(I_p \otimes S_{G_p})^2\|^{1/2} = 4((\sqrt{p})^2)^{1/2} = O(\sqrt{p})$. Thus, $\|T^{4,2,1}P_2\| \le \|\Tilde{T}\| = O(\sqrt{p})$.
\end{proof}

\subsection{Trace Power Method Bounds}
We next give two proofs that apply the trace power method, in the style of \cite{AMP-2016-GraphMatrices}, to graph matrix norms.

\begin{proof}[Proof of Proposition~\ref{prop:T_431}]
    We may write $T^{4,3,1} = \Tilde{A} + \Tilde{B} + \Tilde{C} + \Tilde{D} + \Tilde{\Delta}$, where $\Tilde{A}, \Tilde{B}, \Tilde{C}, \Tilde{D}, \Tilde{\Delta}$ are submatrices of $A, B, C, D, \Delta \in \mathbb{R}^{\mathbb{F}_p^2 \times \mathbb{F}_p^2 }$ defined by
    \begin{align*}
        A_{(a,b), (c,d)} &= \chi(a-c)\chi(a-d)\chi(b-c) \numberthis \\
        B_{(a,b), (c,d)} &= \chi(a-c)\chi(a-d)\chi(b-d) \numberthis \\
        C_{(a,b), (c,d)} &= \chi(a-c)\chi(b-c)\chi(b-d) \numberthis \\
        D_{(a,b), (c,d)} &= \chi(a-d)\chi(b-c)\chi(b-d) \numberthis \\
        \Delta_{(a,b), (c,d)} &= (- \chi(a-c)\chi(a-d)\chi(b-c)\\
        &\quad\, - \chi(a-c)\chi(a-d)\chi(b-d)\\
        &\quad\, - \chi(a-c)\chi(b-c)\chi(b-d)\\
        &\quad\, - \chi(a-d)\chi(b-c)\chi(b-d)) \, \One\{|\{a,b,c,d\}| < 4\}. \numberthis
    \end{align*}
    First, note that $\|\Delta\| \le O(p)$ by Corollary \ref{cor:sparse-matrix-norm}, as there are only $O(p)$ nonzero entries in each row or column and each nonzero entry is $O(1)$.

    Next, we consider $\tr((A^\top A)^k )$, which can be written as
    \begin{align*}
         &\quad \sum_{\substack{a_1, \dots, a_k \in \FF_p \\ b_1, \dots, b_k \in \FF_p \\ c_1, \dots, c_k \in \FF_p \\ d_1, \dots, d_k \in \FF_p}} \chi(a_1 - c_1)\chi(a_1 - d_1)\chi(b_1 - c_1)\chi(a_1 - c_2)\chi(a_1 - d_2)\chi(b_1 - c_2)\cdots\\
         &= \sum_{\substack{a_1, \dots, a_k \in \FF_p \\ c_1, \dots, c_k \in \FF_p}} \prod_{\ell=1}^k\chi(a_{\ell} - c_{\ell})\chi(a_{\ell} - c_{\ell+1})\left(\sum_{b \in \FF_p} \chi(b - c_{\ell})\chi(b - c_{\ell+1}) \right)\left(\sum_{d \in \FF_p} \chi(a_{\ell-1} - d)\chi(a_{\ell} - d)\right), \numberthis
    \end{align*}
    with index arithmetic performed modulo $k$.
    By Proposition \ref{prop:chi-sum-2}, if $i = j$, $\sum_x \chi(x- i)\chi(x-j) = p - 1$, and if $i \ne j$, $\sum_x \chi(x- i)\chi(x- j) = -1$. For a fixed sequence of $a_1, c_1, \dots, a_k, c_k$, let $n_1 \colonequals \# \{\ell: a_{\ell-1} = a_{\ell}\}$ and $n_2 \colonequals \# \{\ell: c_{\ell} = c_{\ell+1}\}$. Note that $n_1, n_2$ take values in $\{0, 1, \dots, k-3, k-2, k\}$. The contribution of such a sequence to the sum is $O(p^{n_1 + n_2})$. The number of sequences $a_1, c_1, \dots, a_k, c_k$ with these values of $n_1$ and $n_2$ is at most $\binom{k}{n_1}\binom{k}{n_2} p^{2k-n_1-n_2} p^{\One\{n_1 = k\} + \One\{n_2 = k\}}$. Thus, the trace above can be bounded by
    \begin{align*}
        \tr((A^\top A)^k ) &\lesssim \sum_{n_1, n_2} \binom{k}{n_1}\binom{k}{n_2}p^{2k-n_1-n_2}p^{\One\{n_1 = k\} + \One\{n_2 = k\}} p^{n_1 + n_2}\\
        &\lesssim (k-1)^2 2^{2k} p^{2k} + 2(k-1) 2^k p^{2k+1} + p^{2k+2}\\
        &\lesssim p^{2k}(k^2 2^{2k}+ k2^kp + p^2). \numberthis
    \end{align*}
    Taking $k = \Theta(\log p)$, we get $\|A^\top A\| \le \text{tr}((A^\top A)^k )^{1/k} \le O(p^2)$, and so $\|A\| = O(p)$. By symmetric arguments, we also get $\|B\|, \|C\|, \|D\| = O(p)$.
    Combining these results, we find
    \begin{align*}
        \|T^{4,3,1}\| \le \|A\| + \|B\| + \|C\|+ \|D\| + \|\Delta\| = O(p), \numberthis
    \end{align*}
    as claimed.
\end{proof}

\begin{proof}[Proof of Proposition~\ref{prop:U_53i}]
    By symmetry, it suffices to prove the claim for $U^{5,3,1}$.

    First, we show $\|U^{5,3,1}\| = O(p^2)$. We may write $U^{5,3,1} = \Tilde{A} + \Tilde{B} + \Tilde{\Delta}$, where $\Tilde{A}, \Tilde{B}, \Tilde{\Delta} $ are submatrices of $A,B, \Delta \in \mathbb{R}^{\mathbb{F}_p^2 \times \mathbb{F}_p^2 } $ defined by
    \begin{align}
        A_{(a,b), (c,d)} &= \sum_{i \in \mathbb{F}_p } \chi(a-i)\chi(b-i)\chi(c-i)\\
        B_{(a,b), (c,d)} &= \sum_{i \in \mathbb{F}_p } \chi(a-i)\chi(b-i)\chi(d-i)\\
        \Delta_{(a,b), (c,d)} &= -\sum_{i \in \{a,b,c,d\}} \bigg(\chi(a-i)\chi(b-i)\chi(c-i) + \chi(a-i)\chi(b-i)\chi(d-i)\bigg).
    \end{align}
    Since each entry of $\Delta$ is $O(1)$, by Corollary \ref{cor:sparse-matrix-norm}, $\|\Delta\| = O(p^2)$. For $A$, we consider $\tr{((A^\top A)^k)}$, which can be written as
    \begin{align*}
        &\quad \sum_{\substack{a_1, \dots, a_k \in \FF_p \\ b_1, \dots, b_k \in \FF_p \\ c_1, \dots, c_k \in \FF_p \\ d_1, \dots, d_k \in \FF_p \\ i_1, \dots, i_k \in \FF_p \\ j_1, \dots, j_k \in \FF_p}} \chi(c_1 - i_1)\chi(a_1 - i_1)\chi(b_1 - i_1)\chi(a_1 - j_1)\chi(b_1 - j_1)\chi(c_2 - j_1) \cdots\\
        &= p^k\sum_{\substack{i_1, \dots, i_k \in \FF_p \\ j_1, \dots, j_k \in \FF_p}} \prod_{\ell=1}^k \left(\sum_{a \in \FF_p} \chi(a - i_{\ell})\chi(a - j_{\ell})\right)
        \left(\sum_{b \in \FF_p}\chi(b - i_{\ell})\chi(b - j_{\ell}))\right) \\
        &\hspace{3 cm} \left(\sum_{c \in \FF_p}\chi(c - j_{\ell-1})\chi(c - i_{\ell}))\right) \numberthis
    \end{align*}
    As in the previous proof, we use Proposition \ref{prop:chi-sum-2}.
    For a sequence $i_1, j_1, \dots, i_k, j_k$, we define $n_1 = \# \{l: i_l = j_l\}$ and $n_2 = \# \{l: j_{l-1}= i_l\}$. Note that $n_1, n_2$ take values in $\{0, 1, \dots, k-1, k\}$ and moreover $n_1 + n_2 \ne 2k - 1$. The contribution of such a sequence to the sum is $O(p^{2n_1 + n_2})$. The number of sequences $i_1, j_1, \dots, i_k, j_k$ with these values of $n_1$ and $n_2$ is at most $\binom{k}{n_1}\binom{k}{n_2} p^{2k - n_1 - n_2} p^{\One\{n_1 + n_2 = 2k\}}$.
    Thus,
    \begin{align*}
        \tr((A^{\top}A)^k)
        &\lesssim p^k \sum_{n_1, n_2} \binom{k}{n_1}\binom{k}{n_2} p^{2k - n_1 - n_2} p^{\One\{n_1 + n_2 = 2k\}} p^{2n_1 + n_2} \\
        &\lesssim p^k k^2 2^{2k} p^{3k} + p^k p^{3k+1} \\
        &\lesssim p^{4k}(k^2 2^{2k}+ p). \numberthis
    \end{align*}
    Taking $k = \Theta(\log p)$, we get $\|A^\top A\| \le \text{tr}\left((A^\top A)^k \right)^{1/k} \le O(p^4)$, and $\|A\| = O(p^{2})$. By a symmetric argument, we also have $\|B\| = O(p^{2})$.
    We then conclude
    \begin{align*}
        \|U^{5,3,1}\| &\le \|A\| + \|B\| + \|\Delta\|= O(p^2), \numberthis
    \end{align*}
    as claimed.
\end{proof}

\subsection{Delicate Estimates: Proof of Theorem~\ref{thm:norm_T_441}}

Finally, we analyze the most complicated graph matrix whose norm we need to control, the matrix $T^{4,4,1}$.

\subsubsection{Spectral Decomposition from Block-Circulant Form}
In order to prove a norm bound for $T^{4,4,1}$, we define another matrix $\widetilde{T}^{4,4,1} \in \mathbb{R}^{(\mathbb{F}_p)_{(2)} \times (\mathbb{F}_p)_{(2)}}$ and utilize its automorphism group. This matrix is defined in the following way:
\begin{align*}
    \widetilde{T}^{4,4,1}_{(a,b), (c,d)} = T^{4,4,1}_{\{a,b\}, \{c,d\}}. \numberthis
\end{align*}
In other words, the entries of $\widetilde{T}$ can be looked up from $T^{4,4,1}$ by converting the indices from ordered pairs to sets.
Up to a permutation of rows and columns,
\begin{align*}
    \widetilde{T} = T^{4,4,1} \otimes \begin{bmatrix}
        1 & 1\\
        1 & 1
    \end{bmatrix}, \numberthis
\end{align*}
so $\|\widetilde{T}^{4,4,1}\| = 2\|T^{4,4,1}\|$.
Moreover, $\widetilde{T}$ has a natural group of automorphisms induced by the structure of $\mathbb{F}_p$, which is defined below:
\begin{definition}
    The \emph{affine group} $\Gamma$ of the finite field $\mathbb{F}_p$ is the group of all invertible affine transformations from $\mathbb{F}_p$ to $\mathbb{F}_p$, i.e., $\Gamma := \{x \mapsto ix + j: i \in \mathbb{F}_p^\times, j \in \mathbb{F}_p\}$. The group multiplication is given by composition.
\end{definition}
It is easy to check that $\Gamma$ of $\mathbb{F}_p$ is indeed a group of automorphisms of $\widetilde{T}$, acting on the indices by $g((a,b)) = (g(a), g(b))$ for $(a,b) \in (\mathbb{F}_p)_{(2)}$ and $g \in \Gamma$. Recall that the the entries of $T^{4,4,1}$ are
\begin{align*}
    T^{4,4,1}_{\{a,b\}, \{c,d\}} &= \chi((a-c)(a-d)(b-c)(b-d)), \numberthis
\end{align*}
so the entries of $\widetilde{T}^{4,4,1}$ are likewise
\begin{align*}
    \widetilde{T}^{4,4,1}_{(a,b), (c,d)} &= \chi((a-c)(a-d)(b-c)(b-d)), \numberthis
\end{align*}
and we may verify that $\Gamma$ is an automorphism group of $\widetilde{T}$ by checking that, for any $i \in \FF_p^{\times}$ and $j \in \FF_p$,
\begin{align*}
    \widetilde{T}^{4,4,1}_{(ia + j, ib + j), (ic + j, id + j)} &= \chi((ia-ic)(ia-id)(ib-ic)(ib-id))\\
    &= \chi((a-c)(a-d)(b-c)(b-d))\chi(i)^4\\
    &= \chi((a-c)(a-d)(b-c)(b-d))\\
    &= \widetilde{T}^{4,4,1}_{(a,b), (c,d)}. \numberthis
\end{align*}

We observe in particular that this group acts \emph{transitively} on the index set $(\mathbb{F}_p)_{(2)}$.
Using representation theory, it was shown by \cite{Lovasz-1975-SpectraGraphsTransitiveGroups,Babai-1979-SpectraCayleyGraphs} that the spectrum of a matrix with a transitive automorphism group can be expressed using the irreducible representations (and their characters) of the automorphism group. We will make use of this observation, but we give a self-contained proof for our special case without resorting to representation theory.

\begin{definition}
    A real symmetric matrix $M \in \mathbb{R}^{dn \times dn}$ is said to be \emph{block-circulant} if its rows and columns can be permuted such that it can be written in the block form
    \begin{equation}
        \label{eq:block-circulant-general}
        M = \begin{bmatrix}
            B^{(0)} & B^{(1)} & B^{(2)} & \dots & B^{(d-1)}\\
            B^{(d-1)} & B^{(0)} & B^{(1)} & \dots & B^{(d-2)}\\
            B^{(d-2)} & B^{(d-1)} & B^{(0)} & \dots & B^{(d-3)}\\
            \vdots & \vdots & \vdots & \ddots & \vdots\\
            B^{(1)} & B^{(2)} & B^{(3)} & \dots & B^{(0)}
        \end{bmatrix},
    \end{equation}
    where $B^{(i)} \in \mathbb{R}^{n \times n}$ for $0 \le i \le d-1$.
\end{definition}

The use of this definition is that the computation of the spectrum of a block-circulant matrix may be condensed in the following way.
\begin{proposition}
    \label{prop:block-circulant-spectrum}
    Suppose $M$ is block-circulant with block structure as in \eqref{eq:block-circulant-general}.
    Then, the spectrum of $M$ is the disjoint union of the spectra of $d$ smaller $n \times n$ Hermitian matrices $S^{(\psi)}$ taken over $\psi$ all additive characters of $\mathbb{Z}/d\mathbb{Z}$, where the matrices $S^{(\psi)}$ are defined as
        \begin{equation}
            S^{(\psi)} = \sum_{i=0}^{d-1} \psi(i) B^{(i)}.
        \end{equation}
\end{proposition}
\begin{proof}
    First, note that $S^{(\psi)}$ are all Hermitian, since $\psi(-i) = \overline{\psi(i)}$ and $B^{(-i)} = B^{(i)^{\top}}$.
    Let $u \in \mathbb{C}^{n}$ be an eigenvector of $S^{(\psi)}$ with eigenvalue $\lambda$. Define $w^{(\psi)} = (1 = \psi(0), \psi(1), \dots, \psi(d - 1))^\top \in \mathbb{C}^{d}$. Then, consider the vector $v = w^{(\psi)} \otimes u \in \mathbb{C}^{dn}$.
    The $i$th block of length $n$ of the vector $Mv$, starting counting from $i = 0$, is
    \begin{equation}
       \sum_{j = 0}^{d - 1} \psi(j) B^{(j - i)} u = \psi(i) \sum_{j = 0}^{d - 1} \psi(j) B^{(j)} u = \psi(i) S^{(\psi)} u = \lambda \cdot \psi(i) u,
    \end{equation}
    which is $\lambda$ multiplied by the $i$th block of length $n$ of $v$.
    Thus, $v$ is indeed an eigenvector of $M$ with eigenvalue $\lambda$.

    Moreover, the $w^{(\psi)}$ are mutually orthogonal by the orthogonality of characters, and it follows that the spectrum of $M$ is the disjoint union of the spectra of the $S^{(\psi)}$.
\end{proof}

\begin{proposition}\label{prop:block_circulant}
    Let $M \in \mathbb{R}^{(\mathbb{F}_p)_{(2)} \times (\mathbb{F}_p)_{(2)}}$ be a real symmetric matrix. Suppose the affine group $\Gamma$ of $\mathbb{F}_p$ is a group of automorphisms of $M$, where the action of an element $g \in \Gamma$ on the index set $(\mathbb{F}_p)_{(2)}$ is given by the natural one: $g((a,b)) = (g(a), g(b))$.
    Then, $M$ is block-circulant with the form
    \begin{equation}
        M = \begin{bmatrix}
            B^{(0)} & B^{(1)} & B^{(2)} & \dots & B^{(p-2)}\\
            B^{(p-2)} & B^{(0)} & B^{(1)} & \dots & B^{(p-3)}\\
            B^{(p-3)} & B^{(p-2)} & B^{(0)} & \dots & B^{(p-4)}\\
            \vdots & \vdots & \vdots & \ddots & \vdots\\
            B^{(1)} & B^{(2)} & B^{(3)} & \dots & B^{(0)}
        \end{bmatrix},
    \end{equation}
    where the set of indices of the $i$th partition of the rows and the columns consists of the set of pairs $\{(h^i a, h^i(a+1)): a \in \mathbb{F}_p\} \subset \left(\mathbb{F}_p\right)_{(2)}$ in which $h \in \mathbb{F}_p^\times$ is a multiplicative generator of $\mathbb{F}_p^\times$.
    Let $S^{(\psi)}$ be as in Proposition~\ref{prop:block-circulant-spectrum}.
    Then, the following also hold:
    \begin{enumerate}
        \item The matrices $S^{(\psi)}$ all have the all-ones vector $\boldsymbol{1}$ as an eigenvector.
        \item The matrices $Q_1 S^{(\psi)} Q_1$ have the same sets of eigenvalues, where $Q_1$ is the orthogonal projection to the orthogonal complement of $\boldsymbol{1}$ (as in our earlier notation). In other words, the $S^{(\psi)}$ have a common set of eigenvalues except for the eigenvalues corresponding to the all-ones eigenvector.
    \end{enumerate}
\end{proposition}

\begin{proof}
    To show $M$ is block-circulant with the prescribed index partition as in the statement of the proposition, we only need to check
    \begin{equation}
        M_{(h^ia, h^i(a+1)), (h^jb, h^j(b+1))} = M_{(h^{i+1}a, h^{i+1}(a+1)), (h^{j+1}b, h^{j+1}(b+1))},
    \end{equation}
    which follows from $\Gamma$ being a group of automorphisms of $M$ and $(x \mapsto hx) \in \Gamma$.

    Next, we check that $\boldsymbol{1}$ is an eigenvector of each $S^{(\psi)}$.
    We have
    \begin{equation}
        (B^{(i)} \boldsymbol{1})_a = \sum_{b \in \FF_p} M_{(h^ia,h^i(a + 1)), (h^ib, h^i(b + 1))} = \sum_{b \in \FF_p} M_{(h^i(a - b),h^i(a - b + 1)), (0, h^i)}
    \end{equation}
    by automorphism invariance of $M$, and clearly the final quantity does not depend on $a$.
    Thus, $\boldsymbol{1}$ is an eigenvector of each $B^{(i)}$ (though with eigenvalue depending on $i$), and the claim follows as each $S^{(\psi)}$ is a linear combination of the $B^{(i)}$.

    It remains to show that the remaining eigenvalues of $S^{(\psi)}$ do not depend on $\psi$.
    Consider the collection of vectors $v_t = \left(e_p(0), e_p(t), \dots, e_p((p-1)t)\right)^\top$ where $0 \le t \le p-1$. By orthogonality of characters, the $v_t$ form an orthogonal basis of $\mathbb{C}^{\mathbb{F}_p}$.
    We first consider how each $B^{(i)}$ acts on this basis.
    We have
        \begin{align*}
            (B^{(i)} v_t)_a &= \sum_{b \in \mathbb{F}_p} B^{(i)}_{a,b} e_p(bt)\\
            &= \sum_{b \in \mathbb{F}_p} M_{(a, a+1), (h^ib, h^i(b+1))} e_p(bt)
            \intertext{and, changing variables $b \leftarrow b + ah^{-i}$ and using the automorphism group $\Gamma$ of $M$, we get,}
            &= \sum_{b \in \mathbb{F}_p} M_{(a, a+1), (h^ib + a, h^i(b+1) + a)} e_p\left((b + ah^{-i})t\right) \\
            &= \left[\sum_{b \in \mathbb{F}_p} M_{(0, 1), (h^ib, h^i(b+1))} e_p(bt) \right] e_p\left(ah^{-i}t\right)\\
            &= \gamma^{(i)}_t (v_{h^{-i}t})_a, \numberthis
        \end{align*}
        where $\gamma^{(i)} \colonequals \sum_{b \in \mathbb{F}_p} M_{(0, 1), (h^ib, h^i(b+1))} e_p(bt)$.
        Therefore,
        \begin{align}\label{eq:block_circulant_1}
            B^{(i)}v_t = \gamma^{(i)}_t v_{h^{-i}t}. \numberthis
        \end{align}

    Now, let $\psi_0$ be the trivial character of $\mathbb{Z}/(p-1)\mathbb{Z}$, i.e., $\psi_0(i) = 1$ for $i \in \mathbb{Z}/(p-1)\mathbb{Z}$. Then, $S^{(\psi_0)} = \sum_{i=0}^{p-2} B^{(i)}$.
        Let $u$ be an eigenvector of $S^{(\psi_0)}$ with eigenvalue $\lambda$ and with $\langle u, \boldsymbol{1} \rangle = 0$. Write $u$ in the basis $v_t$, $u = \sum_{t = 0}^{p-1} c_t v_t$. We have $c_0 = 0$ as $v_0 = \boldsymbol{1}$. Recall that $h \in \mathbb{F}_p^\times$ is a generator of $\mathbb{F}_p^\times$, so we may alternatively write
        \begin{align*}
            u &= \sum_{t=0}^{p-1} c_t v_t = \sum_{t=1}^{p-1} c_t v_t = \sum_{i=0}^{p-2} c_{h^i} v_{h^i}. \numberthis
        \end{align*}
        Since the eigenvalue corresponding to $u$ is $\lambda$, we have
        \begin{align*}
            \lambda \sum_{i=0}^{p-2} c_{h^i} v_{h^i}
            &= \lambda u \\
            &= S^{(\psi_0)} u \\
            &= \sum_{i=0}^{p-2} B^{(i)} \sum_{j=0}^{p-2} c_{h^j} v_{h^j}\\
            &= \sum_{j=0}^{p-2}\sum_{i=0}^{p-2} c_{h^j} \gamma^{(i)}_{h^j} v_{h^{-i+j}}
            \intertext{and, changing variables $i \leftarrow -i+j$, we get,}
            &= \sum_{j=0}^{p-2}\sum_{i=0}^{p-2} c_{h^j} \gamma^{(-i+j)}_{h^j} v_{h^{i}}\\
            &= \sum_{i=0}^{p-2} \left(\sum_{j=0}^{p-2} c_{h^j} \gamma^{(-i+j)}_{h^j}\right) v_{h^{i}}, \numberthis
        \end{align*}
        and thus, since the $v_t$ form a basis,
        \begin{align}\label{eq:block_circulant_2}
            \sum_{j=0}^{p-2} c_{h^j} \gamma^{(-i+j)}_{h^j} = \lambda c_{h^i}.
        \end{align}

        Now, for $\psi \neq \psi_0$ another additive character, consider the vector $u^{(\psi)} \colonequals \sum_{i=0}^{p-2} \psi(-i) c_{h^i} v_{h^i}$. We have
        \begin{align*}
            S^{(\psi)} u^{(\psi)} &= \sum_{i=0}^{p-2} \psi(i) B^{(i)} \sum_{j=0}^{p-2} \psi(-j) c_{h^j} v_{h^j}\\
            &= \sum_{j=0}^{p-2} \sum_{i=0}^{p-2} \psi(i-j)c_{h^j} B^{(i)}v_{h^j}
            \intertext{where, using \eqref{eq:block_circulant_1}, we get}
            &= \sum_{j=0}^{p-2} \sum_{i=0}^{p-2} \psi(i-j)c_{h^j} \gamma^{(i)}_{h^j} v_{h^{-i+j}}
            \intertext{and, changing variables $i \leftarrow -i+j$, we get,}
            &= \sum_{j=0}^{p-2} \sum_{i=0}^{p-2} \psi(-i)c_{h^j}\gamma^{(-i+j)}_{h^j} v_{h^{i}} \\
            &= \sum_{i=0}^{p-2} \psi(-i)\left(\sum_{j=0}^{p-2} c_{h^j}\gamma^{(-i+j)}_{h^j}\right) v_{h^{i}}
            \intertext{and finally, using \eqref{eq:block_circulant_2}, we get}
            &= \lambda \sum_{i=0}^{p-2} \psi(-i) c_{h^i} v_{h^i}\\
            &= \lambda u^{(\psi)}. \numberthis
        \end{align*}
        Thus, $u^{(\psi)}$ is an eigenvector of $S^{(\psi)}$ of the same eigenvalue $\lambda$, and the result follows.
\end{proof}

\subsubsection{Character Sum Estimates}
When we apply the machinery developed above to $\widetilde{T}$, we will be left with smaller matrices $S^{(\psi)} \in \mathbb{C}^{\mathbb{F}_p \times \mathbb{F}_p }$ with entries in terms of $\chi$ whose norm we need to bound.
We now prove the character sum bounds we will need in order to do this.

\begin{theorem}\label{thm:character_sum_1}
    Let $T \in \RR^{p \times p}_{\sym}$ have entries
    \begin{equation}
        T_{ij} = \sum_{x \in \FF_p} \chi\bigg( (ix - j)(ix - (j + 1))((i + 1)x - j)((i + 1)x - (j + 1)) \bigg).
    \end{equation}
    Then, $\|T\| = O(p^{5/4})$.
\end{theorem}

\begin{remark}
    \label{rem:T441-tightness}
    As we will see, it is the $5/4$ exponent in this result that yields the same in Theorem~\ref{thm:norm_T_441} for $T^{4,4,1}$.
    A naive argument here can easily show a bound of $O(p^{\frac{3}{2}})$, which gives the same bound for $T^{4,4,1}$, but this is not enough for proving our main theorem.
    While the $5/4$ exponent \emph{is} sufficient for our main theorem, we believe that this is not tight. We conjecture based on numerical experiments that the right magnitude of $\|T^{4,4,1}\|$ is $\Theta(p)$, and the optimal constant for the linear term is $1$, i.e.,
\begin{equation}
    \limsup_{p \to \infty} \frac{1}{p} \|T^{4,4,1}\| = 1.
\end{equation}
\end{remark}

\begin{proof}
    We start by achieving a factorization of the matrix $T$.
    Note that we may write
    \begin{align*}
        T_{ij}
        &= \sum_{a, b \in \FF_p} \chi(a(a - 1)b(b - 1)) \One\{ib - (i + 1)a - j = 0\} \\
        &= \frac{1}{p} \sum_{a, b, x \in \FF_p} \chi(a(a - 1)b(b - 1)) e_p(x(ib - (i + 1)a - j))
        \intertext{and, changing variables $a \leftarrow (a + 1) / 2$ and $b \leftarrow (b + 1) / 2$, we find, writing $\overline{2}$ for the multiplicative inverse of 2 modulo $p$,}
        &= \frac{1}{p} \sum_{a, b, x \in \FF_p} \chi((a^2 - 1)(b^2 - 1)) e_p\left(x\left(\overline{2}i(b + 1) - \overline{2}(i + 1)(a + 1) - j\right)\right) \\
        &= \frac{1}{p} \sum_{x \in \FF_p} e_p\left(x(-\overline{2} - j)\right) \left( \sum_{a \in \FF_p} \chi(a^2 - 1) e_p(\overline{2}(i + 1)xa)\right)\left( \sum_{a \in \FF_p} \chi(b^2 - 1) e_p(\overline{2}ixb)\right) \\
        &= \frac{1}{p} \sum_{x \in \FF_p} e_p\left(x(-\overline{2} - j)\right) K\left(\left(\frac{i}{4}x\right)^2\right) K\left(\left(\frac{i + 1}{4}x\right)^2\right). \numberthis
    \end{align*}
    Thus, there is a unitary matrix $U \in \CC^{n \times n}$ given by a row and column permutation of the discrete Fourier transform matrix such that $T = p^{-1/2} T^{(1)} U$, where the entries of $T^{(1)}$ are given by
    \begin{equation}
        T^{(1)}_{ij} = K\left(i^2j^2\right) K\left((i + 1)^2j^2\right).
    \end{equation}

    Next, consider the slightly adjusted matrix $T^{(2)} \in \RR^{n \times n}$ with entries
    \begin{equation}
        T^{(2)}_{ij} = K\left(i^2j\right) K\left((i + 1)^2j\right).
    \end{equation}
    Letting $S \in \RR^{p \times (p + 1) / 2}$ be the submatrix of $T^{(2)}$ indexed by columns for which $j$ is a quadratic residue (including zero), we have that $T^{(1)}$ is a submatrix of $[1 \, \, 1] \otimes S$ (indeed, it is exactly this matrix with one column removed).
    Thus, we may bound
    \begin{equation}
        \|T^{(1)}\| \leq \|[1 \, \, 1] \otimes S\| = \sqrt{2} \|S\| \leq \sqrt{2} \|T^{(2)}\|.
    \end{equation}

    We may then bound the norm of our original $T$ as
    \begin{align*}
        \|T\|^2
        &= \frac{1}{p} \|T^{(1)}\|^2 \\
        &\leq \frac{2}{p} \|T^{(2)} \|^2 \\
        &= \frac{2}{p} \|T^{(2)}T^{(2)^{\top}} \|^2
        \intertext{and now, applying the Gershgorin circle theorem,}
        &\leq \frac{2}{p} \max_{a \in \FF_p} \sum_{b \in \FF_p} |(T^{(2)} T^{(2)^{\top}})_{ab}| \\
        &= \frac{2}{p} \max_{a \in \FF_p} \sum_{b \in \FF_p} \left| \sum_{x \in \FF_p} K\left(a^2x\right) K\left((a + 1)^2x\right) K\left(b^2x\right) K\left((b + 1)^2x\right)\right|
        \intertext{Now, by Proposition~\ref{prop:kloosterman-corr}, so long as some element among $a^2, (a + 1)^2, b^2, (b + 1)^2 \in \FF_p$ occurs an odd number of times, the inner sum is $O(p^{5/2})$. Unconditionally, the inner sum is always $O(p^3)$. For any particular $a$, there are at most 5 values of $b$ (the two solutions of either $b^2 = a^2$ or $b^2 = (a + 1)^2$, or the one solution of $b^2 = (b + 1)^2$) for which all elements occur an even number of times, so we find}
        &= \frac{2}{p} (5p^3 + O(p^{7/2})) \\
        &= O(p^{5/2}). \numberthis
    \end{align*}
    Thus, $\|T\| = O(p^{5/4})$, as claimed.
\end{proof}

\begin{theorem}\label{thm:character_sum_2}
    For any multiplicative character $\phi$ of $\mathbb{F}_p$,
    \begin{equation}
        \left|\sum_{x, y \in \mathbb{F}_p} \chi(x(x+1)(x-y)((x+1)-y))\phi(y)\right|
        = \left|\sum_{x, y \in \FF_p} \chi(x(x+1)y(y+1))\phi(x-y)\right|
        \leq 2p.
    \end{equation}
\end{theorem}
\noindent
This result follows directly from prior work on character sum estimates, but we fill in the details below for the sake of completeness.
The case of $\phi = \chi$ was treated by \cite{CI-2000-CubicMomentLFunctions}, who conjectured that the same should hold for all $\phi$.
This conjecture is implicitly proved in the result of \cite{Liu-2002-KloostermanTwistedMoments} that we cite in Proposition~\ref{prop:twisted-kloosterman}.

\begin{proof}
    We first treat one special case: for $\phi$ the trivial character, we have
    \begin{equation}
        \sum_{x, y \in \FF_p} \chi(x(x+1)y(y+1))\psi(x-y) = \left(\sum_x \chi(x)\chi(x + 1)\right)^2 = (-1)^2 = 1.
    \end{equation}
    Thus, let us assume $\phi$ is not trivial.

    We begin with some manipulations similar to those in Section 14 of \cite{CI-2000-CubicMomentLFunctions}.
    First, changing variables $x \leftarrow (x - 1) / 2$ and $y \leftarrow (y - 1) / 2$, we have
    \begin{align*}
        &\hspace{-2cm}\sum_{x, y \in \FF_p} \chi(x(x+1)y(y+1))\phi(x-y) \\
        &= \sum_{x, y} \chi\left(\frac{x^2 - 1}{4}\right) \chi\left(\frac{y^2 - 1}{4}\right) \phi\left(\frac{y - x}{2}\right) \\
        &= \overline{\phi}(4) \sum_{x, y} \chi(x^2 - 1)\chi(y^2 - 1) \phi(2(y - x)) \\
        &= \overline{\phi}(-4)\frac{G(\phi)}{p} \sum_{a, x, y} \chi(x^2 - 1)\chi(y^2 - 1) \overline{\phi}(a) e_p(2a(x - y)) \tag{by Proposition~\ref{prop:mult-add-fourier}} \\
        &= \overline{\phi}(-4) \frac{G(\phi)}{p} \sum_a \overline{\phi}(a) \left| \sum_x \chi(x^2 - 1) e_p(2ax) \right|^2 \\
        &= \overline{\phi}(-4) \frac{G(\phi)}{p} \sum_a \overline{\phi}(a) K(a^2)^2, \tag{by Proposition~\ref{prop:kloosterman-rewrite}}
    \end{align*}
    and applying Proposition~\ref{prop:twisted-kloosterman} and using that $|G(\phi)| = \sqrt{p}$ gives the result.
\end{proof}

\subsubsection{Final Steps}

\begin{proof}[Proof of Theorem~\ref{thm:norm_T_441}]
Since $\widetilde{T}$ is invariant under the action of the affine group $\Gamma$, by Proposition~\ref{prop:block_circulant}, we have
\begin{equation}
    \spec(\widetilde{T}) = \bigsqcup_{\psi} \spec(S^{(\psi)}),
\end{equation}
where the disjoint union ranges over the additive characters of $\mathbb{Z}/(p-1)\mathbb{Z}$, $S^{(\psi)} = \sum_{i=0}^{p-2}\psi(i) B^{(i)}$, $B^{(i)}$ are the blocks of $\widetilde{T}$ with entries given by
\begin{equation}
    B^{(i)}_{a,b} = \chi((a - h^ib)(a+1 - h^ib)(a - h^i(b+1))(a+1 - h^i(b+1))),
\end{equation}
and $h\in \mathbb{F}_p^\times$ is fixed to be a generator of $\mathbb{F}_p$. To show $\|\widetilde{T}\| = O\left(p^{\frac{5}{4} }\right)$, it is sufficient to show the same bound for the norms of the smaller matrices $S^{(\psi)}$.

By Proposition \ref{prop:block_circulant}, the all-$1$ vector $\boldsymbol{1}$ is an eigenvector for all $S^{(\psi)}$, and $Q_1S^{(\psi)}Q_1$ has a common set of $p-1$ eigenvalues. By Theorem \ref{thm:character_sum_2}, the eigenvalues of $S^{(\psi)}$ corresponding to $\boldsymbol{1}$ are at most
\begin{align*}
    \quad \left| \sum_{x \in \mathbb{F}_p } S^{(\psi)}_{0,x} \right| &= \left| \sum_{x\in \mathbb{F}_p} \sum_{i=0}^{p-2}\psi(i)B^{(i)}_{0,x} \right| \\
    &= \left| \sum_{x\in \mathbb{F}_p} \sum_{i=0}^{p-2} \chi((-h^ix)(1-h^ix)(-h^i(x+1))(1-h^i(x+1))) \phi(h^i) \right| \\
    &= \left| \sum_{x\in \mathbb{F}_p}\sum_{z \in \mathbb{F}_p^\times } \chi(x(x+1)(1-zx)(1-z(x+1))) \phi(z) \right|
    \intertext{and, changing variables $z \leftarrow 1/y$, we have,}
    &= \left|\sum_{x \in \mathbb{F}_p}\sum_{y \in \mathbb{F}_P^\times } \chi(x(x+1)(x-y)((x+1)-y))\phi^{-1}(y)\right|\\
    &= \left|\sum_{x \in \mathbb{F}_p}\sum_{y \in \mathbb{F}_p } \chi(x(x+1)(x-y)((x+1)-y))\phi^{-1}(y)\right|\\
    &\leq 2p \numberthis
\end{align*}
in absolute value, so $\|Q_0 S^{(\psi)} Q_0 \| \le 2p$ for all $S^{(\psi)}$. By Theorem \ref{thm:character_sum_1}, we have
\begin{equation}
    \|Q_1 S^{(\psi)} Q_1\| = \|Q_1 S^{(\psi_0)} Q_1\| = O(p^{5/4}),
\end{equation}
as the entries of $S^{(\psi_0)}$ are
\begin{align*}
    S^{(\psi_0)}_{x,y} &= \sum_{i=0}^{p-2} B^{(i)}_{x,y}\\
    &= \sum_{i=0}^{p-2} \chi((x-h^iy)(x+1 - h^iy)(x - h^i(y+1))(x+1-h^i(y+1)))\\
    &= \sum_{z \in \mathbb{F}_p^\times } \chi((x-zy)(x+1 - zy)(x - z(y+1))(x+1-z(y+1)))\\
    &= -\One_{x \not\in \{0, -1\} } + \sum_{z \in \mathbb{F}_p } \chi((x-zy)(x+1 - zy)(x - z(y+1))(x+1-z(y+1))), \numberthis
\end{align*}
which differ from the entries of the matrix in Theorem \ref{thm:character_sum_1} by at most $1$.

In conclusion, $\|S^{\psi}\| \le \max \{\|Q_0 S^{\psi} Q_0 \|, \|Q_1 S^{\psi} Q_1\|\} = O(p^{\frac{5}{4}})$, $\|\widetilde{T}\| = O(p^{\frac{5}{4}})$, and we conclude $\|T^{4,4,1}\| = O(p^{\frac{5}{4}})$.
\end{proof}

\section{Optimality Over Feige-Krauthgamer Pseudomoments}
\label{sec:optimality-fk}

In this section, we show that our lower bound is optimal over those achievable by FK pseudomoments.
To be precise, let us define a new SDP corresponding to this restricted type of pseudomoment, a variant of \eqref{degree_4_SOS}:
\begin{equation}
\label{degree_4_SOS_FK}
    \mathrm{FK}_{4}(G) \colonequals \left\{
    \begin{array}{ll}
    \text{maximize} & \sum_{i = 1}^n M^{0, 1}_{\emptyset, i} \\
    \text{subject to} & M^{r, c} \in \RR^{\binom{[n]}{r} \times \binom{[n]}{c}} \text{ for } r, c \in \{0, 1, 2\}, \\
    & M^{r, c}_{S, T} \text{ depends only on } S \cup T, \\
    & M^{r, c}_{S, T} = 0 \text{ whenever } S \cup T \notin \mathcal{K}(G), \\
    & M^{r, c}_{S, T} \text{ depends only on } |S \cup T| \text{ when } S \cup T \in \mathcal{K}(G), \\
    & M = \begin{bmatrix}
    1 & M^{0, 1} & M^{0, 2} \\
    M^{1, 0} & M^{1, 1} & M^{1,2}\\
    M^{2, 0} & M^{2,1} & M^{2,2}
\end{bmatrix} \succeq 0
\end{array}
    \right\}.
\end{equation}
Since the conditions of this SDP are more restrictive than those of $\SOS_4(G)$, we always have
\begin{equation}
    \SOS_4(G) \geq \mathrm{FK}_4(G).
\end{equation}
Our proof strategy has been to show that $\mathrm{FK}_4(G)$ is large.
However, the following result, the main one of this section, shows a limitation to this approach.

\begin{theorem}\label{thm:FK_optimality}
Over primes $p \equiv 1 \Mod{4}$, $\mathrm{FK}_4(G_p) = \Theta(p^{1/3})$.
\end{theorem}
\noindent
Note that the proof of our main result Theorem~\ref{thm:main_theorem} already showed that $\mathrm{FK}_4(G_p) \gtrsim p^{1/3}$, so it suffices to show a matching upper bound.

Following the corresponding argument from the literature on ER graphs, attributed to Kelner and described in detail in \cite{HKP-2015-PlantedCliqueSOSMPW}, we prove this by contradiction.
Namely, we show that if the value of $\mathrm{FK}_4(G_p)$ is too large, then the positive semidefiniteness constraint would have to be violated.

In our assumption for the sake of contradiction, we will only consider FK pseudomoments specified by $\alpha_1, \alpha_2, \alpha_3, \alpha_4$ that achieve an objective value of at least $c\cdot p^{\frac{1}{3} }$, for some large $c > 0$. Note also that under the FK pseudomoments, the value of the program is $p \alpha_1$, so this amounts to a lower bound on $\alpha_1$. Under this assumption, we first establish some preliminary propositions relating the $\alpha_i$.

\subsection{Pseudomoment Comparison Bounds}

Let us first define a few vectors and matrices that will be useful for the proof.

\begin{definition}
    Let $\widetilde{\mathbb{E} }$ be a degree $4$ pseudoexpectation that is feasible for the program \eqref{degree_4_SOS}. We define vectors $v^{(0)}, v^{(1)}, v^{(2)} \in \mathbb{R}^{\mathbb{F}_p}$ and matrices $A^{(0)}, A^{(1)}, A^{(2)} \in \mathbb{R}^{\mathbb{F}_p \times \mathbb{F}_p } $, and $B^{(0)}, B^{(1)}, B^{(2)} \in \mathbb{R}^{\{\emptyset\} \sqcup \mathbb{F}_p \times \{\emptyset\} \sqcup \mathbb{F}_p  }$ as follows:
    \begin{align}
    v^{(0)}_a &= \widetilde{\mathbb{E} }[x_a], \\
        v^{(1)}_a &= \widetilde{\mathbb{E} }[x_0x_a], \\
        v^{(2)}_a &= \widetilde{\mathbb{E} }\left[\left(\sum_{i\in \mathbb{F}_p } x_i\right)^2 x_a\right], \\
        A^{(0)}_{a,b} &= \widetilde{\mathbb{E} }[x_ax_b], \\
        A^{(1)}_{a,b} &= \widetilde{\mathbb{E} }[x_0x_ax_b],\\
        A^{(2)}_{a,b} &= \widetilde{\mathbb{E} }\left[\left(\sum_{i \in \mathbb{F}_p } x_i\right)^2 x_ax_b\right],
        \intertext{}
        B^{(0)} &= \begin{bmatrix}
           1 & {v^{(0)}}^\top\\
           v^{(0)} & A^{(0)}
        \end{bmatrix}, \\
        B^{(1)} &= \begin{bmatrix}
            \widetilde{\mathbb{E} }[x_0] & {v^{(1)}}^\top\\
            v^{(1)} & A^{(1)}
        \end{bmatrix}, \\
        B^{(2)} &= \begin{bmatrix}
            \widetilde{\mathbb{E} }\left[\left(\sum_{i \in \mathbb{F}_p }x_i\right)^2 \right] & {v^{(2)}}^\top\\
            v^{(2)} & A^{(2)}
        \end{bmatrix}.
    \end{align}
\end{definition}

\begin{proposition}
    $A^{(i)} \succeq 0$ and $B^{(i)} \succeq 0$ for each $i \in \{0, 1, 2\}$.
\end{proposition}
\begin{proof}
    Since $A^{(i)}$ is a principal submatrix of $B^{(i)}$ for each $i$, it suffices to consider the $B^{(i)}$.
    $B^{(0)}$ is a principal submatrix of the pseudomoment matrix of $\tEE$.
    For the other two cases, let $v = (v_{\emptyset}, v_0, \dots, v_{p - 1}) \in \RR^{\{\emptyset\} \sqcup \FF_p}$.
    Write $w = (v_0, \dots, v_{p - 1}) \in \RR^{\FF_p}$.
    We then have
    \begin{equation}
        v^{\top} B^{(1)} v = \tEE\left[x_0 (v_{\emptyset} + \langle w, x \rangle)^2 \right] = \tEE\left[x_0^2 (v_{\emptyset} + \langle w, x \rangle)^2 \right] \geq 0,
    \end{equation}
    where we have used that $\tEE$ is non-negative on squares and respects the Boolean constraint $x_0^2 = x_0$.
    Similarly,
    \begin{equation}
        v^{\top}B^{(2)} v = \tEE\left[\left(\sum_{i \in \FF_p} x_i^2\right)(v_{\emptyset} + \langle w, x \rangle)^2 \right] \geq 0,
    \end{equation}
    completing the proof.
\end{proof}



\begin{proposition}
    \label{prop:fk-pseudomoment-sums}
    For all $p$, for any pseudoexpectation $\tEE$ for the program \eqref{degree_4_SOS} with FK pseudomoments, the following hold:
    \begin{align}
        \sum_{i \in \mathbb{F}_p\setminus \{0,1\} } \widetilde{\mathbb{E} }[x_0x_1x_i] &= \frac{p-5}{4}\alpha_3 \\
        \sum_{i,j \in \mathbb{F}_p\setminus \{0,1\}} \widetilde{\mathbb{E} }[x_0x_1x_ix_j] &= \left(\frac{(p-2)(p-3)}{32} + O(p^{3/2})\right)\alpha_4 + \frac{p-5 }{4}\alpha_3\\
        \sum_{i \in \mathbb{F}_p} \widetilde{\mathbb{E} }[x_0x_i] &= \alpha_1 + \frac{(p-1)}{2}\alpha_2\\
        \sum_{i,j \in \mathbb{F}_p} \widetilde{\mathbb{E} }[x_0 x_i x_j] &= \alpha_1 + \frac{3(p-1)}{2}\alpha_2 + \frac{(p-1)(p-5)}{8}\alpha_3.
    \end{align}
\end{proposition}
\begin{proof}
    These claims can all be shown using elementary counting arguments and Weil's bound (our Theorem~\ref{thm:weil}). We only prove the first two claims.

    The first sum is equal to $\alpha_3$ times the number of triangles in $G_p$ containing the edge $\{0,1\}$.
    The latter equals $\frac{p-5}{4}$, which is the number of common neighbors of $0$ and $1$, and indeed of any pair of adjacent vertices, in $G_p$ (see Proposition~\ref{fact:regularity}).

    In the second sum, the terms with $i = j$ make a contribution of $\sum_{i \in \mathbb{F}_p\setminus \{0,1\} } \widetilde{\mathbb{E} }[x_0x_1x_i] = \frac{p-5}{4}\alpha_3$, the value of the first sum. The terms with $i\ne j$, using Weil's bound, make a contribution of $(\frac{(p-2)(p-3)}{32} + O(p^{\frac{3}{2} }))\alpha_4$.
\end{proof}

\begin{proposition}\label{prop:alpha_lower_bound}
    Consider an infinite sequence of primes $p$ and FK pseudomoments for each such $p$ satisfying $p\alpha_1 \geq cp^{1/3}$ for some $c > 0$ (i.e., that a collection of FK pseudomoments achieve an objective value of order $\Omega(p^{1/3})$).
    Then, over this sequence,
    \begin{align}
        \alpha_2 &= \Omega(\alpha_1^2),\\
        \alpha_3 &= \Omega\left(\frac{\alpha_2^2 }{\alpha_1 } \right),\\
       \alpha_4 &= \Omega\left(\frac{\alpha_3^2 }{\alpha_2 } \right), \\
       \alpha_1 &= o(p \alpha_2),\\
        \alpha_2 &= o(p \alpha_3),\\
        \alpha_3 &= o(p\alpha_4).
    \end{align}
\end{proposition}
\begin{proof}
    Since $B^{(0)} \succeq 0$, for the Schur complement with respect to the upper left $1 \times 1$ block we also have $A^{(0)} - v^{(0)}{v^{(0)}}^\top \succeq 0$.
    Thus,
    \begin{align*}
        0 &\le \left\langle A^{(0)} - v^{(0)}{v^{(0)}}^\top , J \right\rangle\\
        &= \sum_{a,b \in \mathbb{F}_p } \widetilde{\mathbb{E} }[x_ax_b] - \widetilde{\mathbb{E} }[x_a]\widetilde{\mathbb{E} }[x_b]\\
        &= p\alpha_1 + \frac{p(p-1)}{2}\alpha_2 - p^2 \alpha_1^2, \numberthis
    \end{align*}
    and rearranging this we find
    \begin{align*}
    \alpha_2 &\ge \frac{2p}{(p-1)}\alpha_1^2\left(1 - \frac{1}{p\alpha_1}\right)
    \intertext{Here, since $p \alpha_1 \ge c \cdot p^{\frac{1}{3} }$ by assumption, we may continue}
    &\geq \frac{2p}{(p-1)}\alpha_1^2\left(1 - o(1)\right) \\
    &= \Omega(\alpha_1^2). \numberthis
    \end{align*}
    Also, rearranging differently, we have
    \begin{equation}
        \alpha_1 \leq \frac{1}{p}\alpha_1 + \frac{1}{2}\frac{\alpha_2}{\alpha_1} = O\left(\frac{\alpha_2}{\alpha_1}\right) = o(p\alpha_2).
    \end{equation}

    We carry out a similar calculation for $B^{(1)}$:
    \begin{align*}
        0 &\le \left\langle A^{(1)} - \frac{1}{\widetilde{\mathbb{E} }[x_0] } v^{(1)}{v^{(1)}}^\top , J \right\rangle\\
        &= \sum_{a,b \in \mathbb{F}_p } \widetilde{\mathbb{E} }[x_0x_ax_b] - \frac{1}{\widetilde{\mathbb{E} }[x_0] }\widetilde{\mathbb{E} }[x_0x_a]\widetilde{\mathbb{E} }[x_0x_b]\\
        &= \alpha_1 + \frac{3(p-1) }{2}\alpha_2 + \frac{(p-1)(p-5) }{8} \alpha_3 - \frac{1}{\alpha_1 }\left(\alpha_1 + \frac{p-1}{2}\alpha_2 \right)^2\\
        &= \frac{p-1}{2}\alpha_2 + \frac{(p-1)(p-5)}{8}\alpha_3 - \frac{(p-1)^2}{4} \frac{\alpha_2^2}{\alpha_1 }. \numberthis
    \end{align*}
    Rearranging,
    \begin{align*}
        \alpha_3 &\ge \frac{2(p-1)}{p-5} \frac{\alpha_2^2}{\alpha_1} \left(1 - \frac{2\alpha_1}{(p-1)\alpha_2}\right)
        \intertext{and, since $\alpha_1 = o(p\alpha_2)$,}
        &\geq \frac{2(p-1)}{p-5} \frac{\alpha_2^2}{\alpha_1} \left(1 - o(1)\right) \\
        &= \Omega\left(\frac{\alpha_2^2}{\alpha_1} \right). \numberthis
    \end{align*}
    Rearranging differently as before, we also find $\alpha_2 = O\left(\frac{\alpha_1}{\alpha_2} \alpha_3\right) = o(p\alpha_3 )$.

    Finally, for $B^{(2)}$, let us write $s(x) \colonequals\sum_{i \in \FF_p} x_i$.
    We then have
    \begin{align*}
        0 &\le \left\langle A^{(2)} - \frac{1}{\widetilde{\mathbb{E} }[s(x)^2] } v^{(2)}{v^{(2)}}^\top , J \right\rangle\\
        &= \sum_{a,b\in \mathbb{F}_p } \left(\widetilde{\mathbb{E} }[s(x)^2x_ax_b ] - \frac{1}{\widetilde{\mathbb{E} }[s(x)^2]}\widetilde{\mathbb{E} }[s(x)^2 x_a]\widetilde{\mathbb{E} }[s(x)^2x_b ]\right) \\
        &= \frac{p(p-1)}{2} \widetilde{\mathbb{E} }[x_0x_1s(x)^2 ] - \frac{p^2}{p\alpha_1 + \frac{p(p-1)}{2}\alpha_2 } \left(\widetilde{\mathbb{E} }[x_0s(x)^2 ]\right)^2\\
        &= \frac{p(p-1)}{2} \widetilde{\mathbb{E} }\left[4x_0x_1 + 4x_0x_1s(x) + x_0x_1s(x)^2 \right] - \frac{p^2}{p\alpha_1 + \frac{p(p-1)}{2}\alpha_2 } \left(\widetilde{\mathbb{E} }\left[x_0 + 2x_0s(x) + x_0s(x)^2 \right]\right)^2\\
        &= 2p(p-1)\alpha_2 + \frac{p(p-1)(p-5)}{2}\alpha_3 + \left(\frac{p(p-1)(p-2)(p-3)}{64} + O(p^{7/2})\right)\alpha_4 \\
        &\quad + \frac{p(p-1)(p-5)}{8}\alpha_3 - \frac{p^2}{p\alpha_1 + \frac{p(p-1)}{2}\alpha_2 }\left(\alpha_1 + (p-1)\alpha_2 + \frac{p-1}{2}\alpha_2 + \frac{(p-1)(p-5)}{8}\alpha_3 \right)^2. \numberthis
    \end{align*}
    Since $\alpha_1 = o(p\alpha_2)$ and $\alpha_2 = o(p\alpha_3 )$, rearranging we have
    \begin{align*}
        \alpha_4 &\ge \frac{64(1-o(1))}{p(p-1)(p-2)(p-3)}  \left(\frac{2(1-o(1))}{\alpha_2} \frac{(p-1)^2(p-5)^2}{64}\alpha_3^2 - \frac{5p(p-1)(p-5)}{8}\alpha_3- 2p(p-1)\alpha_2 \right) \\
        &= \frac{2(p-1)(p-5)^2}{p(p-2)(p-3)} \frac{\alpha_3^2}{\alpha_2 } (1 - o(1)) \\
        &= \Omega\left(\frac{\alpha_3^2}{\alpha_2} \right), \numberthis
    \end{align*}
    and rearranging differently gives $\alpha_3 = O\left( \frac{\alpha_2}{\alpha_3} \alpha_4 \right) = o(p \alpha_4)$, completing the proof.
\end{proof}

\begin{proposition}\label{prop:alpha_upper_bound}
    Under the same assumptions as Proposition~\ref{prop:alpha_lower_bound},
    \begin{align}
        \alpha_2 &= O\left(\frac{1}{\sqrt{p}}\alpha_1\right)\\
        \alpha_3 &= O\left(\frac{1}{\sqrt{p}}\alpha_2\right)\\
        \alpha_4 &= O\left(\frac{1}{\sqrt{p}}\alpha_3\right)
    \end{align}
\end{proposition}

\begin{proof}
    For $t \in \mathbb{F}_p$, define the vector $v_t = (e_p(0), e_p(t), e_p(2t), \dots, e_p((p-1)t))^\top \in \mathbb{C}^{\mathbb{F}_p }$.
    Define the matrix $V \in \mathbb{C}^{\mathbb{F}_p \times \mathbb{F}_p }$ as
    \begin{align*}
        V &\colonequals \sum_{t \in \mathbb{F}_p^\times \setminus (\mathbb{F}_p^\times )^2 } v_t v_t^* \\
        &= \frac{p}{2} I - \frac{1}{2} J - \frac{\sqrt{p}}{2} S_{G_p} \\
        &\succeq 0. \numberthis
    \end{align*}
    Then, we have
    \begin{align*}
        0 &\le \langle A^{(0)}, V\rangle\\
        &= \sum_{a,b \in \mathbb{F}_p } \widetilde{\mathbb{E} }[x_ax_b] \left( \frac{p}{2}\One{\{a=b\}} - \frac{1 + \sqrt{p}\cdot \chi(a-b) }{2} \right)\\
        &= \frac{p(p-1)}{2}\alpha_1 - \frac{p(p-1)}{2}\frac{1+ \sqrt{p}}{2}\alpha_2, \numberthis
    \end{align*}
    and rearranging this gives
    \begin{equation}
        \alpha_2 \le \frac{2}{1 + \sqrt{p}} \alpha_1,
    \end{equation}
    so $\alpha_2 = O\left(p^{-1/2}\alpha_1 \right)$.

    Similarly,
    \begin{align*}
        0 &\le \langle A^{(1)}, V\rangle\\
        &= \sum_{a,b \in \mathbb{F}_p } \widetilde{\mathbb{E} }[x_0x_ax_b] \left( \frac{p}{2}\One{\{a=b\}} - \frac{1 + \sqrt{p} \cdot \chi(a-b) }{2} \right)\\
        &= \frac{p-1}{2}\alpha_1 + \frac{(p-1)^2}{4}\alpha_2 - \sum_{a,b \in \mathbb{F}_p: a\ne b } \widetilde{\mathbb{E} }[x_0x_ax_b] \frac{1 + \sqrt{p}\cdot \chi(a-b) }{2}\\
        &= \frac{p-1}{2}\alpha_1 + \frac{(p-1)^2}{4}\alpha_2 -  \frac{(1 + \sqrt{p})(p-1) }{2}\alpha_2 - \frac{(1+ \sqrt{p})(p-1)(p-5)}{16} \alpha_3, \numberthis
    \end{align*}
    which gives
    \begin{align*}
        \alpha_3 &\le \frac{16}{(1+ \sqrt{p})(p-1)(p-5)}\left(\frac{p-1}{2}\alpha_1 + \frac{(p-1)^2}{4}\alpha_2 \right)
        \intertext{Since, by Proposition \ref{prop:alpha_lower_bound}, $\alpha_1 = o(p\alpha_2)$, we further have}
        &\leq \frac{16}{(1+ \sqrt{p})(p-1)(p-5)}\left(\frac{(p-1)^2(1 + o(1))}{4}\alpha_2 \right), \numberthis
    \end{align*}
    so $\alpha_3 = O\left(p^{-1/2}\alpha_2 \right)$.

    Finally, we have
    \begin{align*}
        0 &\le \langle A^{(2)}, V \rangle\\
        &= \sum_{a,b \in \mathbb{F}_p } \widetilde{\mathbb{E} }\left[\left(\sum_{i \in \mathbb{F}_p } x_i \right)^2 x_ax_b \right] \left( \frac{p}{2}\One{\{a=b\}} - \frac{1 + \sqrt{p}\cdot \chi(a-b) }{2} \right)\\
        &= \sum_{a\in \mathbb{F}_p } \left(\widetilde{\mathbb{E} }[x_a] + 2\sum_{i \in \mathbb{F}_p \setminus \{a\} } \widetilde{\mathbb{E} }[x_ax_i] + \sum_{i,j\in \mathbb{F}_p \setminus \{a\} } \widetilde{\mathbb{E} }[x_ax_ix_j] \right) \frac{p-1}{2}\\
        &\quad -\sum_{a,b\in \mathbb{F}_p: a\ne b } \left(4\widetilde{\mathbb{E} }[x_ax_b] + 4\sum_{i\in \mathbb{F}_p \setminus \{a,b\} } \widetilde{\mathbb{E} }[x_ax_bx_i] + \sum_{i,j \in \mathbb{F}_p\setminus \{a,b\} } \widetilde{\mathbb{E} }[x_ax_bx_ix_j] \right) \frac{1 + \sqrt{p}\cdot\chi(a-b) }{2}\\
        &= S_1 - S_2, \numberthis \label{eq:S1-S2}
    \end{align*}
    where we write
    \begin{align}
        S_1 &\colonequals \sum_{a\in \mathbb{F}_p } \left(\widetilde{\mathbb{E} }[x_a] + 2\sum_{i \in \mathbb{F}_p \setminus \{a\} } \widetilde{\mathbb{E} }[x_ax_i] + \sum_{i,j\in \mathbb{F}_p \setminus \{a\} } \widetilde{\mathbb{E} }[x_ax_ix_j] \right) \frac{p-1}{2}, \\
        S_2 &\colonequals \sum_{a,b\in \mathbb{F}_p: a\ne b } \left(4\widetilde{\mathbb{E} }[x_ax_b] + 4\sum_{i\in \mathbb{F}_p \setminus \{a,b\} } \widetilde{\mathbb{E} }[x_ax_bx_i] + \sum_{i,j \in \mathbb{F}_p\setminus \{a,b\} } \widetilde{\mathbb{E} }[x_ax_bx_ix_j] \right) \frac{1 + \sqrt{p}\cdot \chi(a-b) }{2}.
    \end{align}

    From the automorphism group of $G_p$ (see Proposition~\ref{fact:automorphisms}), we notice that the value of the terms in $S_1$ does not depend on $a \in \mathbb{F}_p$. So, we may write $S_1$ as
    \begin{align*}
        S_1 &= \frac{p(p-1)}{2} \left(\widetilde{\mathbb{E} }[x_0] + 2\sum_{i \in \mathbb{F}_p \setminus \{0\} } \widetilde{\mathbb{E} }[x_0x_i] + \sum_{i,j\in \mathbb{F}_p\setminus \{0\} } \widetilde{\mathbb{E} }[x_0x_ix_j]  \right)\\
        &= \frac{p(p-1)}{2}\alpha_1 + \frac{p(p-1)^2}{2} \alpha_2 + \left( \frac{p(p-1)^2(p-5)}{16} \alpha_3 + \frac{p(p-1)^2}{4}\alpha_2 \right)
        \intertext{By Proposition \ref{prop:alpha_lower_bound}, we have $\alpha_1 = o(p^2 \alpha_3)$ and $\alpha_2 = o(p\alpha_3)$, so the main contribution in $S_1$ comes from the $\alpha_3$ term:}
        &= (1+o(1))\frac{p(p-1)^2(p-5)}{16} \alpha_3. \numberthis \label{eq:S1}
    \end{align*}

    Similarly, we notice that each term in the sum of $S_2$ vanishes when $\{a,b\}$ is not an edge in $G_p$, and when $\{a,b\}$ is an edge in $G_p$ does not depend on $a$ and $b$. So, we may write $S_2$ as
    \begin{align*}
        S_2 &= \frac{p(p-1)}{2} \left(4 \widetilde{\mathbb{E} }[x_0x_1] + 4 \sum_{i\in \mathbb{F}_p \setminus \{0,1\} } \widetilde{\mathbb{E} }[x_0x_1x_i] + \sum_{i,j \in \mathbb{F}_p \setminus \{0,1\} } \widetilde{\mathbb{E} }[x_0x_1x_ix_j] \right) \frac{1 + \sqrt{p} }{2} \\
        &= \frac{(1+\sqrt{p} )p(p-1)}{4} \left(4\alpha_2 + 4 \frac{p-5}{4}\alpha_3 + \left(\frac{(p-2)(p-3)}{32} + O(p^{3/2}) \right)\alpha_4 + \frac{p-5}{4}\alpha_3  \right)
        \intertext{By Proposition \ref{prop:alpha_lower_bound}, we have $\alpha_2 = o(p^2\alpha_4)$ and $\alpha_3 = o(p\alpha_4)$, so the main contribution in $S_2$ comes from the $\alpha_4$ term:}
        &= (1 + o(1))\frac{(1+\sqrt{p} )p(p-1)(p-2)(p-3)}{128}\alpha_4. \numberthis \label{eq:S2}
    \end{align*}

    Combining \eqref{eq:S1-S2}, \eqref{eq:S1}, and \eqref{eq:S2}, we have
    \begin{align*}
        0 &\le S_1 - S_2\\
        &= (1+o(1))\frac{p(p-1)^2(p-5)}{16} \alpha_3 - (1 + o(1))\frac{(1+\sqrt{p} )p(p-1)(p-2)(p-3)}{128}\alpha_4\\
        &= O\left(p^4 \alpha_3\right) - \Omega\left(p^{9/2} \alpha_4 \right), \numberthis
    \end{align*}
    so $\alpha_4 = O(p^{1/2}\alpha_3)$, completing the proof.
\end{proof}

\subsection{Preliminary Character Sum Estimates}

Before proceeding to the proof of Theorem~\ref{thm:FK_optimality}, we also establish some character sum bounds that we will need.

\begin{definition}
    Let $u \in \mathbb{R}^{\binom{\mathbb{F}_p }{2} }$ be the vector defined by
    \begin{equation}
        u_{\{i,j\}} = \chi(ij)(\chi(i-j) + 1).
    \end{equation}
\end{definition}
\begin{proposition}\label{prop:u_norm}
    For any $p \equiv 1 \Mod{4}$,
    \begin{align}
        \|u\|^2 &= O(p^2)\\
        \|(P_0 + P_1)u\|^2 &= O(1).
    \end{align}
\end{proposition}

\begin{proof}
    Since each $u_{\{i,j\}} = O(1)$, we immediately have $\|u\|^2 = O(p^2)$.


    Recall that $P_0 + P_1$ is the orthogonal projection to the subspace
    \begin{equation}
        \mathbb{V}_0 \oplus \mathbb{V}_1 = \left\{w \in \mathbb{R}^{\binom{\mathbb{F}_p}{2} }: w_{\{i,j\}} = v_i + v_j \text{ for some } v \in \mathbb{R}^{\mathbb{F}_p } \text{ and for all } \{i,j\} \in \binom{\mathbb{F}_p}{2} \right\}.
    \end{equation}
    Suppose $\left((P_0 + P_1) u\right)_{\{i,j\}} = v_i + v_j$ for some $v \in \mathbb{R}^{\mathbb{F}_p }$. Then,
    \begin{equation}
        v = \argmin_v \sum_{\{i,j\} \in \binom{\mathbb{F}_p}{2} } \left(u_{\{i,j\}} - (v_i + v_j)\right)^2.
    \end{equation}
    By taking derivatives with respect to each $v_i$ to write the first-order optimality conditions for $v$, we find that, for all $i \in \FF_p$,
    \begin{align*}
        0 &= \sum_{j \in \mathbb{F}_p \setminus \{i\} } \left(u_{\{i,j\}} - v_i - v_j \right)\\
         &= \sum_{j \in \mathbb{F}_p \setminus \{i\} } u_{\{i,j\} } - \sum_{j \in \mathbb{F}_p } v_j - (p-2)v_i. \numberthis
    \end{align*}
    Adding these over all $i \in \FF_p$, we find
    \begin{align*}
        \sum_{j \in \mathbb{F}_p } v_j = \frac{\sum_{i \in \mathbb{F}_p } \sum_{j \in \mathbb{F}_p \setminus \{i\} } u_{\{i,j\}} }{2p-2} = - \frac{p}{2p-2}, \numberthis
    \end{align*}
    and substituting this into each individual condition, we solve for $v_i$ and find
    \begin{align*}
        v_i &= \frac{\sum_{j \in \mathbb{F}_p \setminus \{i\} } u_{\{i,j\}}  - \sum_{j\in \mathbb{F}_p }v_j}{p-2}\\
        &= \frac{\chi(i)(-1-\chi(i)) + \frac{p}{2p-2}}{p-2}. \numberthis
    \end{align*}
    In particular, $v_i = O(1/p)$, and thus $\|(P_0 + P_1)u\|^2 = \sum_{\{i,j\} } (v_i + v_j)^2 = O(1)$.
\end{proof}

\begin{proposition}\label{prop:quadratic_bound}
    For the graph matrices $T^{i, j, k}$ as defined in Table~\ref{table:graph_matrices}, we have
    \begin{align}
        u^* T^{3,0,1} u &= O(p^2)\\
        u^* T^{4,0,1} u &= O(p^2)\\
        u^* T^{4,2,1} u &= O(p^{5/2})\\
        u^* T^{4,2,2} u &= O(p^{5/2})\\
        u^* T^{4,1,1} u &= O(p^{5/2}).
    \end{align}
\end{proposition}
\begin{proof}
    We may bound
    \begin{align*}
        \quad \left| u^* T^{4,1,1} u \right| \le \begin{bmatrix}
            a^{(0)} & a^{(1)}
        \end{bmatrix} \begin{bmatrix}
            b^{(0,0)} & b^{(0,1)}\\
            b^{(1,0)} & b^{(1,1)}
        \end{bmatrix} \begin{bmatrix}
            a^{(0)} \\
            a^{(1)}
        \end{bmatrix}, \numberthis
    \end{align*}
    where $a^{(i)} = \|M^{(i)} u\|$ and $b^{(i,j)} = \|M^{(i)} T^{4,1,1} M^{(j)}\|$, and $M^{(0)} = P_0 + P_1$ and $M^{(1)} = P_2$. By Propositions \ref{prop:T_4ij} and \ref{prop:u_norm}, we have
    \begin{align*}
        \quad \left| u^* T^{4,1,1} u \right| &\le \begin{bmatrix}
            O(1) & O(p)
        \end{bmatrix} \begin{bmatrix}
            O(p^{3/2}) & O(p^{3/2})\\
            O(p^{3/2}) & O(p^{1/2})
        \end{bmatrix} \begin{bmatrix}
            O(1) \\
            O(p)
        \end{bmatrix}\\
        &= O(p^{5/2}), \numberthis
    \end{align*}
    as claimed.
    The other claims follow similarly using the other norm bounds on graph matrices from Section~\ref{sec:graph_matrix_norm_bounds}.
\end{proof}

\begin{proposition}\label{prop:quadratic_bound_2}
    For all $p \equiv 1 \Mod{4}$,
    \begin{align*}
        u^* T^{4,4,1} u &= \sum_{\substack{\{a,b\}, \{c,d\} \in \binom{\mathbb{F}_p }{2} \\ \{a,b\} \cap \{c,d\} = \emptyset }} \chi(abcd) \chi(a-b)\chi(a-c)\chi(a-d)\chi(b-c)\chi(b-d)\chi(c-d)\\
        &= O(p^3). \numberthis
    \end{align*}
\end{proposition}
\begin{proof}
    We start by noting that the constraint to $\{a, b\} \cap \{c, d\} = \emptyset$ is superfluous. Summing over all 4-tuples $a,b,c,d \in \mathbb{F}_p^\times$ only incurs an overcounting factor of $4$, so we may rewrite
    \begin{align*}
        &\quad \sum_{\substack{\{a,b\}, \{c,d\} \in \binom{\mathbb{F}_p }{2}:\\ \{a,b\} \cap \{c,d\} = \emptyset }} \chi(abcd) \chi(a-b)\chi(a-c)\chi(a-d)\chi(b-c)\chi(b-d)\chi(c-d)\\
        &= \frac{1}{4}\sum_{a,b,c,d\in \mathbb{F}_p^\times } \chi(abcd) \chi(a-b)\chi(a-c)\chi(a-d)\chi(b-c)\chi(b-d)\chi(c-d)\\
        &= \frac{p-1}{4} \sum_{a,b,c\in \mathbb{F}_p^\times } \chi(abc) \chi(a-b)\chi(a-c)\chi(b-c)\chi(a-1)\chi(b-1)\chi(c-1)\\
        &= \frac{p-1}{4} \sum_{a,b,c\in \mathbb{F}_p^\times } \chi(a^{-1}-b^{-1})\chi(a^{-1}-c^{-1})\chi(b^{-1}-c^{-1})\chi(a^{-1} - 1)\chi(b^{-1} - 1)\chi(c^{-1} - 1)
        \intertext{and, changing variables to $x = a^{-1} - 1, y = b^{-1} - 1,$ and $z = c^{-1} - 1$, we have}
        &= \frac{p-1}{4} \sum_{x,y,z\in \mathbb{F}_p \setminus \{-1\} } \chi(xyz)\chi(x-y)\chi(y-z)\chi(z-x)\\
        &= \frac{p-1}{4} \bigg(\sum_{x,y,z\in \mathbb{F}_p } \chi(xyz)\chi(x-y)\chi(y-z)\chi(z-x) \\
        &\hspace{1cm}- 3\sum_{x,y\in \mathbb{F}_p \setminus \{-1\} }\chi(xy)\chi(x-y)\chi(y+1)\chi(x+1)\bigg) \\
        &= \frac{p-1}{4} (S_1 - S_2), \numberthis
    \end{align*}
    where we denote
    \begin{align}
        S_1 &\colonequals \sum_{x,y,z\in \mathbb{F}_p } \chi(xyz)\chi(x-y)\chi(y-z)\chi(z-x), \\
        S_2 &\colonequals 3\sum_{x,y\in \mathbb{F}_p \setminus \{-1\} }\chi(xy)\chi(x-y)\chi(y+1)\chi(x+1).
    \end{align}
    We immediately have $S_2 = O(p^2)$, so it suffices to show the same bound for $S_1$.

    We will use the Gauss sum identity
    \begin{equation}
        \chi(x) = \frac{1}{\sqrt{p}} \sum_{a \in \mathbb{F}_p} \chi(a) e_p\left(ax\right),
    \end{equation}
    which follows from combining Propositions~\ref{prop:mult-add-fourier} and \ref{prop:gauss-sum-norm}, and noting that we have $\chi(-1) = 1$ since $p \equiv 1 \Mod{4}$.
    Using this, we may rewrite $S_1$:
    \begin{align*}
    S_1 &= \sum_{x,y,z\in \mathbb{F}_p } \chi(xyz)\chi(x-y)\chi(y-z)\chi(z-x)\\
    &= \sum_{x,y,z\in \mathbb{F}_p^\times } \chi\left(1- \frac{y}{x}\right)\chi\left(1- \frac{z}{y}\right)\chi\left(1- \frac{x}{z}\right)\\
    &= (p-1) \sum_{\substack{x, y, z \in \FF_p^{\times} \\ xyz = 1}} \chi\left(1- x\right)\chi\left(1- y\right)\chi\left(1- z\right)\\
    &= (p-1)p^{-3/2} \sum_{\substack{x, y, z \in \FF_p^{\times} \\ xyz = 1}} \sum_{a, b, c \in \mathbb{F}_p} \chi(abc) e_p\left((1 - x)a + (1-y)b + (1 - z)c\right) \\
    &= (p-1)p^{-3/2} \sum_{a, b, c \in \mathbb{F}_p} \chi(abc) e_p(a + b + c)  \sum_{\substack{x, y, z \in \FF_p^{\times} \\ xyz = 1}} e_p\left(-ax -by -cz\right) \\
    &= (p-1)p^{-3/2 } \sum_{a,b, c \in \mathbb{F}_p} \chi(abc) e_p(a + b + c)  \overline{K_3(abc)} \\
    &= (p-1)p^{-3/2 } \sum_{d \in \mathbb{F}_p^\times} \chi(d) \overline{K_3(d)} \sum_{abc = d } e_p(a + b + c)  \\
    &= (p-1)p^{-3/2 } \sum_{d \in \mathbb{F}_p^\times} \chi(d) \left|K_3(d)\right|^2. \numberthis
    \end{align*}
    By Proposition~\ref{prop:twisted-kloosterman-2}, we find $|S_1| = O(p^2)$, and the result follows.
\end{proof}

\subsection{Proof of Theorem~\ref{thm:FK_optimality}}

\begin{proof}[Proof of Theorem \ref{thm:FK_optimality}]
    Let us write $k = k(p) \colonequals \mathrm{FK}_4(G_p)$, and let $\alpha_1, \alpha_2, \alpha_3, \alpha_4$ be the FK pseudomoments that achieve this value.
    Suppose for the sake of contradiction that $\mathrm{FK}_4(G_p) \geq c p^{1/3}$ for a large constant $c > 0$ to be specified later.
    Since $k = p\alpha_1$, $\alpha_1 = k/p \geq cp^{-2/3}$.

    In the following computation, without specifying otherwise, by summing over $\{a,b\}$ we mean the summation over all $2$-element subsets $\{a,b\} \in \binom{\mathbb{F}_p}{2}$. Let $C$ be a constant that we will specify later.
    Following the proof strategy of \cite{HKP-2015-PlantedCliqueSOSMPW}, we evaluate the following pseudoexpectation using Proposition~\ref{prop:fk-pseudomoment-sums}:
    \begin{align*}
        0 &\le \widetilde{\mathbb{E} }\left[\left(Ck^2 x_0 - \sum_{\{a,b\}} \chi(ab)x_ax_b\right)^2\right]\\
        &= C^2 k^4 \widetilde{\mathbb{E} }[x_0] - 2Ck^2 \sum_{\{a,b\}} \chi(ab) \widetilde{\mathbb{E} }[x_0x_ax_b] + \widetilde{\mathbb{E} }\left[\left(\sum_{\{a,b\}} \chi(ab)x_ax_b\right)^2\right]\\
        &= C^2 k^4 \frac{k}{p} - 2Ck^2 \frac{(p-1)(p-5)}{16}\alpha_3  + \widetilde{\mathbb{E} }\left[\left(\sum_{\{a,b\}} \chi(ab)x_ax_b\right)^2\right]. \numberthis
    \end{align*}

    Expanding the last term, we get
    \begin{align*}
        &\quad \widetilde{\mathbb{E} }\left[\left(\sum_{\{a,b\}} \chi(ab)x_ax_b\right)^2\right]\\
        &= \sum_{\{a,b\} } \widetilde{\mathbb{E} }[x_ax_b] + \sum_{\substack{\{a,b\}, \{c,d\} \\ |\{a,b\} \cap \{c,d\}| = 1} } \chi(abcd) \widetilde{\mathbb{E} }[x_ax_bx_cx_d] +  \sum_{\substack{\{a,b\}, \{c,d\}\\ \{a,b\} \cap \{c,d\} = \emptyset} } \chi(abcd) \widetilde{\mathbb{E} }[x_ax_bx_cx_d] \numberthis
    \end{align*}
    Observe that the first term evaluates to $\frac{p(p-1)}{4}\alpha_2$. Now, recall the vector $u\in \mathbb{R}^{\binom{\mathbb{F}_p}{2} } $ defined earlier by $u_{\{i,j\}} = \chi(ij)(\chi(i-j)+1)$. We may write the second and third terms as quadratic forms involving the graph matrices defined earlier, evaluated at $u$.
    For the second term, we have
    \begin{align*}
        &\quad \sum_{\substack{\{a,b\}, \{c,d\}:\\ |\{a,b\} \cap \{c,d\}| = 1} } \chi(abcd) \widetilde{\mathbb{E} }[x_ax_bx_cx_d]\\
        &= \sum_{\substack{\{a,b\}, \{a,c\}:\\ b \ne c} } \chi(a^2bc) \widetilde{\mathbb{E} }[x_ax_bx_c]\\
        &= \sum_{\substack{\{a,b\}, \{a,c\}:\\ b \ne c} } \chi(a^2bc)\frac{(\chi(a-b)+1)(\chi(a-c)+1)(\chi(b-c)+1) }{8}\alpha_3\\
        &= \frac{\alpha_3}{8}\sum_{\substack{\{a,b\}, \{a,c\}:\\ b \ne c} } (\chi(b-c)+1)\cdot \chi(ab)(\chi(a-b)+1) \cdot \chi(ac)(\chi(a-c)+1) \\
        &= \frac{\alpha_3}{8} u^* (T^{3,1,1} + T^{3,0,1}) u, \numberthis
        \intertext{and for the third term,}
        &\quad \sum_{\substack{\{a,b\}, \{c,d\}:\\ \{a,b\} \cap \{c,d\} = \emptyset} } \chi(abcd) \widetilde{\mathbb{E} }[x_ax_bx_cx_d]\\
        &= \sum_{\substack{\{a,b\}, \{c,d\}:\\ \{a,b\} \cap \{c,d\} = \emptyset} } \chi(abcd) \frac{\prod_{\{i,j\} \in \binom{\{a,b,c,d\}}{2} } (\chi(i-j) + 1)}{64}\alpha_4\\
        &= \frac{\alpha_4}{64} \sum_{\substack{\{a,b\}, \{c,d\}:\\ \{a,b\} \cap \{c,d\} = \emptyset} }\left(\prod_{(i,j) \in \{a,b\} \times \{c,d\} } (\chi(i-j) + 1) \right)\cdot \chi(ab)(\chi(a-b)+1)\cdot \chi(cd)(\chi(c-d)+1) \\
        &= \frac{\alpha_4}{64} u^* (T^{4,4,1} + T^{4,3,1} + T^{4,2,1} + T^{4,2,2} + T^{4,2,3} + T^{4,1,1} + T^{4,0,1} ) u. \numberthis
    \end{align*}
    Using Propositions \ref{prop:quadratic_bound} and \ref{prop:quadratic_bound_2} and the graph matrix norm bounds in Section \ref{sec:graph_matrix_norm_bounds}, we have
    \begin{align*}
        \left|\sum_{\substack{\{a,b\}, \{c,d\}:\\ |\{a,b\} \cap \{c,d\}| = 1} } \chi(abcd) \widetilde{\mathbb{E} }[x_ax_bx_cx_d]\right|
        &\le \frac{\alpha_3}{8} \left( u^* T^{3,0,1} u + \|T^{3,1,1}\| \|u\|^2 \right)\\
        &= \frac{\alpha_3}{8} \left(O(p^2) + O(\sqrt{p})O\left(p^2 \right)\right)\\
        &= O\left(p^{5/2} \alpha_3\right), \numberthis
    \end{align*}
    and
    \begin{align*}
        &\hspace{-1cm}\sum_{\substack{\{a,b\}, \{c,d\}:\\ \{a,b\} \cap \{c,d\} = \emptyset} } \chi(abcd) \widetilde{\mathbb{E} }[x_ax_bx_cx_d]  \\
        &\le \frac{\alpha_4}{64}\Bigg((\|T^{4,3,1}\| +  \|T^{4,2,3}\| )\|u\|^2
        + u^* (T^{4,0,1} + T^{4,2,1} + T^{4,2,2} + T^{4,1,1} + T^{4,4,1}) u\Bigg)\\
        &= \frac{\alpha_4}{64}\left(O(p)O(p^2) + O(p^{5/2} + p^3) \right)\\
        &= O(p^3 \alpha_4). \numberthis
    \end{align*}
    So, combining these, we have
    \begin{align*}
        \widetilde{\mathbb{E} }\left[\left(\sum_{\{a,b\}} \chi(ab)x_ax_b\right)^2\right] =
        \frac{p(p-1)}{4}\alpha_2 + O\left(p^{5/2} \alpha_3\right) + O(p^3 \alpha_4). \numberthis
    \end{align*}
    By Proposition \ref{prop:alpha_upper_bound}, $O(p^{5/2} \alpha_3 + p^3 \alpha_4) = O(p^2 \alpha_2)$, so
    \begin{align*}
        \widetilde{\mathbb{E} }\left[\left(\sum_{\{a,b\}} \chi(ab)x_ax_b\right)^2\right] = O(p^2 \alpha_2). \numberthis
    \end{align*}

    Thus, substituting into our initial calculation,
    \begin{align*}
        0 &\le C^2 \frac{k^5}{p} - 2Ck^2 \frac{(p-1)(p-5)}{16}\alpha_3  + O(p^2 \alpha_2). \numberthis
    \end{align*}
    By Proposition \ref{prop:alpha_lower_bound}, $\alpha_3 = \Omega\left(\frac{\alpha_2^2}{\alpha_1} \right) = \Omega\left(\frac{p\alpha_2^2}{k} \right)$. Let $C_1, C_2$ be the implied constants in $\alpha_3 = \Omega\left(\frac{p\alpha_2^2}{k} \right)$ and $\widetilde{\mathbb{E} }\left[\left(\sum_{\{a,b\}} \chi(ab)x_ax_b\right)^2\right] = O(p^2 \alpha_2)$. We then have
    \begin{align*}
        0 &\le C^2 \frac{k^5}{p} - 2C C_1 k^2 \frac{(p-1)(p-5)}{16} \frac{p\alpha_2^2}{k}  + C_2 p^2 \alpha_2\\
        &\le C^2 \frac{k^5}{p} - C C_1' k p^3 \alpha_2^2 + C_2 p^2 \alpha_2, \numberthis
    \end{align*}
    where $C_1'$ is a constant depending on $C_1$. Note that $C_1', C_2$ do not depend on $c$.

    Consider the ratio
    \begin{align*}
        &\quad \frac{C^2\frac{k^5}{p} + C_2p^2\alpha_2 }{CC_1' kp^3 \alpha_2^2}\\
        &= \frac{C}{C_1'} \cdot \frac{k^4}{p^4} \cdot \frac{1}{\alpha_2^2} + \frac{C_2}{CC_1'} \cdot \frac{1}{kp} \cdot  \frac{1}{\alpha_2}. \numberthis
    \end{align*}
    Note that if this ratio is less than $1$, then the expression above $C^2 \frac{k^5}{p} - C C_1' k p^3 \alpha_2^2 + C_2 p^2 \alpha_2$ is less than $0$.

    Recall $\alpha_2 = \Omega(\alpha_1^2) = \Omega\left(\frac{k^2}{p^2} \right)$ by Proposition \ref{prop:alpha_lower_bound}. Therefore, by choosing $C$ to be a sufficiently small constant,
    \begin{align*}
        \frac{C}{C_1'} \cdot \frac{k^4}{p^4} \cdot \frac{1}{\alpha_2^2} = \frac{C}{C_1'} \cdot \frac{k^4}{p^4} \cdot O\left(\frac{p^4}{k^4} \right) \le \frac{1}{4}. \numberthis
    \end{align*}
    On the other hand, since $k \ge c p^{1/3 }$ for some constant $c$ by assumption, choosing $c$ large enough we have
    \begin{align*}
        \frac{C_2}{CC_1'} \cdot \frac{1}{kp} \cdot  \frac{1}{\alpha_2} = \frac{C_2}{CC_1'} \cdot O\left(\frac{p}{k^3} \right) \le \frac{1}{4}. \numberthis
    \end{align*}
    Therefore, under a sufficiently small constant $C$, we derive a contradiction
    \begin{align*}
        0 \le \widetilde{\mathbb{E} }\left[\left(Ck^2 x_0 - \sum_{\{a,b\}} \chi(ab)x_ax_b\right)^2\right] < 0. \numberthis
    \end{align*}
    We conclude that the value of the degree $4$ SOS on Paley graphs $G_p$ restricted to FK pseudomoments is $\mathrm{FK}_4(G_p) = O(p^{1/3} )$.
\end{proof}

\clearpage

\begin{figure}[!htb]
    \begin{center}
    \includegraphics[scale=0.53]{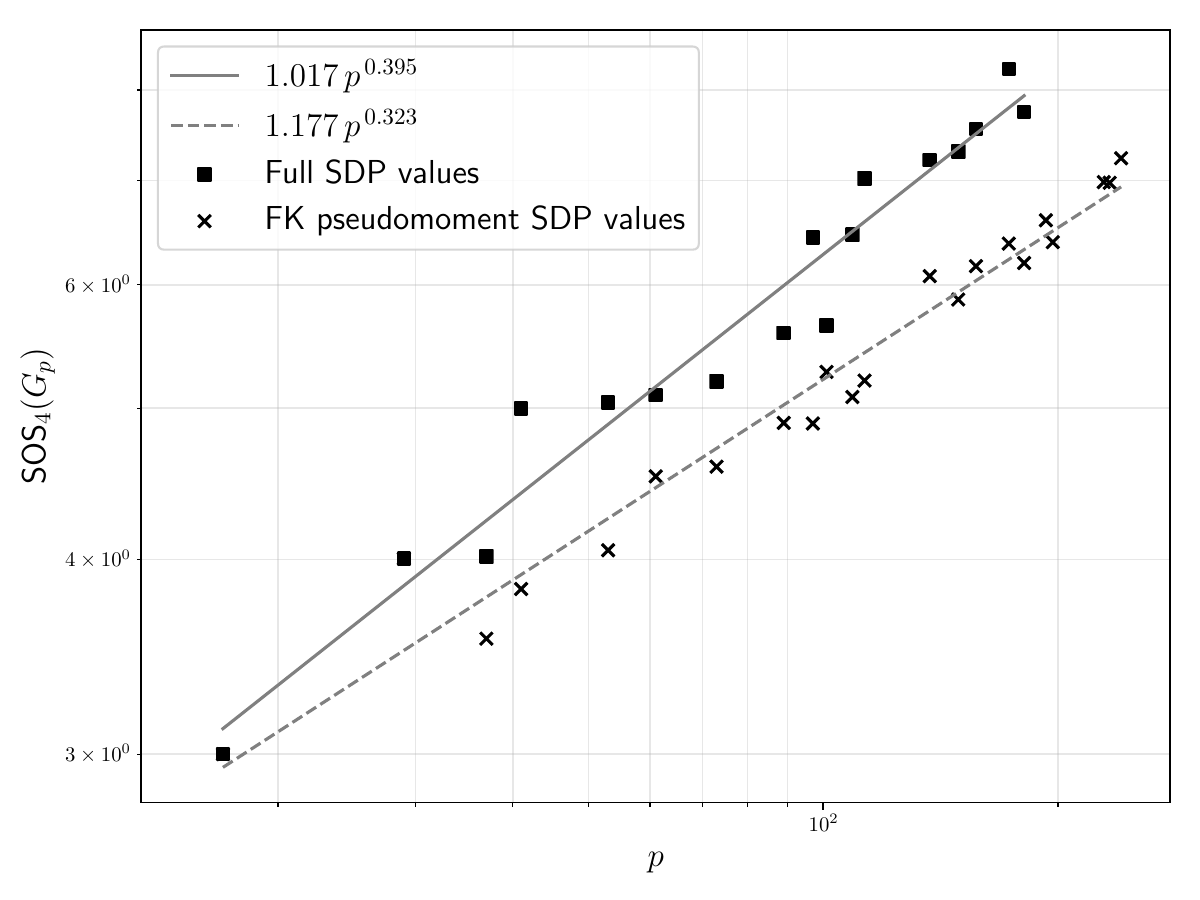}
    \end{center}
    \vspace{-1.5em}
    \caption{For primes $5 \leq p \leq 250$, we present the value of $\SOS_4(G_p)$ and the value of $\mathrm{FK}_4(G_p)$ (where the semidefinite program is restricted to optimize over only FK pseudomoments). We fit power models $ap^b$ to the data and plot the results as well.}
    \label{fig:sdp-numerics}
\end{figure}

\begin{figure}[!htb]
    \begin{center}
        \includegraphics[scale=0.53]{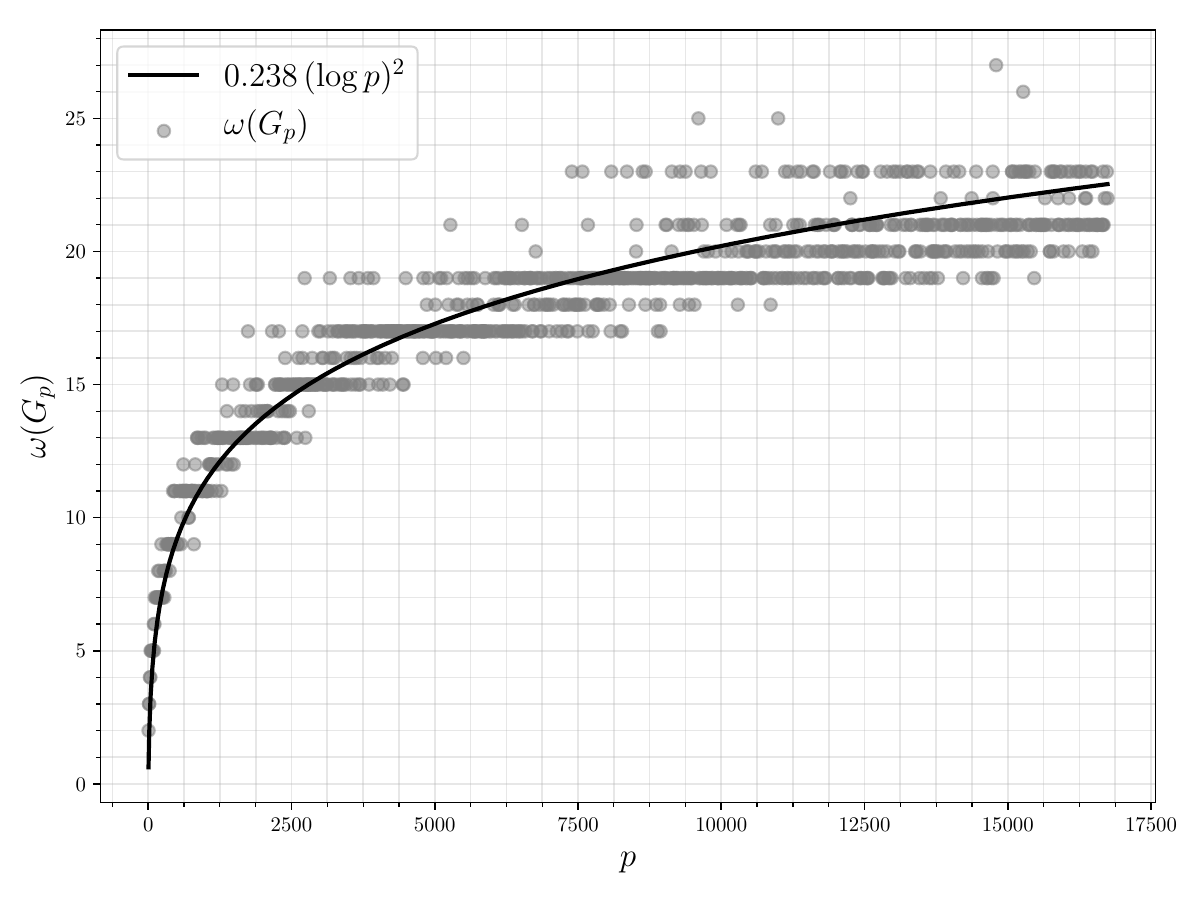}
    \end{center}
    \vspace{-1.5em}
    \caption{For primes $5 \leq p \leq 16741$, we present computations of the true clique number $\omega(G_p)$ (taken from \cite{Shearer-1986-LowerBoundsDiagonalRamsey} and its online supplementary materials). We fit a model $a (\log p)^2$ to the data and plot the results as well.}
    \label{fig:clique-numerics}
\end{figure}

\clearpage

\section{Numerical Experiments}
\label{sec:numerical}

Given our results in Theorems~\ref{thm:main_theorem} and \ref{thm:FK_optimality}, it is natural to ask whether a better lower bound technique than working with FK pseudomoments might prove an optimal lower bound of the form $\SOS_4(G_p) = \Omega(p^{1/2})$.
In Figure~\ref{fig:sdp-numerics}, we present some surprising numerical results suggesting that this is \emph{not} the case.
Namely, in addition to the true values of $\omega(G_p)$, we plot the values of $\SOS_4(G_p)$ (the ``full SDP'') and of $\mathrm{FK}_4(G_p)$ (the ``FK pseudomoment SDP'') on a log-log plot, and fit lines to these results.\footnote{These SDPs are solved using the \texttt{Mosek} solver through the \texttt{CVXPY} interface for \texttt{Python}.}

These results for $\mathrm{FK}_4(G_p)$ confirm the statement of Theorem~\ref{thm:FK_optimality}, with an estimated scaling of $\mathrm{FK}_4(G_p) \sim p^{0.323}$, close to our result showing that $\mathrm{FK}_4(G_p) \sim p^{1/3}$.
For $\SOS_4(G_p)$, the results still indicate a scaling below $p^{1/2}$, estimated at $\SOS_4(G_p) \sim p^{0.395}$.
Based on these results, it seems reasonable to conjecture that $\SOS_4(G_p) = O(p^{1/2 - \epsilon})$ for some $\epsilon > 0$.
We note that this prediction is compatible with that of \cite{KM-2023-SDPCliquePaley}, who, based experiments solving a weaker SDP than degree 4 SOS as proposed by \cite{GLV-2009-BlockDiagonalSDPPaleyClique}, experimentally found that $\SOS_4(G_p) \lesssim p^{0.456}$.

Unfortunately, the size of these SDPs prevents us from exceeding roughly $p \approx 250$, so we cannot be very confident in our scaling estimates.
However, we believe that better understanding whether degree 4 (or higher) SOS can break the ``$\sqrt{p}$ barrier'' is an important question for future study.

\section{General Graph Matrix Norm Bounds Do Not Derandomize}
\label{sec:failure-graphical}

In this section, we give a simple example of a graph matrix for which the norm bound of \cite{AMP-2016-GraphMatrices} for ER graphs fails to hold for Paley graphs.
Since the bound of \cite{AMP-2016-GraphMatrices} is a crucial ingredient in the proof of the $\Omega(p^{1/2})$ SOS lower bound of \cite{BHKKMP-2019-PlantedClique}, we take this as some evidence that a sufficiently high degree of SOS can prove a bound of the form $\omega(G_p) \leq O(p^{1/2 - \epsilon})$.
In particular, this gives some theoretical evidence for the numerical observations in the previous section.

Let $S \in \RR^{n \times n}$ be a symmetric matrix with diagonal equal to zero and off-diagonal entries in $\{\pm 1\}$, the Seidel adjacency matrix of some graph.
We consider the ``diamond-shaped'' graph matrix $M = M(S)$ formed from $S$, whose entries are
\begin{equation}
    M_{xy} = \One{\{x \neq y\}} \sum_{\substack{a, b \in [n] \\ a \neq b}} S_{a,x}S_{a,y} S_{b,x}S_{b,y},
\end{equation}
where we note that we do not need to include the constraints $a, b \notin \{x, y\}$ since these are automatically enacted by having $S_{a,a} = 0$ for all $a$.
See Figure~\ref{fig:graph-matrix-example} for the corresponding shape.

\begin{figure}[!t]
    \begin{center}
    \tikzset{every picture/.style={line width=0.75pt}} 

\begin{tikzpicture}[x=0.75pt,y=0.75pt,yscale=-1,xscale=1]

\draw   (65.2,142.7) .. controls (65.2,140.21) and (67.21,138.2) .. (69.7,138.2) .. controls (72.19,138.2) and (74.2,140.21) .. (74.2,142.7) .. controls (74.2,145.19) and (72.19,147.2) .. (69.7,147.2) .. controls (67.21,147.2) and (65.2,145.19) .. (65.2,142.7) -- cycle ;
\draw   (205,141.5) .. controls (205,139.01) and (207.01,137) .. (209.5,137) .. controls (211.99,137) and (214,139.01) .. (214,141.5) .. controls (214,143.99) and (211.99,146) .. (209.5,146) .. controls (207.01,146) and (205,143.99) .. (205,141.5) -- cycle ;
\draw    (74.2,142.7) -- (135.53,91.03) ;
\draw   (135.53,91.03) .. controls (135.53,88.55) and (137.55,86.53) .. (140.03,86.53) .. controls (142.52,86.53) and (144.53,88.55) .. (144.53,91.03) .. controls (144.53,93.52) and (142.52,95.53) .. (140.03,95.53) .. controls (137.55,95.53) and (135.53,93.52) .. (135.53,91.03) -- cycle ;
\draw   (135.53,192.03) .. controls (135.53,189.55) and (137.55,187.53) .. (140.03,187.53) .. controls (142.52,187.53) and (144.53,189.55) .. (144.53,192.03) .. controls (144.53,194.52) and (142.52,196.53) .. (140.03,196.53) .. controls (137.55,196.53) and (135.53,194.52) .. (135.53,192.03) -- cycle ;
\draw    (144.53,192.03) -- (205.87,140.37) ;
\draw    (144.53,91.03) -- (205.87,140.37) ;
\draw    (74.2,142.7) -- (135.53,192.03) ;
\draw  [color={rgb, 255:red, 74; green, 144; blue, 226 }  ,draw opacity=1 ][dash pattern={on 4.5pt off 4.5pt}] (53.3,142.7) .. controls (53.3,118.32) and (60.64,98.55) .. (69.7,98.55) .. controls (78.76,98.55) and (86.1,118.32) .. (86.1,142.7) .. controls (86.1,167.08) and (78.76,186.85) .. (69.7,186.85) .. controls (60.64,186.85) and (53.3,167.08) .. (53.3,142.7) -- cycle ;
\draw  [color={rgb, 255:red, 208; green, 2; blue, 27 }  ,draw opacity=1 ][dash pattern={on 4.5pt off 4.5pt}] (193.1,141.5) .. controls (193.1,117.12) and (200.44,97.35) .. (209.5,97.35) .. controls (218.56,97.35) and (225.9,117.12) .. (225.9,141.5) .. controls (225.9,165.88) and (218.56,185.65) .. (209.5,185.65) .. controls (200.44,185.65) and (193.1,165.88) .. (193.1,141.5) -- cycle ;

\draw (36,93.4) node [anchor=north west][inner sep=0.75pt]    {$\textcolor[rgb]{0.29,0.56,0.89}{A}$};
\draw (236,90.4) node [anchor=north west][inner sep=0.75pt]    {$\textcolor[rgb]{0.82,0.01,0.11}{B}$};

\end{tikzpicture}
\end{center}
    \caption{We illustrate the graph matrix used as an example in Section~\ref{sec:failure-graphical}.}
    \label{fig:graph-matrix-example}
\end{figure}
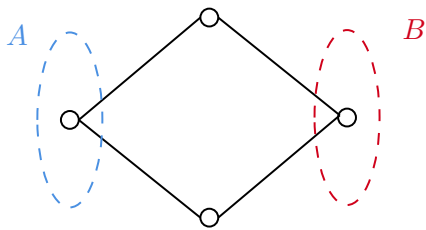

For any such $S$ and $x \neq y$, we have
\begin{align*}
    M_{xy}
    &= \sum_{a} S_{a,x} S_{a,y} \sum_{b \in [p] \setminus \{a\}} S_{b,x}S_{b,y} \\
    &= \sum_a S_{a,x}S_{a,y} ((S^2)_{x,y} - S_{a,x}S_{a,y}) \\
    &= (S^2)_{x,y}^2 - \sum_a S_{a,x}^2 S_{a,y}^2 \\
    &= (S^2)_{x,y}^2 - (p - 2). \numberthis
\end{align*}
That is, $M$ is the matrix $(S^2)^{\circ 2} - (p - 2)\boldsymbol{1}\boldsymbol{1}^{\top}$ with the diagonal zeroed out, where $\circ$ denotes the entrywise power of matrices.

When $S$ is the Seidel adjacency matrix of the Paley graph, we have $S^2 = p I - \boldsymbol{1}\boldsymbol{1}^{\top}$ by Proposition~\ref{prop:chi-sum-2}.
Thus in this case we have $M(S) = (p - 3)I -(p - 3)\boldsymbol{1}\boldsymbol{1}^{\top}$, and thus $\|M\| = (p - 1)(p - 3) \sim p^2$.

On the other hand, when $S$ is the Seidel adjacency matrix of a random ER graph on $p$ vertices, then the results of \cite{AMP-2016-GraphMatrices} show that, since the shape of $M$ has minimum vertex separator of size 1, with high probability, $\|M\| \leq \widetilde{O}(p^{3/2})$.
Thus, the Paley graph adjacency matrix fails to satisfy this basic graphical matrix bound.

We may understand this intuitively as follows: in the random case, $(S^2)^{\circ 2}$ is a random matrix whose entries have mean and standard deviation both on the scale of $p$.
Subtracting $p11^{\top}$ from this matrix centers the entries, so the bound of \cite{AMP-2016-GraphMatrices} says that the norm of $S$ behaves like a $p \times p$ random matrix with i.i.d.\ entries of size $p$.
However, for the Paley graph adjacency matrix, the exact quadratic equation satisfied by $S$ means that these fluctuations are no longer present, making the norm larger.

\section*{Acknowledgments}
\addcontentsline{toc}{section}{Acknowledgments}

We thank Afonso Bandeira, Chris Jones, and Daniel Spielman for helpful discussions, and the anonymous reviewers for their careful reading of the paper.

\bibliographystyle{alpha}
\addcontentsline{toc}{section}{References}
\bibliography{main}

\end{document}